\documentclass[twoside,11pt]{article}

\usepackage{jmlr2e}
\usepackage{amsmath}
\usepackage{graphicx, psfrag, epsf}
\usepackage{enumerate}
\usepackage{natbib}
\usepackage{url} 
\usepackage{bm}
\usepackage{setspace}
\usepackage{verbatim}

\DeclareMathOperator{\E}{\mathbb{E}}
\DeclareMathOperator{\V}{\mathbb{V}}

\DeclareMathOperator*{\argmax}{argmax} 
\DeclareMathOperator{\Length}{length}

\DeclareMathOperator{\N}{\mathrm{N}}
\DeclareMathOperator{\MN}{\mathrm{MN}}

\newtheorem{assumption}{Assumption}

\newcommand{\tr}{\mathrm{tr}}
\newcommand{\etr}{\mathrm{etr}}

\newcommand{\noop}[1]{}

\ShortHeadings{Generalized probabilistic principal component analysis}{Gu and Shen}
\firstpageno{1}

\begin{document}

\title{Generalized probabilistic principal component analysis of correlated data}

\author{\name Mengyang Gu \email mengyang@pstat.ucsb.edu \\
       \addr Department of  Statistics and Applied Probability\\
       University of California, Santa Barbara\\
       5511 South Hall\\
       Santa Barbara, CA 93106-3110
       \AND
       \name Weining Shen \email weinings@uci.edu \\
       \addr Department of Statistics \\
      University of California, Irvine \\
      2206 Bren Hall \\
      Irvine, CA 92697-1250}

\editor{}

\maketitle

\begin{abstract}

Principal component analysis (PCA) is a well-established tool in machine learning and data processing. The principal axes in PCA were shown to be equivalent to the maximum marginal likelihood estimator of the factor loading matrix in a latent factor model for the observed data, assuming that the latent factors are independently distributed as standard normal distributions. However, the independence assumption may be unrealistic for many scenarios such as modeling multiple time series, spatial processes, and functional data, where the outcomes are correlated.  In this paper, we introduce the generalized probabilistic principal component analysis (GPPCA) to study the latent factor model for multiple correlated outcomes, where each factor is modeled by a Gaussian process. Our method generalizes the previous probabilistic formulation of PCA (PPCA)  by
providing the closed-form maximum marginal likelihood estimator of the factor loadings and other parameters.  Based on the explicit expression of the precision matrix in the marginal likelihood that we derived, the number of the computational operations is linear to the number of output variables. Furthermore, we also provide the closed-form expression of the marginal likelihood when   other covariates are included in the mean structure. We highlight the advantage of GPPCA in terms of the practical relevance, estimation accuracy and computational convenience.  Numerical studies of simulated and real data confirm the excellent finite-sample performance of the proposed approach. 

%



\end{abstract}


\section{Introduction}

	\label{sec:Intro}
	Principal component analysis (PCA) is one of the oldest and most widely known approaches for dimension reduction. It has been used in many applications, including exploratory data analysis, regression, time series analysis, image processing, and functional data analysis. The most common solution of the PCA is to find a linear projection that transforms the set of original correlated variables onto a projected space of new uncorrelated variables by maximizing the variation of the projected space \citep{jolliffe2011principal}. This solution, despite its wide use in practice, lacks a probabilistic description of the data. 
	
	
	
	A probabilistic formulation of the PCA was first introduced by \cite{tipping1999probabilistic}, where the authors considered a Gaussian latent factor model, and then obtained the PCA (principal axes) as the solution of a maximum marginal likelihood problem, where the latent factors were marginalized out. This approach, known as the probabilistic principal component analysis (PPCA), assumes that the latent factors are independently distributed following a standard normal distribution. However, the independence assumption of the factors is usually too restrictive for many applications, where the variables of interest are correlated between different inputs, e.g. times series, images, and spatially correlated data. The latent factor model was extended to incorporate the dependent structure in  previous studies. For example, the linear model of coregionalization (LMC) was studied in modeling multivariate outputs of spatially correlated data  \citep{gelfand2004nonstationary,gelfand2010handbook}, where each factor is modeled by a Gaussian process (GP) to account for the spatial correlation in the data. When the factor loading matrix is shared, the LMC becomes a semiparameteric latent factor model, introduced in machine learning literature \citep{seeger2005semiparametric,alvarez2011kernels}, and was widely applied in emulating computationally expensive computer models with multivariate outputs \citep{higdon2008computer,fricker2013multivariate}, where each factor is modeled by a GP over a set of inputs, such as the physical parameters  of the partial differential equations. However, the PCA solution is no longer the maximum marginal likelihood estimator of the factor loading matrix when the factors at two inputs are correlated.
	
	
	
	
	
	 In this work, we propose a new approach called generalized probabilistic principal component analysis (GPPCA), as an extension of the PPCA  for the correlated output data. We assume each column of the factor loading matrix is orthonormal for the identifiability purpose. Based on this assumption, we  
	obtain a closed-form solution for the maximum marginal likelihood estimation of the factor loading matrix when the covariance function of the factor processes is shared. This result is an extension of the PPCA for the correlated factors, and the connection between these two approaches is studied. When the covariance functions of the factor processes are different, the maximum marginal likelihood estimation of the factor loading matrix is equivalent to an optimization problem with orthogonal constraints, sometimes referred as the Stiefel manifold. A fast numerical search algorithm on the Stiefel manifold is introduced by \cite{wen2013feasible} for the optimization problem.

    

	There are several approaches for estimating the factor loading matrix for the latent factor model and semiparameteric latent factor model in the Frequentist and Bayesian literature. One of the most popular approaches for estimating the factor loading matrix is  PCA (see e.g., \cite{bai2002determining,bai2003inferential,higdon2008computer}). Under the orthonormality assumption for the factor loading vectors, the PCA can be obtained from the maximum likelihood estimator of the factor loading matrix. However, the correlation structure of each factor is not incorporated for the estimation.  In \cite{lam2011estimation} and \cite{lam2012factor}, the authors considered estimating the factor loading matrix based on the sample covariance of the output data at the first several time lags when modeling high-dimensional time series. We will numerically compare our approach to the aforementioned Frequentist approaches. 

   Bayesian approaches have also been widely studied for factor loading matrix estimation.  \cite{west2003bayes7} points out the connection between PCA and a class of  generalized singular g-priors, and introduces a spike-and-slab prior that induces the sparse factors in the latent factor model assuming the factors are independently distributed. 
   When modeling spatially correlated data, priors are also discussed for the spatially varying factor loading matrices \citep{gelfand2004nonstationary} in LMC.
	The closed-form marginal likelihood obtained in this work is more computationally feasible than the previous results, as the inverse of the covariance matrix is shown to have an explicit form. 
	
	 
	 
	Our proposed method is also connected to the popular kernel approach, which has been used for nonlinear component analysis \citep{scholkopf1998nonlinear} by mapping the output data to a high-dimensional feature space through a kernel function. This method, known as the kernel PCA, is widely applied in various problems, such as the image analysis \citep{mika1999kernel}  and novelty detection \citep{hoffmann2007kernel}. However, the main focus of our method is to apply the kernel function for capturing the correlation of the outputs at different inputs (e.g. the time point, the location of image pixels or the physical parameters in the PDEs).

	
	We highlight a few contributions of this paper.  First of all, we derive the closed-form maximum marginal likelihood estimator (MMLE) of the factor loading matrix, when the factors are modeled by GPs. Note our expression of the marginal likelihood (after integrating out the factor processes) is more computationally feasible than the previous result, because the inverse of the covariance matrix is shown to have an explicit form, which makes the computational complexity linear to the number of output variables. Based on this closed-form marginal likelihood, we are able to obtain the MMLE of the other parameters, such as the variance of the noise and kernel parameters, and the predictive distribution of the outcomes. Our second contribution is that we provide a fully probabilistic analysis of the mean and other regression parameters, when some covariates are included in the mean structure of the factor model.  The empirical mean of data was often subtracted before applying PPCA and LMC \citep{tipping1999probabilistic,higdon2008computer}, which does not quantify the uncertainty when the output is linearly dependent on some covariates.  Here we manage to marginalize out the regression parameters in the mean structure explicitly without increasing the computational complexity. Our real data application examples demonstrate the improvements in out-of-sample prediction when the mean structure is incorporated in the data analysis. Lastly, the proposed  estimator of the factor loading matrix in GPPCA are closely connected to the PCA and PPCA, and we will discuss how the correlation in the factors affects the estimators of the factor loading matrix and predictive distributions. Both the simulated and real examples show the improved accuracy in estimation and prediction, when the output data are correlated.

	The rest of the paper is organized as follows. The main results of the closed-form marginal likelihood and the maximum marginal likelihood estimator of the factor loading matrix are introduced in Section \ref{subsec:GPPCA}. In Section \ref{subsec:par_est}, we provide the maximum marginal likelihood estimator for the noise parameter and kernel parameters, after marginalizing out the factor processes. Section \ref{subsec:mean_structure} discusses the estimators of the factor loading matrix and other parameters when some covariates are included in the model. The comparison between our approach and other approaches in estimating the factor loading matrix is studied in Section \ref{sec:comparison}, with a focus on the connection between GPPCA and PPCA. Simulation results are provided in Section \ref{sec:numerical_results}, for both the correctly specified and mis-specified models with unknown noise and covariance parameters. Two real data examples are shown in  Section \ref{sec:real_eg} and we conclude this work with discussion on several potential extensions in Section \ref{sec:conclusion}. 
	 
	
	\section{Main results}
		\label{sec:main_results}
   We state our main results in this section. In section \ref{subsec:GPPCA}, we derive a computationally feasible expression of the marginal distribution for the latent factor model after marginalizing out the factor processes, based on which we show the maximum marginal likelihood estimator of the factor loading matrix. In Section \ref{subsec:par_est}, we discuss the parameter estimation and predictive distribution. 
   We extend our method to study the factor model by allowing the intercept and additional covariates in the mean structure in Section \ref{subsec:mean_structure}.
   
   \subsection{Generalized probabilistic principal component analysis}
   \label{subsec:GPPCA}
      To begin with, let $\mathbf y(\mathbf x)= (y_1(\mathbf x),...,y_k(\mathbf x))^T$ be a $k$-dimensional real-valued output vector at a $p$-dimensional input vector $\mathbf x$.  Let $\mathbf Y=[\mathbf y(\mathbf x_1),..., \mathbf y(\mathbf x_n)]$ be a $k\times n$ matrix of the observations  at inputs $\{\mathbf x_1,...,\mathbf x_n\}$. In this subsection and the next subsection, we assume that each row of the $\mathbf Y$ is centered at zero. 
      


	Consider the following latent factor model 
	\begin{equation}
	\mathbf y(\mathbf x)= \mathbf A \mathbf z(\mathbf x) + \bm \epsilon,
	\label{equ:model_1}
	\end{equation}
	where   $\bm \epsilon \sim N(0, \sigma^2_0 \mathbf I_k)$ is a vector of  independent Gaussian noises, with $\mathbf I_k$ being the $k\times k$ identity matrix. The $k\times d$ factor loading matrix $\mathbf A=[\mathbf a_1,...,\mathbf a_d]$ relates the $k$-dimensional output to  $d$-dimensional factor processes $\mathbf z(\mathbf x)=(z_1(\mathbf x),...,z_d(\mathbf x))^T$, where $d \leq k$.

 In many applications, each output is correlated. For example, model (\ref{equ:model_1}) is widely used in analyzing multiple time series, where $y_l(x)$'s are correlated across different time points for every $l=1,...,k$. Model  (\ref{equ:model_1}) was also used for multivariate spatially correlated outputs, often referred as the linear model of coregionalization (LMC) \citep{gelfand2010handbook}.  In these studies,  each factor is modeled  by a zero-mean Gaussian process (GP), meaning that for any set of inputs $\{\mathbf x_1,...,\mathbf x_n \}$, $\mathbf Z_l=(z_l(\mathbf x_1),...,z_l(\mathbf x_n))$ follows a multivariate normal distribution  
	\begin{equation}
	\mathbf Z^T_l \sim \MN(\mathbf 0, \bm \Sigma_l),
	\end{equation} 
	where the $(i,\, j)$ entry of $\bm \Sigma_l$ is parameterized by a covariance function $\sigma^2_l K_l(\mathbf x_i, \mathbf x_j)$ for $l=1,...,d$ and $1\leq i,j\leq n$. We defer the discussion of the kernel in the Section~\ref{subsec:par_est}.

Note that the model (\ref{equ:model_1}) is unchanged if one replaces the pair $(\mathbf A, \mathbf z(\mathbf x))$ by $(\mathbf A \mathbf E, \mathbf E^{-1} \mathbf z(\mathbf x))$  for any invertible matrix $\mathbf E$. As pointed out in \cite{lam2011estimation}, only the $d$-dimensional linear subspace of $\mathbf A$, denoted as $\mathcal M(\mathbf A)$, can be uniquely identified, since $\mathcal M(\mathbf A)=\mathcal M(\mathbf A \mathbf E)$ for any invertible matrix $\mathbf E$. Due to this reason, we assume the columns of $\mathbf A$ in model (\ref{equ:model_1}) are orthonormal for identifiablity purpose \citep{lam2011estimation,lam2012factor}.
	 \begin{assumption}
	 \begin{equation}
	 \mathbf A^T \mathbf A=\mathbf I_d.
	 \label{equ:A}
	 \end{equation}
	 \label{assumption:A}
	 \end{assumption}
	Note the Assumption \ref{assumption:A} can be relaxed by assuming $ \mathbf A^T \mathbf A=c\mathbf I_d$ where $c$ is a positive constant which can potentially depend on $k$, e.g. $c=k$. As each factor process has the variance $\sigma^2_l$, typically estimated from the data, we thus derive the results based on Assumption \ref{assumption:A} herein. 
	
	Denote the vectorization of the output $\mathbf Y_v=\mathrm{vec}(\mathbf Y)$ and the $d\times n $ latent factor matrix $\mathbf Z=(\mathbf z(\mathbf x_1),...,\mathbf z(\mathbf x_n))$ at inputs $\{\mathbf x_1,...,\mathbf x_n\}$. We first give the marginal distribution of $\mathbf Y_v$ (after marginalizing out $\mathbf Z$) with an explicit inverse of the covariance matrix in  Lemma \ref{lemma:Y_v_inv_cov}.

\begin{lemma}
Under Assumption \ref{assumption:A}, the marginal distribution of $\mathbf Y_v$ in model (\ref{equ:model_1}) is the multivariate normal distribution as follows, 
\begin{align}
\mathbf Y_v \mid \mathbf A, \sigma^2_0, \bm \Sigma_{1},..., \bm \Sigma_{d} & \sim \MN\left(\mathbf 0, \,  \sum^d_{l=1} \bm \Sigma_l \otimes (\mathbf a_l \mathbf a^T_l)+ \sigma^2_0 \mathbf I_{nk} \right) 
\label{equ:Y_v_cov}  \\
 & \sim  \MN\left(\mathbf 0, \,   \sigma^2_0 \left(\mathbf I_{nk}- \sum^d_{l=1} (\sigma^2_0 \bm \Sigma_l^{-1}+\mathbf I_n )^{-1} \otimes (\mathbf a_l \mathbf a^T_l) \right)^{-1} \right).
\label{equ:Y_v_inv_cov}
\end{align} 
\label{lemma:Y_v_inv_cov}
\end{lemma}
The form in (\ref{equ:Y_v_cov}) appeared in the previous literature (e.g. \cite{gelfand2004nonstationary}) and its derivation is given in Appendix B. However, directly computing the marginal likelihood by expression (\ref{equ:Y_v_cov}) may be expensive, as the covariance matrix is $nk\times nk$.  Our expression (\ref{equ:Y_v_inv_cov}) of the marginal likelihood is  computationally more feasible than the expression (\ref{equ:Y_v_cov}), as the inverse of the covariance matrix of $\mathbf Y_v$ is derived explicitly in (\ref{equ:Y_v_inv_cov}). Based on the marginal likelihood in (\ref{equ:Y_v_inv_cov}), we derive  the maximum marginal estimation of $\mathbf A$ where the covariance matrix for each latent factor is assumed to be the same as in Theorem \ref{thm:est_A_shared_cov} below. 

 
 	
	\begin{theorem}
	For model (\ref{equ:model_1}), assume $\bm \Sigma_1=...=\bm \Sigma_d=\bm \Sigma$. Under Assumption \ref{assumption:A}, after marginalizing out $\mathbf Z$, the  likelihood function is maximized when  
	\begin{equation}
	 \hat {\mathbf A}=\mathbf U \mathbf R,
	\label{equ:A_est_shared_cov}
	\end{equation}
	 where $\mathbf U$ is a $k \times d$ matrix of the first $d$ principal eigenvectors of $ 	\mathbf G={\mathbf Y (\sigma^2_0 \bm \Sigma^{-1}+  \mathbf I_n )^{-1}  \mathbf Y^T}$, 
	and $\mathbf R$ is an arbitrary $d \times d$ orthogonal rotation matrix. 
	\label{thm:est_A_shared_cov}
	\end{theorem}

By Theorem~\ref{thm:est_A_shared_cov}, the solution $\mathbf {\hat A}$ is not unique because of the arbitrary rotation matrix. However, the linear subspace of the column space of the estimated factor loading matrix, denoted by $\mathcal M(\mathbf {\hat A})$, is uniquely determined by (\ref{equ:A_est_shared_cov}).

In general, the covariance function of each factor can be different. We are able to express the maximum marginal likelihood estimator as the solution to an optimization problem with the orthogonal constraints, stated in Theorem \ref{thm:est_A_diff_cov}. 

	\begin{theorem}
     Under Assumption \ref{assumption:A}, after marginalizing out $\mathbf Z$, the maximum marginal likelihood estimator of  $\mathbf A$ in model (\ref{equ:model_1}) is  
     	\begin{equation}
\mathbf {\hat A}= \argmax_{\mathbf A} \sum^d_{l=1}  {\mathbf a^T_l \mathbf G_l \mathbf a_l},  \quad \text{s.t.} \quad \mathbf A^T \mathbf A=\mathbf I_d,
	\label{equ:A_est_diff_cov}
	\end{equation}
	where $\mathbf G_l= { \mathbf Y  (\sigma^2_0 \bm \Sigma^{-1}_l+\mathbf I_n )^{-1}\mathbf Y^T}$.

	\label{thm:est_A_diff_cov}
	\end{theorem}
	The subset of matrices $\mathbf A$ that satisfies the orthogonal constraint $\mathbf A^T \mathbf A=\mathbf I_d$ is often referred as the \textit{Stiefel manifold}. Unlike the case where the covariance of each factor processes is shared, no closed-form solution of the optimization problem in (\ref{thm:est_A_diff_cov}) has been found.  A numerical optimization algorithm that preserves the orthogonal constraints  in  (\ref{equ:A_est_diff_cov}) is introduced in \cite{wen2013feasible}.  The main idea of their algorithm is to find the gradient of the objective function in the tangent space at the current step, and iterates by a curve along the  projected negative descent on the manifold. The curvilinear search is applied to find the appropriate step size that guarantees the convergence to a stationary point. We implement this approach to numerically optimize the marginal likelihood to obtain the estimated factor loading matrix in Theorem \ref{thm:est_A_diff_cov}.


	We call the method of estimating $\mathbf A$ in Theorem \ref{thm:est_A_shared_cov} and Theorem \ref{thm:est_A_diff_cov} the generalized probabilistic principal component analysis (GPPCA) of correlated data, which is a direct extension of the PPCA in \cite{tipping1999probabilistic}. Although both approaches obtain the maximum marginal likelihood estimator of the factor loading matrix, after integrating out the latent factors, the key difference is that in GPPCA, the latent factors at different inputs are allowed to be correlated, whereas the latent factors in PPCA are assumed to be independent.  A detailed numerical comparison between our method and other approaches including the PPCA will be given in Section~\ref{sec:comparison}. 
	
	Another nice feature of the proposed GPPCA method is that 
	the estimation of the factor loading matrix can be applied to any covariance structure of the factor processes.  In this paper, we use kernels to parameterize the covariance matrix as an illustrative example. There are many other ways to specify the covariance matrix or the inverse of the covariance matrix, such as the Markov random field and the dynamic linear model, and these approaches are readily applicable in our latent factor model (\ref{equ:model_1}). 
	
	
	 
	
	



For a function with a $p$-dimensional input, we use a product kernel to model the covariance for demonstration purposes \citep{sacks1989design}, meaning that for the $l$th factor,
\begin{equation}
\sigma^2_l K_l(\mathbf x_a, \mathbf x_b)=\sigma^2_l \prod^p_{m=1} K_{lm}( x_{am}, x_{bm}),  
\label{equ:K_l}
\end{equation}
 for any input $\mathbf x_a=(x_{a1},...,x_{ap})$ and $\mathbf x_a=(x_{b1},...,x_{bp})$, where $K_{lm}(\cdot, \cdot)$ is a one-dimensional kernel function of the $l$th factor that models the correlation of the $m$th coordinate of any two inputs.  
 

Some widely used one-dimensional kernel functions include the power exponential kernel and the Mat{\'e}rn kernel. For any two inputs $\mathbf x_a, \mathbf x_b\in \mathcal X$, the Mat{\'e}rn kernel  is 
 
 \begin{equation}
K_{lm}(x_{am}, x_{bm})=\frac{1}{2^{\nu_{lm}-1}\Gamma(\nu_{lm})}\left(\frac{|x_{am}- x_{bm}|}{\gamma_{lm}} \right)^{\nu_{lm}} \mathcal K_{\nu_{lm}} \left(\frac{|x_{am}- x_{bm}|}{\gamma_{lm}} \right),
\label{equ:matern}
\end{equation}
where $\Gamma(\cdot)$ is the gamma function and $\mathcal{K}_{\nu_{lm}}(\cdot)$ is the modified Bessel function of the second kind with a positive roughness parameter $\nu_{lm}$ and a nonnegative range parameter $\gamma_{lm}$ for $l=1,...,d$ and $m=1,...,p$. The Mat{\'e}rn kernel contains a wide range of different kernel functions. In particular, when $\nu_{lm}=1/2$, the Mat{\'e}rn kernel becomes the exponential kernel, $K_l(x_{am}, x_{bm})=\exp(-|x_{am}- x_{bm}|/\gamma_{lm})$, and the corresponding factor process is the Ornstein-Uhlenbeck process, which is a continuous autoregressive process with order 1. When $\nu_{lm} \to \infty$, the Mat{\'e}rn kernel becomes the Gaussian kernel, i.e., $K_l(x_{am}, x_{bm})=\exp(-|x_{am}- x_{bm}|^2/\gamma^2_{lm})$, where the factor process is infinitely differentiable.    The Mat{\'e}rn kernel has a closed-form expression when $(2\nu_{lm}+1)/2 \in \mathbb N$. For example, the Mat{\'e}rn kernel with $\nu_{lm}=5/2$ has the following form
 \begin{equation}
K_{lm}(x_{am}, x_{bm})=\left(1+\frac{\sqrt{5}|x_{am}-x_{bm}|}{\gamma_{lm}}+\frac{5|x_{am}-x_{bm}|^2}{3\gamma_{lm}^2}\right)\exp\left(-\frac{\sqrt{5}|x_{am}-x_{bm}|}{\gamma_{lm}}\right) \,,
\label{equ:matern_5_2}
\end{equation}
for any inputs $\mathbf x_a$ and $\mathbf x_b$ with $l=1,...,d$ and $m=1,...,p$. In this work, we use the Mat{\'e}rn kernel in (\ref{equ:matern_5_2}) for the simulation and real data analysis for demonstration purposes. Specifying a sensible kernel function depends on real applications and our results in this work apply to all commonly used kernel functions. We will also numerically compare different approaches when the kernel function is misspecified in Appendix C.

\subsection{Parameter estimation and predictive distribution}
\label{subsec:par_est}

The probabilistic estimation of the factor loading matrix depends on the variance of the noise and the covariances of the factor processes. We discuss the estimation of these parameters by assuming that the covariances of the factors are parameterized by a product of the kernel functions for demonstration purposes. We also obtain the predictive distribution of the data in this subsection. The probabilistic estimation of the factor loading matrix in the GPPCA can be also applied when the covariances of the factors are specified or estimated in other ways.  






We denote  $\tau_l:=\frac{\sigma^2_l}{\sigma^2_0}$ as the signal's variance to noise ratio (SNR) for the $l$th factor process, as a transformation of $\sigma^2_{l}$ in  (\ref{equ:K_l}). The maximum likelihood estimator of $\sigma^2_0$ has a  closed form expression using this parameterization. Furthermore, let the correlation matrix of the $k$th factor process be $\mathbf K_l$ with the $(i, j)$th term being $K_l(\mathbf x_i, \mathbf x_j)$. After this transformation, the estimator of $\mathbf A$ in Theorems \ref{thm:est_A_shared_cov} and   \ref{thm:est_A_diff_cov} becomes a function of the parameters $\bm \tau=(\tau_1,...,\tau_d)$ and $\bm \gamma=(\bm \gamma_1,...,\bm \gamma_d)$.  Under Assumption \ref{assumption:A}, after marginalizing out $\mathbf Z$, the maximum likelihood estimator of $\sigma^2_0$ becomes a function of $\mathbf A$, $\bm \tau$ and $\bm \gamma$ as  
\begin{align}
\hat \sigma^2_0&= \frac{\hat S^2}{nk},
\label{equ:hat_sigma_2_0}
\end{align} 
where $\hat S^2 =\tr(  \mathbf {Y}^T \mathbf {Y}    )-\sum^d_{l=1} \mathbf {a}^T_l \mathbf {Y}   (\tau^{-1}_l \mathbf  {K}^{-1}_l +\mathbf I_n )^{-1} \mathbf {Y} ^T  \mathbf {a}_l$. Ignoring the constants, the likelihood of $\bm \tau$ and $\bm \gamma$ by plugging $\mathbf {\hat A}$  and $\hat \sigma^2_0$ satisfies 
\begin{equation}
L(\bm \tau, \bm \gamma \mid \mathbf Y, \mathbf {\hat A}, \hat \sigma^2_0 )\propto \left\{ \prod^d_{l=1} | \tau_l \mathbf  {K}_l +\mathbf I_n |^{-1/2} \right\}|\hat S^2|^{-nk/2}.
\label{equ:profile_lik}
\end{equation}

A derivation of Equation (\ref{equ:profile_lik}) is given in the Appendix. Since there is no closed-form expression for the  parameter estimates in the kernels, one often numerically maximizes the Equation (\ref{equ:profile_lik})  to estimate these parameters 
\begin{equation}
( \bm {\hat \tau},  \bm {\hat \gamma}):=\argmax_{(\bm \tau, \bm \gamma)} L(\bm \tau, \bm \gamma \mid \mathbf Y, \mathbf {\hat A},  \hat \sigma^2_0 ).
\label{equ:tau_gamma_est}
\end{equation}
After obtaining $\hat \sigma^2_0$ and $\bm {\hat \tau}$ from (\ref{equ:hat_sigma_2_0}) and  (\ref{equ:tau_gamma_est}), respectively, we transform the expressions back to get the estimator of $\sigma^2_l$ as
\[\hat \sigma^2_l= \hat \tau_l\hat \sigma^2_0,\]
for $l=1,...,d$. Since both the estimator of $\hat \sigma^2_0$ and $\mathbf {\hat A}$ in Theorem \ref{thm:est_A_shared_cov} and  \ref{thm:est_A_diff_cov}  can be expressed as a function of $(\bm \tau, \bm \gamma)$, in each iteration, one can use the Newton's method \citep{nocedal1980updating} to find $(\bm \tau, \bm \gamma)$ based on the likelihood  in (\ref{equ:profile_lik}), after plugging the estimator of $\hat \sigma^2_0$ and $\mathbf {\hat A}$.

We have a few remarks regarding the expressions in (\ref{equ:hat_sigma_2_0}) and (\ref{equ:tau_gamma_est}). First, under Assumption \ref{assumption:A}, the likelihood of $(\bm \tau, \bm \gamma)$ in (\ref{equ:profile_lik}) can also be obtained by marginalizing out $\sigma^2_0$ using the objective prior $\pi(\sigma^2_0)\propto 1/\sigma^2_0$, instead of maximizing over $\sigma^2_0$. 

Second, consider the first term at the right hand side of (\ref{equ:hat_sigma_2_0}). As each row of $\mathbf Y$ has a zero mean, let $\mathbf S_0:={\mathbf Y\mathbf Y^T}/{n}={\sum^n_{i=1}\mathbf y(\mathbf x_i) \mathbf y(\mathbf x_i)^T}/{n},$ be the sample covariance matrix for $\mathbf y(\mathbf x_i)$. One has $\tr(\mathbf Y \mathbf Y^T)=n\sum^k_{i=1}\lambda_{0i}$, where $\lambda_{0i}$ is the $i$th eigenvalue of $\mathbf S_0$. The second term at the right hand side of (\ref{equ:hat_sigma_2_0}) is the variance explained by the projection. In particular, when the conditions in Theorem~\ref{thm:est_A_shared_cov} hold, i.e. $\bm \Sigma_1=...=\bm \Sigma_d$, one has   $\sum^d_{l=1} \mathbf {\hat a}^T_l \mathbf {Y}   (\tau^{-1}_l \mathbf  {K}^{-1}_l +\mathbf I_n )^{-1} \mathbf {Y} ^T  \mathbf {\hat a}_l =n\sum^d_{l=1}\hat \lambda_l$, where $\hat \lambda_l$ is the $l$th largest eigenvalues of $\mathbf Y(\sigma^2_0\bm \Sigma^{-1}+\mathbf I_n )^{-1} \mathbf Y^T/n$. The estimation of the noise is then the average variance being lost in the projection. Note that the projection in the GPPCA takes into account the correlation of the factor processes, whereas the projection in the PPCA assumes the independent factors. This difference makes the GPPCA more accurate in estimating the subspace of the factor loading matrix when the factors are correlated, as shown in various numerical examples in Section \ref{sec:numerical_results}.

Thirdly, although the model in (\ref{equ:model_1}) is regarded as a nonseparable model \citep{fricker2013multivariate}, the computational complexity of our algorithm is the same with that for the separable model \citep{Gu2016PPGaSP,conti2010bayesian}. Instead of inverting an $nk \times nk$  covariance matrix, the expression of the  likelihood in  (\ref{equ:profile_lik}) allows us to proceed in the same way when the covariance matrix for each factor has a size of $n\times n$. The number of  computational operations  of the likelihood is at most $\max(O(dn^3), O(kn^2))$, which is much smaller than the $O(n^3k^3)$ for inverting an $nk \times nk$  covariance matrix, because one often has $d \ll k$. When the input is one-dimensional and the Mat{\'e}rn kernel in (\ref{equ:matern}) is used, the computational operations are only $O(dkn)$ for computing the likelihood in (\ref{equ:profile_lik}) without any approximation (see e.g. \cite{hartikainen2010kalman}). We implement this algorithm in the ${\tt FastGaSP}$  package available  on CRAN.

Note that the estimator in (\ref{equ:profile_lik}) is known as the Type II maximum likelihood estimator, which is widely used in estimating the kernel parameters. When the number of the observations is small, the  estimator in (\ref{equ:profile_lik}) is not robust, in the sense that the estimated range parameters can be very small or very large, which makes the covariance matrix either a diagonal matrix or a singular matrix. This might be unsatisfactory in certain applications, such as emulating computationally expensive computer models \citep{oakley1999bayesian}. An alternative way is to use the maximum marginal posterior estimation that prevents the two unsatisfying scenarios of the estimated covariance matrix. We refer to \cite{Gu2018robustness} and \cite{gu2018jointly} for the theoretical properties of the maximum marginal posterior estimation and an ${\sf R}$ package is available on CRAN (\cite{gu2018robustgasp}). 


Given the parameter estimates, we can also obtain the predictive distribution for the outputs. 
Let $\hat K_l(\cdot, \cdot)$ be the $l$th kernel function after plugging the estimates $\bm{\hat \gamma}_l$ and let $\bm {\hat \Sigma}_l$ be the estimator of the covariance matrix for the $l$th factor, where the $(i,j)$ element of $\bm {\hat \Sigma}_l$ is $ \hat \sigma^2_l \hat K_l(\mathbf x_i, \mathbf x_j)$, with $1\leq i,j\leq n$ and $l=1,...,d$. We have the following predictive distribution for the output at any given input. 

\begin{theorem}
Under the Assumption \ref{assumption:A}, for any $\mathbf x^*$, one has

 \[\mathbf Y(\mathbf x^*) \mid \mathbf Y, \mathbf {\hat A},  \bm {\hat \gamma}, \bm {\hat \sigma}^2, { \hat \sigma}^2_0 \sim \MN \left(\bm {\hat \mu}^*(\mathbf x^*), \bm {\hat \Sigma}^*(\mathbf x^*) \right),\]
 where
 \begin{equation}
 \bm {\hat \mu}^*(\mathbf x^*)=\mathbf {\hat A}  \mathbf {\hat z}(\mathbf x^*),
 \label{equ:hat_mu}
 \end{equation} 
 with $\mathbf {\hat z}(\mathbf x^*)=( {\hat z}_1(\mathbf x^*),...,{\hat z}_d(\mathbf x^*) )^T $, with ${\hat z}_l(\mathbf x^*)=\bm {\hat \Sigma}^T_l(\mathbf x^*) ({\bm {\hat \Sigma}_l+\hat \sigma^2_0 \mathbf I_n })^{-1}\mathbf Y^T \mathbf {\hat a}_l$, $\bm {\hat \Sigma}_l(\mathbf x^*)=\hat \sigma^2_l(\hat K_l(\mathbf x_1, \mathbf x^*),...,\hat K_l(\mathbf x_n, \mathbf x^*))^T$ for $l=1,...,d$,
 and 
  \begin{equation}
\bm {\hat \Sigma}^*(\mathbf x^*)=  \mathbf {\hat A} \mathbf {\hat D}(\mathbf x^*) \mathbf {\hat A}^T+ \hat \sigma^2_0(\mathbf I_k - \mathbf {\hat A} \mathbf {\hat A}^T),
 \label{equ:hat_Sigma}
 \end{equation} 
 with $\mathbf {\hat D}(\mathbf x^*) $ being a diagonal matrix, and its $l$th diagonal term, denoted as ${\hat D}_l(\mathbf x^*) $, has the following expression
\[{\hat D}_l(\mathbf x^*)=  \hat\sigma^2_l \hat K_l(\mathbf x^*,\, \mathbf x^*) + \hat \sigma^2_0 -  \bm {\hat \Sigma}^T_l(\mathbf x^*) \left({\bm {\hat \Sigma}_l+\hat \sigma^2_0 \mathbf I_n }\right)^{-1} \bm {\hat \Sigma}_l(\mathbf x^*),    \]
  for $l=1,...,d$.
\label{thm:prediction}
\end{theorem} 



 





Next we give the posterior distribution of $\mathbf A \mathbf Z$ in Corollary \ref{cor:posterior_AZ}.



\begin{corollary}[Posterior distribution of $\mathbf A \mathbf Z$]
Under the Assumption (\ref{assumption:A}), the posterior distribution of  $ \mathbf A \mathbf Z$ is
\[ (\mathbf A \mathbf Z \mid \mathbf Y, \mathbf {\hat A},  \bm {\hat \gamma}, \bm {\hat \sigma}^2, { \hat \sigma}^2_0)  \sim \MN\left( \mathbf {\hat A} \mathbf {\hat Z}, \hat \sigma^2_0\sum^d_{l=1} \mathbf {\hat D}_l \otimes \mathbf {\hat a}_l \mathbf {\hat a}^T_l\right),\]
where 
 $ \mathbf {\hat Z}=(\mathbf {\hat Z}^T_1,...,\mathbf {\hat Z}^T_d)^T $, $\mathbf {\hat Z}^T_l= \bm {\hat \Sigma}_l ( \bm {\hat \Sigma}_l+ \hat \sigma^2_0 \mathbf I_n)^{-1} \mathbf Y^T\mathbf {\hat a}_l$,  and $\mathbf {\hat D}_l=\left(\sigma^2_0 \bm {\hat \Sigma}^{-1}_l+{\mathbf I_n}\right)^{-1} $,
 for $l=1,...,d$.


\label{cor:posterior_AZ}
\end{corollary}


The Corollary \ref{cor:posterior_AZ}  is a direct consequence of Theorem \ref{thm:prediction}, so the proof is omitted. Note that the uncertainty of the parameters and the factor loading matrix are not taken into consideration for predictive distribution of $\mathbf Y(\mathbf x^*)$ in Theorem \ref{thm:prediction} and the posterior distribution of $\mathbf A \mathbf Z$ in Corollary \ref{cor:posterior_AZ}, because of the use of the plug-in estimator for $(\mathbf A,\sigma^2_0, \bm \sigma^2, \bm \gamma)$. The resulting posterior credible interval may be narrower than it should be when the sample size is small to moderate. The uncertainty in $\mathbf A$ and other model parameters could be obtained by Bayesian analysis with a prior placed on these parameters for these scenarios.  


 \subsection{Mean structure}
 \label{subsec:mean_structure}
 In many applications, the outputs are not centered at zero. For instance, \cite{Bayarri09} and \cite{Gu2016PPGaSP} studied emulating the height of the pyroclastic flow generated from TITAN2D computer model, where the flow volume in the chamber is positively correlated to height of the flow at each spatial coordinate. Thus, modeling the flow volume as a covariate in the mean function typically improves the accuracy of the emulator. When $\mathbf Y$ is not centered around zero, one often subtracts the mean of each row of $\mathbf Y$ before the inference  \citep{higdon2008computer,paulo2012calibration}. The full Bayesian analysis of the regression parameters are discussed in coregionalization models of multivariate spatially correlated data (see e.g. \cite{gelfand2004nonstationary}) using the Markov Chain Monte Carlo (MCMC) algorithm, but the computation may be too complex to implement in many studies. 
 
 
 
 Consider the latent factor model with a mean structure for a $k$-dimensional output vector at the input $\mathbf x$, 
	\begin{equation}
	\mathbf y(\mathbf x)= (\mathbf h(\mathbf x) \mathbf B)^T+ \mathbf A \mathbf z(\mathbf x) + \bm \epsilon,
	\label{equ:model_with_mean}
	\end{equation}
where $\mathbf h(\mathbf x):=( h_1(\mathbf x),..., h_q(\mathbf x) )$ is $1\times q$ known mean basis function related to input $\mathbf x$ and possibly other covariates, $\mathbf B=(\bm \beta_1,...,\bm \beta_k)$ is a $q \times k$ matrix of the regression parameters. The regression parameters could be different for each row of the outcomes, and $\bm \epsilon \sim N(0, \sigma^2_0 \mathbf I_k)$ is a vector of the independent Gaussian noises, with $\mathbf I_k$ being the $k\times k$ identity matrix.


  For any set of inputs $\{\mathbf x_1,...,\mathbf x_n \}$, we assume $\mathbf Z_l=(z_l(\mathbf x_1),...,z_l(\mathbf x_n))$ follows a multivariate normal distribution  
	\begin{equation}
	\mathbf Z^T_l \sim \MN(\mathbf 0, \bm \Sigma_l),
	\end{equation} d
	where the $(i,\, j)$ entry of $\bm \Sigma_l$ is parameterized by  $K_l(\mathbf x_i, \mathbf x_j)$ for $l=1,...,d$ and $1\leq i,j\leq n$.


Denote $\mathbf H$ the $n\times q$ matrix with $(i,j)$th term being $h_j(\mathbf x_i)$ for $1\leq i\leq n$ and  $q<n$. We let $n>q$ and assume $\mathbf H$ is a full rank matrix. Further denote $\mathbf M= \mathbf I- \mathbf H (\mathbf H^T\mathbf H)^{-1}\mathbf H^T$. We apply a Bayesian approach for the regression parameters by assuming the objective  prior  $\pi(\mathbf B)\propto 1$ \cite{berger2001objective,berger2009formal}. We first marginalize out $\mathbf B$ and then marginalize out $\mathbf Z$ to obtain the marginal likelihood for estimating the other parameters .

\begin{lemma}
Let the prior of the regression parameters be $\pi(\mathbf B)\propto 1$. Under Assumption  \ref{assumption:A}, after marginalizing out $\mathbf B$ and $\mathbf Z$, the maximum likelihood estimator for $\sigma^2_0$ is 
\begin{equation}
\hat \sigma^2_0=\frac{S^2_M}{k(n-q)},
\label{equ:sigma_2_0_with_mean}
\end{equation}
where $S^2_M= \tr(\mathbf Y  {\mathbf M} \mathbf Y^T)-\sum^d_{l=1} \mathbf a^T_l \mathbf Y  {\mathbf M}  ({\mathbf M}+ \tau^{-1}_l \mathbf K^{-1}_l )^{-1}  {\mathbf M} \mathbf Y^T \mathbf a_l$. Moreover, the marginal density of the data satisfies
 \begin{align}
 p( \mathbf Y \mid \mathbf A,  \bm { \tau},  \bm { \gamma},  { \hat \sigma}^2_0 )  &\propto \left\{\prod^d\limits_{l=1} \left|\tau_l \mathbf K_l +\mathbf I_n \right|^{-1/2} \left|  \mathbf H^T (\tau_l \mathbf K_l+\mathbf I_n)^{-1} \mathbf H \right|^{-\frac{1}{2}}\right\}  \left| S^2_M \right|^{- \left(\frac{k(n-q)}{2}\right) }.
 \label{equ:marginal_with_mean}  
  \end{align}  
%
%
%
\label{lemma:Y_mean_structure}
\end{lemma}

\begin{remark}
Under Assumption \ref{assumption:A}, the likelihood for $(\bm \tau, \bm \gamma)$ in (\ref{equ:marginal_with_mean}) are equivalent to the maximum marginal likelihood estimator by marginalizing out both $\mathbf B$ and $\sigma^2_0$ using the objective prior $\pi(\mathbf B, \sigma^2_0)\propto 1/\sigma^2_0$, instead of maximizing over $\sigma^2_0$. 
\end{remark}

  Since there is no closed-form expression for the parameters $(\bm \tau, \bm \gamma)$ in the kernels, one can numerically maximize the Equation (\ref{equ:marginal_with_mean})  to estimate $\mathbf A$ and other   parameters.  
  
\begin{align}
\mathbf {\hat A}&= \argmax_{\mathbf A} \sum^d_{l=1} \mathbf a^T_l  \mathbf G_{l,M}  \mathbf a_l,  \quad \text{s.t.} \quad \mathbf A^T \mathbf A=\mathbf I_d, \label{equ:est_A_with_mean} \\
( \bm {\hat \tau},  \bm {\hat \gamma})&=\argmax_{(\bm \tau, \bm \gamma)} p( \mathbf Y \mid \mathbf {\hat A},   \bm { \tau},  \bm { \gamma} ).
\label{equ:est_tau_gamma_with_mean}
\end{align}

When $ {\bm \Sigma}_1=...= {\bm \Sigma}_d$, the closed-form expression of $\mathbf {\hat A}$ can be obtained similarly in Theorem \ref{thm:est_A_shared_cov}. In general, we can use the approach in \cite{wen2013feasible} for solving the optimization problem in (\ref{equ:est_A_with_mean}). After obtaining $\bm {\hat \tau}$ and $\hat \sigma^2_0$, we transform them to get $\hat {\sigma}^2_l=\hat \tau_l \hat \sigma^2_0$ for $l=1,...,d$.

Let $\bm {\hat \Sigma}_l$ be a matrix with the $(i,j)$-term as $\hat \sigma^2_l \hat K_l(\mathbf x_i, \mathbf x_j)$, where $\hat K_l(\mathbf x_i, \mathbf x_j)$ is the kernel function after plugging the estimator $ {\hat \gamma}_l$ for $1\leq l\leq d$. We first marginalize out $\mathbf B$ and then marginalize out $\mathbf Z$. The rest of the parameters are estimated by the maximum marginal likelihood estimator by (\ref{equ:sigma_2_0_with_mean}), (\ref{equ:est_A_with_mean}) and (\ref{equ:est_tau_gamma_with_mean}) in the  predictive distribution given below.


\begin{theorem}
Under the Assumption \ref{assumption:A} and assume the objective prior $\pi(\mathbf B)\propto 1$.  After  marginalizing out $\mathbf B$, $\mathbf Z$, and plugging in the maximum marginal likelihood estimator of $(\mathbf {A}, \bm {\gamma},\bm \sigma^2, \sigma^2_0)$, the predictive distribution of model (\ref{equ:model_with_mean}) for any $\mathbf x^*$ is 
 \[\mathbf Y(\mathbf x^*) \mid \mathbf Y, \mathbf {\hat A},  \bm {\hat \gamma}, \bm \hat {\bm \sigma}^2, { \hat \sigma}^2_0 \sim \MN \left(\bm {\hat \mu}_M^*(\mathbf x^*), \bm {\hat \Sigma}_M^*(\mathbf x^*) \right).\]
 Here
 \begin{align}
 \bm {\hat \mu}_M^*(\mathbf x^*)&= 
( \mathbf h(\mathbf x^*){ \hat {\mathbf B} })^T  +  \mathbf {\hat A}  \mathbf {\hat z}_M(\mathbf x^*),
 \label{equ:hat_mu} \\
 \bm {\hat \Sigma}^*_M(\mathbf x^*)&=  \mathbf {\hat A} \mathbf {\hat D}_M(\mathbf x^*) \mathbf {\hat A}^T +\hat \sigma^2_0(1+ \mathbf h(\mathbf x^*)(\mathbf H^T\mathbf H)^{-1}\mathbf h^T(\mathbf x^*))  (\mathbf I_k-\mathbf {\hat A} \mathbf  {\hat A}^T),
 \label{equ:hat_Sigma}
 \end{align} 
 %
 where $\mathbf{\hat B}=(\mathbf H^T \mathbf H)^{-1}\mathbf{H}^T (\mathbf Y -\mathbf {\hat A} \mathbf {\hat Z}_M)^T  $, $\mathbf {\hat Z}_M= (\mathbf {\hat Z}^T_{1,M},...,\mathbf {\hat Z}^T_{d,M})^T$ with $\mathbf {\hat Z}_{l,M}=\mathbf a^T_l \mathbf Y \mathbf M (  \hat{\bm \Sigma}_l \mathbf M+ \hat \sigma^2_0 \mathbf I_n )^{-1} \hat {\bm \Sigma}_l  $,  $\mathbf {\hat z}_M(\mathbf x^*)=( {\hat z}_{1,M}(\mathbf x^*),...,{\hat z}_{d,M}(\mathbf x^*) )^T$ with  $ {\hat z}_{l,M}(\mathbf x^*)=\hat{\bm \Sigma}^T_l(\mathbf x^*)  ( \hat{\bm \Sigma}_l\mathbf M+ \hat \sigma^2_0 \mathbf I_n )^{-1} \mathbf M \mathbf Y  \mathbf a_l$,  for $l=1,...,d$, and $\mathbf {\hat D}_M(\mathbf x^*) $ is a diagonal matrix with the $l$th term:
\begin{align*}
{\hat D}_{l,M}(\mathbf x^*)&=\hat \sigma^2_l \hat K_l(\mathbf x^*,\, \mathbf x^*) +\hat \sigma^2_0-  \bm {\hat \Sigma}^T_l(\mathbf x^*) \tilde {\bm \Sigma}_l^{-1} \bm {\hat \Sigma}_l(\mathbf x^*) \\
&\quad \quad +(\mathbf h^T(\mathbf x^*)- \mathbf H^T  \tilde {\bm \Sigma}^{-1}_l  \mathbf {\hat \Sigma}_l(\mathbf x^*) )^T (\mathbf H^{T}  \tilde {\bm \Sigma}^{-1}_l   \mathbf H  )^{-1} (\mathbf h^T(\mathbf x^*)- \mathbf H^T  \tilde {\bm \Sigma}^{-1}_l  \mathbf {\hat \Sigma}_l(\mathbf x^*) ), 
\end{align*}
with $\tilde {\bm \Sigma}_l={\bm {\hat \Sigma}_l+\hat \sigma^2_0 \mathbf I_n }$ for $l=1,...,d$.
\label{thm:prediction_with_mean}
\end{theorem} 

In Theorem \ref{thm:prediction_with_mean}, the estimated mean parameters are $\mathbf {\hat B}=\E[\mathbf B \mid  \mathbf Y, \mathbf {\hat A},  \bm {\hat \gamma}, \bm \hat {\bm \sigma}^2, { \hat \sigma}^2_0]$, which could be used for inferring the trend of some given covariates (e.g. the gridded temperature example in Section \ref{subsec:gridded_temperature}).     

Denote $\mathbf Y(\mathbf x^*)=(\mathbf Y^T_1(\mathbf x^*), \,  \mathbf Y^T_2(\mathbf x^*))^T$ where $\mathbf Y_1(\mathbf x^*)$  and $\mathbf Y_2(\mathbf x^*)$ are two vectors of dimensions $k_1$ and $k_2$ ($k_1+k_2=k$), respectively. Assuming the same conditions in Theorem \ref{thm:prediction_with_mean} hold, if one observes both $\mathbf Y_1(\mathbf x^*)$ and $\mathbf Y$, the predictive distribution of $\mathbf Y_2(\mathbf x^*)$ follows 
\begin{equation}
\mathbf Y_2(\mathbf x^*) \mid \mathbf Y_1(\mathbf x^*),  \mathbf Y, \mathbf {\hat A},  \bm {\hat \gamma},  \bm {\hat \sigma}^2, { \hat \sigma}^2_0 \sim \MN \left(\bm {\hat \mu}_{M,2|1}^*(\mathbf x^*), \bm {\hat \Sigma}_{M,2|1}^*(\mathbf x^*) \right).
\label{equ:cond_pred_with_mean}
\end{equation}
where $\bm {\hat \mu}_{M,2|1}^*(\mathbf x^*)= \bm {\hat \mu}_{M,2}^*(\mathbf x^*)+\bm {\hat \Sigma}_{M,12}^*(\mathbf x^*)^T\bm {\hat \Sigma}_{M,11}^*(\mathbf x^*)^{-1}(\mathbf Y_1(\mathbf x^*)- \bm {\hat \mu}_{M,1}^*(\mathbf x^*)) $ with ${\hat \mu}_{M,1}^*(\mathbf x^*)$ and ${\hat \mu}_{M,2}^*(\mathbf x^*)$ being the first $k_1$ and last $k_2$ entries of ${\hat \mu}_{M}^*(\mathbf x^*)$; $\bm {\hat \Sigma}_{M,2|1}^*(\mathbf x^*)= \bm {\hat \Sigma}_{M,22}^*(\mathbf x^*)-\bm {\hat \Sigma}_{M,12}^*(\mathbf x^*)^T\bm {\hat \Sigma}_{M,11}^*(\mathbf x^*)^{-1}\bm {\hat \Sigma}_{M,12}^*$ with $\bm {\hat \Sigma}_{M,11}$, $\bm {\hat \Sigma}_{M,22}$ and $\bm {\hat \Sigma}_{M,12}$ being the first $k_1 \times k_1$, last $k_2\times k_2$ entries in the diagonals and $k_1\times k_2$ entries in the off-diagonals of  $\bm {\hat \Sigma}_{M}^*$, respectively.

\section{Comparison to other approaches}
\label{sec:comparison}
In this section, we compare our method to various other frequently used approaches and discuss their connections and differences using examples.  First of all, note that the maximum likelihood estimator (MLE) of the factor loading matrix $\mathbf A$ under the Assumption \ref{assumption:A} is $\mathbf U_0 \mathbf R$ (without marginalizing out $\mathbf Z$), where $\mathbf U_0$ is the first $d$ ordered eigenvectors of $\mathbf Y \mathbf Y^T$ and $\mathbf R$ is an arbitrary orthogonal rotation matrix. This corresponds to the solution of principal component analysis, which is widely used in the literature for the inference of the latent factor model. For example,  \citet{bai2002determining} and \citet{bai2003inferential} assume that $\mathbf A^T \mathbf A = k \mathbf I_d$ and estimate  $\mathbf A$ by $\sqrt{k} \mathbf U_0$ in modeling high-dimensional time series. The estimation of factor loading matrix by the PCA is also applied in emulating  multivariate outputs from a computer model \citep{higdon2008computer}, where the factor loading matrix is estimated by the singular value decomposition of the standardized output matrix.  
 


The principal axes of the PCA are the same with those obtained from the PPCA, in which the factor loading matrix  is estimated by the maximum marginal likelihood, after marginalizing out the independent and normally distributed factors \citep{tipping1999probabilistic}. The estimator of the factor loadings is found to be the first $d$ columns of $\mathbf {\tilde U}_0  (\mathbf {\tilde  D}_0-\sigma^2_0 \mathbf I_d) \mathbf R$, where $\mathbf {\tilde D}_0$ is a diagonal matrix whose $l$th diagonal term is the $l$th largest eigenvalues of $\mathbf Y \mathbf Y^T/n$ and $\mathbf R$ is an arbitrary $d\times d$ orthogonal rotation matrix.  


 

The PPCA gives a probabilistic model of the PCA by modeling $\mathbf Z$ via independent normal distributions. However, when outputs are correlated across different inputs, modeling the factor processes as independent normal distributions  may not sensible in some applications. In comparison, the factors are allowed to be correlated in GPPCA; and we marginalize the factors out to estimate $\mathbf A$ to account for the uncertainty. This is why our approach can be viewed as a generalized approach of the PPCA for the correlated data.


The second observation is that the estimation of the factor loading matrix in the PCA or PPCA typically assumes the data are standardized. However, the standardization process could cause a loss of information and the uncertainty in the standardization is typically not considered. This problem is also  resolved by GPPCA, where the intercept and other covariates can be included in the model and the mean parameters can be marginalized out in estimating the factor loading matrix, as discussed in Section \ref{subsec:mean_structure}.


Next we illustrate the difference between the GPPCA and PCA using Example \ref{eg:demonstration}.

\begin{example}
The data is sampled from the model (\ref{equ:model_1}) with the shared covariance matrix $\bm \Sigma_1=\bm \Sigma_2=\bm \Sigma$, where $x$ is equally spaced from $1$ to $n$ and the kernel function is assumed to follow (\ref{equ:matern_5_2}) with $\gamma=100$ and $\sigma^2=1$. We choose $k=2$, $d=1$ and $n=100$. Two scenarios are implemented with $\sigma^2_0=0.01$ and $\sigma^2_0=1$, respectively. The parameters $(\sigma^2_0, \sigma^2, \gamma)$ are assumed to be unknown and estimated from the data.
\label{eg:demonstration}
\end{example}

\begin{figure}[t]
\centering
  \begin{tabular}{c}
    \includegraphics[width=1\textwidth,height=.3\textwidth]{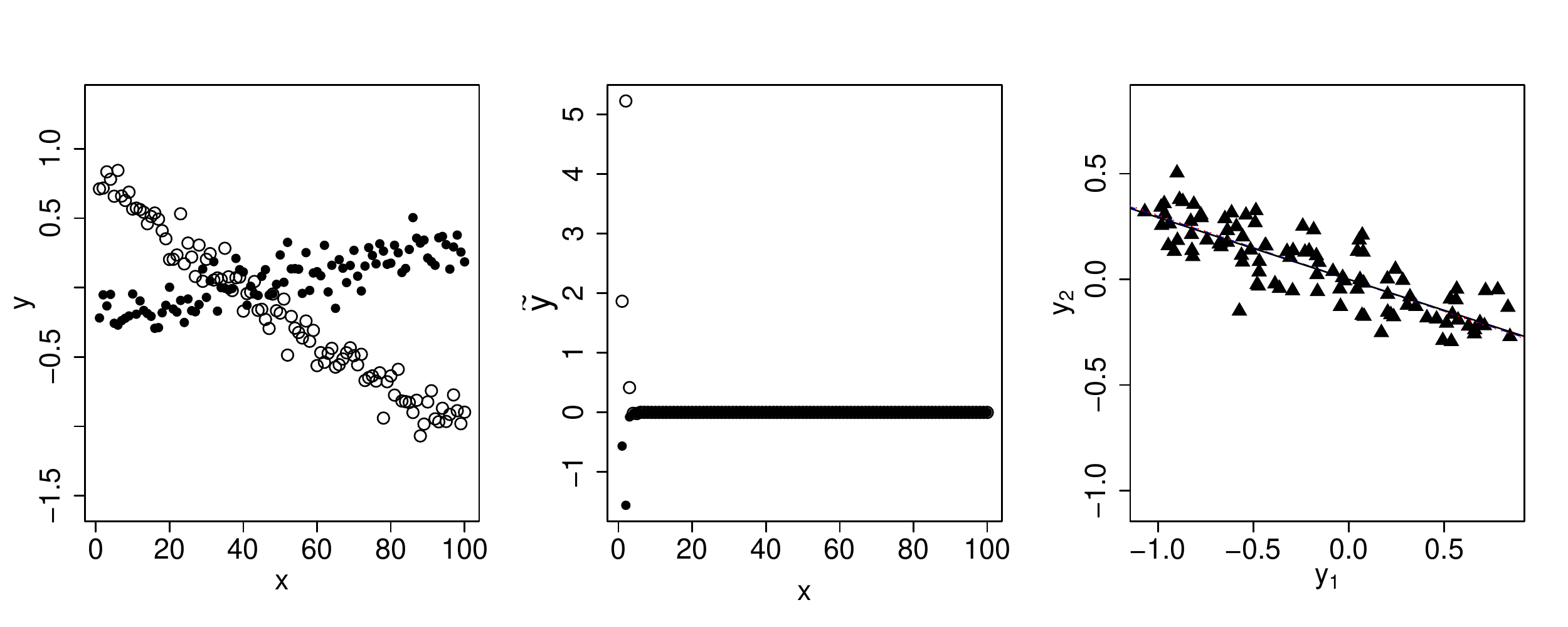}\vspace{-.15in}\\
        \includegraphics[width=1\textwidth,height=.3\textwidth]{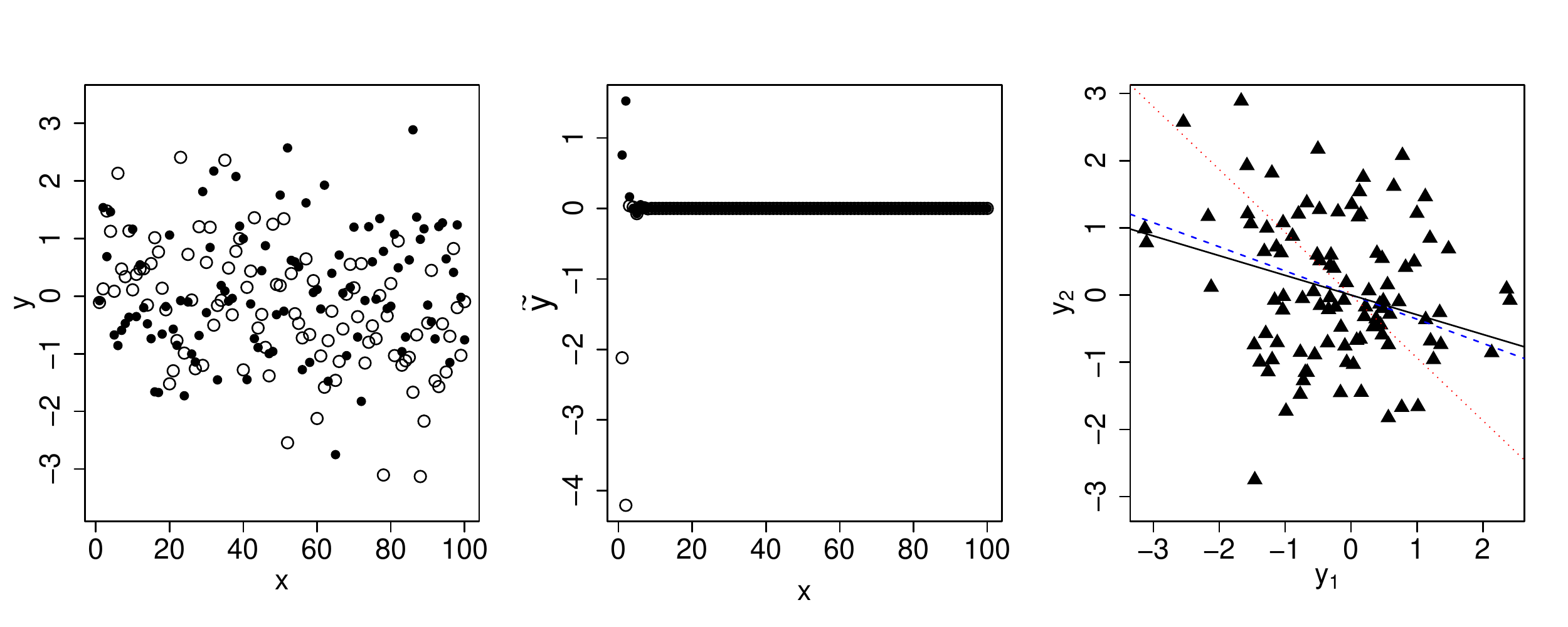}
  \end{tabular}
   \caption{Estimation of the factor loading matrix by the PCA and GPPCA for Example \ref{eg:demonstration} with the variance noise being  $\sigma^2_0=0.01$ and $\sigma^2_0=1$, graphed in the upper panels and lower panels, respectively. The circles and dots are the first and second rows of $\mathbf Y$ in the left panel, and of the transformed output $\mathbf {\tilde Y}=\mathbf Y \mathbf L$ in the middle panels,  where $\mathbf L=\mathbf U \mathbf D^{1/2}$ with $\mathbf U$ being the eigenvectors and the diagonals of $\mathbf D$ being the eigenvalues of the eigendecomposition of $(\hat \sigma^2_0\bm{ \hat \Sigma}^{-1}+\mathbf I_n )^{-1}$, where the $(i,j)$-term of $\bm{ \hat \Sigma}$ is $\hat \sigma^2 \hat K(\mathbf x_i, \mathbf x_j)$ by plugging the estimated range parameter $\hat \gamma$. The circles and dots in the middle panels almost overlap when $x$ is slightly larger than 0.  In the right panels,  the black solid lines,  red dotted lines and blue dash lines are the subspace of $\mathbf A$, the first eigenvector of $\mathbf U_0$ and the first eigenvector of $\mathbf G$ in Theorem~\ref{thm:est_A_shared_cov}, respectively, with the black triangles being the outputs. The black, blue and red lines almost overlap in the upper right panel.  }
\label{fig:demonstration}
\end{figure}

Note the linear subspace spanned from the column space of estimated loading matrix by the PCA or PPCA  is the same, which is  $\mathcal M(\mathbf U_0)$.  Thus we only compare the GPPCA to the PCA in Figure \ref{fig:demonstration} where $\mathbf A$ is a two-dimensional vector generated from a uniform distribution on the Stiefel manifold \citep{hoff2013bayesian}.  The signal to noise ratio (SNR) is $\tau=10^2$ and $\tau=1$ for the upper and lower panels in Figure \ref{fig:demonstration}, respectively.  


 From Figure \ref{fig:demonstration}, we observe that when the SNR is large, two rows of the outputs are strongly correlated, as shown in the upper left panel, with the empirical correlation being around $-0.83$ between two rows of the output $\mathbf Y$. The estimated subspaces by the PCA and GPPCA both match the true $\mathbf A$ equally well in this scenario, shown in the upper right panel. When the variance of the noise gets large, the outputs are no longer very correlated. For example, the empirical correlation between two simulated output variables is only around $-0.18$.  As a result, the angle between the estimated subspace and the column space of $\mathbf A$ by the PCA is large, as shown in the right lower panel. 
 
The GPPCA by Theorem \ref{thm:est_A_shared_cov} essentially transforms the output by $\mathbf {\tilde Y}=\mathbf Y \mathbf L$, graphed in the middle panels, where $\mathbf L=\mathbf U \mathbf D^{1/2}$ with $\mathbf U$ and $\mathbf D$ being a matrix of eigenvectors and a diagonal matrix of the eigenvalues from the eigendecomposition of $(\hat \sigma^2_0\bm {\hat  \Sigma}^{-1}+\mathbf I_n )^{-1}$, respectively, where variance parameter and kernel parameter are estimated by the MMLE discussed in Section \ref{subsec:par_est}. The two rows of the transformed outputs are strongly correlated, shown in the middle panels. The empirical correlation between two rows of the transformed outputs graphed in the lower panel is about $-0.99$, even though the variance of the noise is as large as the variance of the signal. The subspace by the GPPCA is equivalent to the first eigenvector of the transformed output for this example, and it is graphed as the blue dashed curves in the right panels. The estimated subspace by the GPPCA is  close to the truth in both scenarios, even when the variance of the noise is large in the second scenario.


\begin{figure}[t]
\centering
  \begin{tabular}{c}
    \includegraphics[width=1\textwidth,height=.35\textwidth]{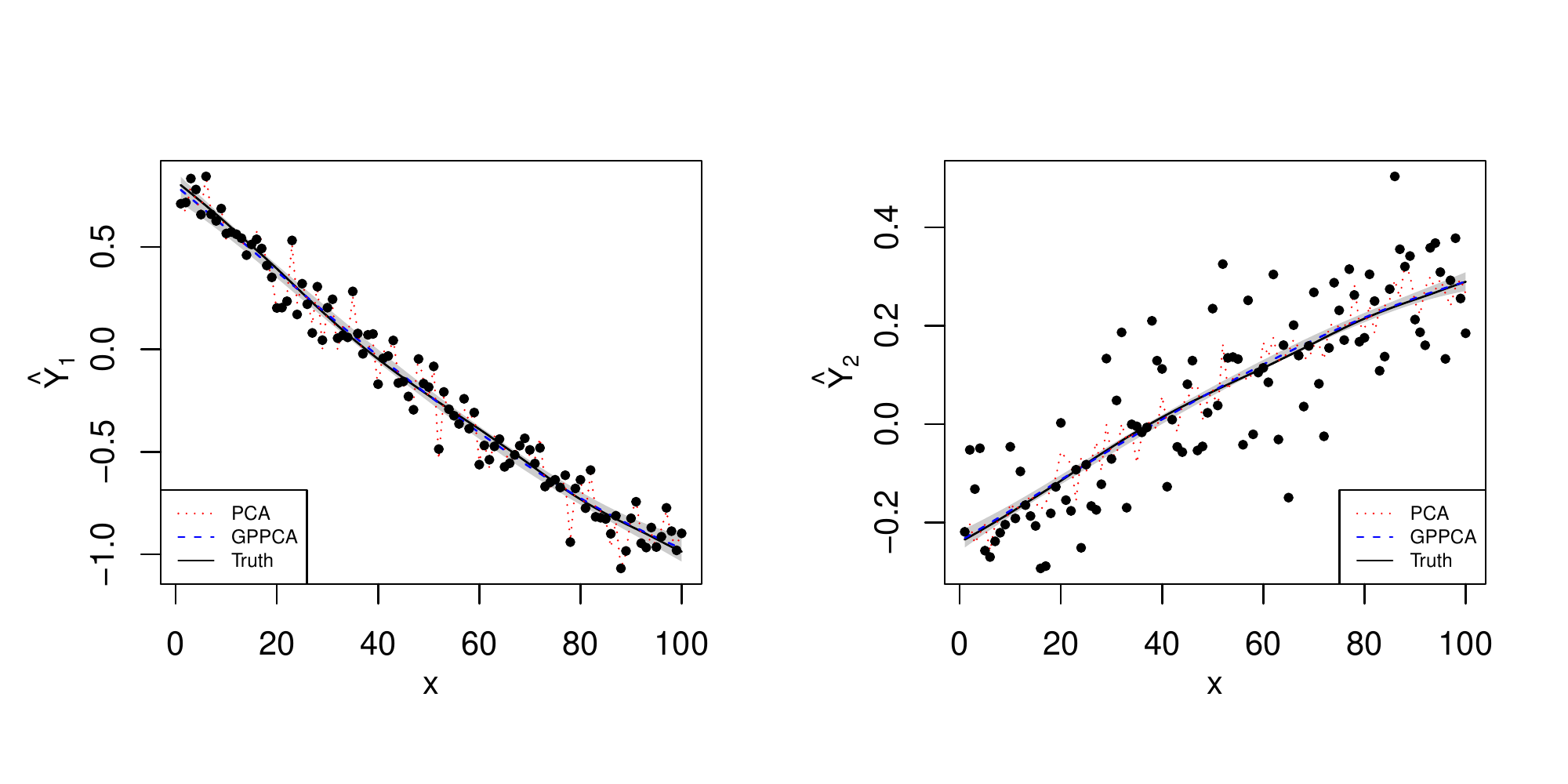}\vspace{-.5in} \\
        \includegraphics[width=1\textwidth,height=.35\textwidth]{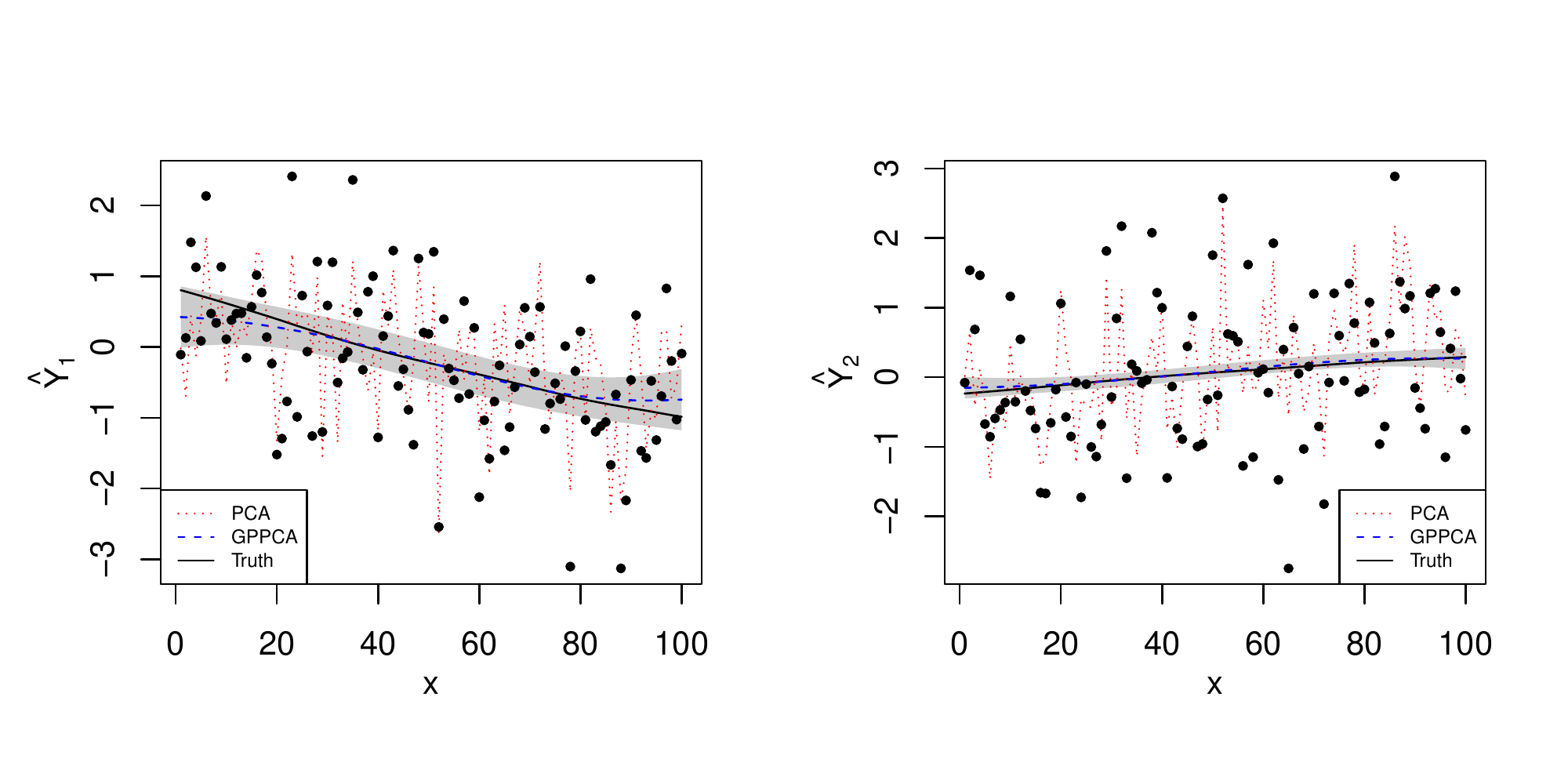}
  \end{tabular}
  \vspace{-.3in}
   \caption{Estimation of the mean of the output $\mathbf Y$ for Example \ref{eg:demonstration} with the variance of the noise being  $\sigma^2_0=0.01$ and $\sigma^2_0=1$, graphed in the upper panels and lower panels, respectively. The first row and second row of $\mathbf Y$ are graphed as the black curves in the left panels and right panels, respectively. The red dotted curves and the blue dashed curves are the prediction by the PCA and GPPCA, respectively. The grey region is the $95\%$ posterior credible interval from GPPCA. The black curves, blue curves and grey regions almost overlap in the upper panels. }
\label{fig:demonstration_AZ}
\end{figure}


For PCA, the mean of the outputs is typically estimated  by the maximum likelihood estimator $\mathbf {\hat A}_{pca} \mathbf {\hat A}^T_{pca}\mathbf Y$, where $\mathbf {\hat A}_{pca}=\mathbf U_0$ \citep{bai2002determining}. In Figure \ref{fig:demonstration_AZ},  the PCA estimation of the mean for Example \ref{eg:demonstration} is graphed as the red curves and the posterior mean of the output in the GPPCA in Corollary \ref{cor:posterior_AZ} is graphed as the blue curves. The PCA  underestimates the variance of the noise and hence has a large estimation error.  In comparison, the estimated mean of the output by the GPPCA is more accurate, as the correlation in each  output variable is properly modeled through the GPs of the latent factors.  



Note that we restrict $\mathbf A$ to satisfy $\mathbf A^T \mathbf A=\mathbf I_d$ when simulating data examples in Figure \ref{fig:demonstration}. In practice, we find this constraint only affects the estimation of the variance parameter $\sigma^2_l$  in the kernel, $l=1,...,d$, because the meaning of this parameter changes.

There are some other estimators of the factor loading matrix in modeling high-dimensional time series. For example,  \cite{lam2011estimation,lam2012factor} estimate the factor loading matrix of model (\ref{equ:model_1}) by  $\hat {\mathbf A}_{LY}:=\sum^{q_0}_{q=1}\bm {\hat \Sigma}_y(q) \bm {\hat \Sigma}^T_y(q)$, where $\bm {\hat \Sigma}_y(q)$ is the $k\times k$ sample covariance at lag $q$ of the output and $q_0$ is fixed to be a small positive integer. This approach is sensible, because  $\mathcal M(\mathbf A)$ is shown to be spanned from $\sum^{q_0}_{q=1}\bm { \Sigma}_y(q) \bm { \Sigma}^T_y(q)$ under some reasonable assumptions, where $\bm {\Sigma}_y(q)$ is the underlying   lag-$q$ covariance of the outputs.  It is also suggested in \cite{lam2012factor} to estimate the latent factor by $\mathbf {\hat Z}_{LY}=\hat {\mathbf A}^T_{LY}\mathbf Y$, meaning that  the mean of the output is estimated by $ \hat {\mathbf A}_{LY} \mathbf {\hat Z}_{LY}=\hat {\mathbf A}_{LY} \hat {\mathbf A}^T_{LY}\mathbf Y$.  This estimator and the PCA are both included for comparison in Section \ref{sec:numerical_results}.. 



\section{Simulated examples}
\label{sec:numerical_results}
In this section, we numerically compare different approaches studied before. We use several criteria to examine the estimation. The first criterion is the largest principal angle between the estimated subspace $\mathcal M(\mathbf{\hat A})$ and the true subspace $\mathcal M({\mathbf A})$. Let  
$0\leq \phi_1\leq...\leq \phi_d\leq \pi/2$ be the principal angles between $\mathcal M({\mathbf A})$ and $\mathcal M({\mathbf {\hat A}})$, recursively defined by 
\begin{equation*}
\phi_i=\mbox{arccos} \left( \max_{\mathbf a \in \mathcal M(\mathbf A), \mathbf {\hat a} \in \mathcal M(\mathbf {\hat A}) } |\mathbf a^T \mathbf {\hat a}|  \right)=\mbox{arccos}(|\mathbf a^T_i \mathbf {\hat a}_i| ), 
\end{equation*}
subject to 
\[ ||\mathbf a||=||\mathbf {\hat a}||=1, \, \mathbf a^T \mathbf a_i=0,\,\mathbf {\hat a}^T \mathbf {\hat a}_i=0,\, i=1,...,d-1, \]
where $||\cdot||$ denotes the $L_2$ norm. The largest principal angle is $\phi_d$, which quantifies how close two linear subspaces are. When two subspaces are identical, all principal angles are zero. When the columns of the $\mathbf A$ and $\mathbf {\hat A}$ form orthogonal bases of the $\mathcal M(\mathbf A)$ and $\mathcal M(\mathbf {\hat A})$, the cosine of the largest principal angle is equal to the smallest singular value of $\mathbf A^T \mathbf {\hat A}$ \citep{bjorck1973numerical,absil2006largest}. Thus the largest principal angle can be calculated efficiently through the singular value decomposition of $\mathbf {A}^T \mathbf {\hat A}$. 



We numerically compare four approaches for estimating $\mathbf A$. The first approach is the PCA, which estimates $\mathbf A$ by $\mathbf U_0$, where $\mathbf U_0$ is the first $d$ eigenvectors of  $\mathbf Y \mathbf Y^T/n$. Note the other  version of the PCA and the PPCA  have the same largest principal angle between the estimated subspace of  $\mathbf A$ and the true subspace of  $\mathbf A$, so the results are omitted. The GPPCA  is the second approach. When the covariance of the factor processes is the same, the closed-form expression of the estimator of the factor loading matrix is given in Theorem \ref{thm:est_A_shared_cov}. When the covariance of the factor processes is different, we implement the optimization algorithm in \cite{wen2013feasible}  that preserves the orthogonal constraints to obtain the maximum marginal likelihood estimation of the factor loading matrix in Theorem \ref{thm:est_A_diff_cov}. In both cases, the estimator $\mathbf {\hat A}$ can be written as a function of $(\bm \gamma,\bm \tau, \sigma^2_0)$ which are estimated by maximizing the marginal likelihood after integrating out $\mathbf Z$ and plugging $\mathbf {\hat A}$. The third approach, denoted as LY1, estimates $\mathbf A$ by $\bm {\hat \Sigma}_y(1) \bm {\hat \Sigma}^T_y(1)$, where $\bm {\hat \Sigma}_y(1)$ is the sample covariance of the output at lag $1$ and the fourth approach, denoted as LY5, estimates $\mathbf A$ by $\sum^{q_0}_{q=1}\bm {\hat \Sigma}_y(q) \bm {\hat \Sigma}^T_y(q)$ with $q_0=5$, used  in  \cite{lam2012factor} and \cite{lam2011estimation}, respectively. 

We also compare the performance of different approaches by the average mean squared errors (AvgMSE) in predicting the mean of the output over $N$ experiments as follows 
\begin{equation}
    \mbox{AvgMSE}= \sum^N_{l=1} \sum^k_{j=1} \sum^n_{i=1} \frac{( \hat Y^{(l)}_{j,i}- \E[Y^{(l)}_{j,i}])^2 }{knN},
    \label{eg:avgmse}
\end{equation}
where $\E[Y^{(l)}_{j,i}]$ is the $(j,i)$  element of the mean of the output matrix at the $l$th experiment, and $\hat Y^{(l)}_{j,i}$ is the estimation. As discussed in Section \ref{sec:numerical_results}, the estimated mean of the output matrix by the PCA, LY1 and LY5 is   $\mathbf {\hat A} \mathbf {\hat A}^T\mathbf Y$, where $\mathbf {\hat A}$ is the estimated factor loading matrix in each approach (\cite{bai2002determining,lam2011estimation,lam2012factor}). In GPPCA, we use the posterior mean of $\mathbf A \mathbf Z$ in Corollary \ref{cor:posterior_AZ} to estimate mean of the output matrix.



The cases of the shared covariance  and the different covariances of the factor processes are studied  in Example \ref{eg:shared_cov} and Example \ref{eg:diff_cov}, respectively. we assume that $\mathbf A$ is sampled from the uniform distribution on the Stiefel manifold \citep{hoff2013bayesian}, and the kernels are correctly specified with unknown parameters in these  examples. In Appendix C, we compare different approaches when  the factor loading matrix, kernel functions or the factors are misspecified.






\begin{figure}[t]
\centering
  \begin{tabular}{cccc}
      \includegraphics[width=.25\textwidth,height=.35\textwidth]{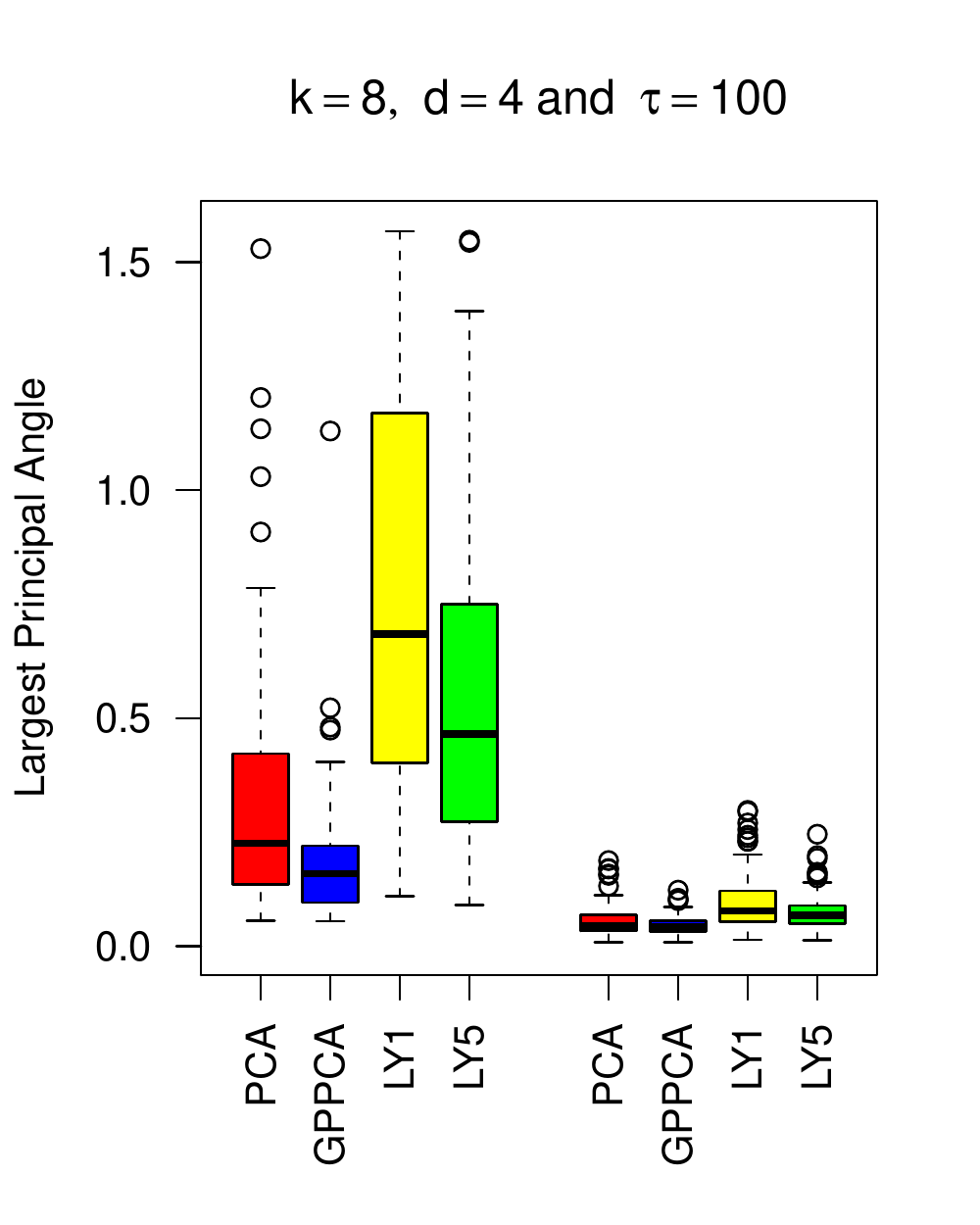}
        \includegraphics[width=.25\textwidth,height=.35\textwidth]{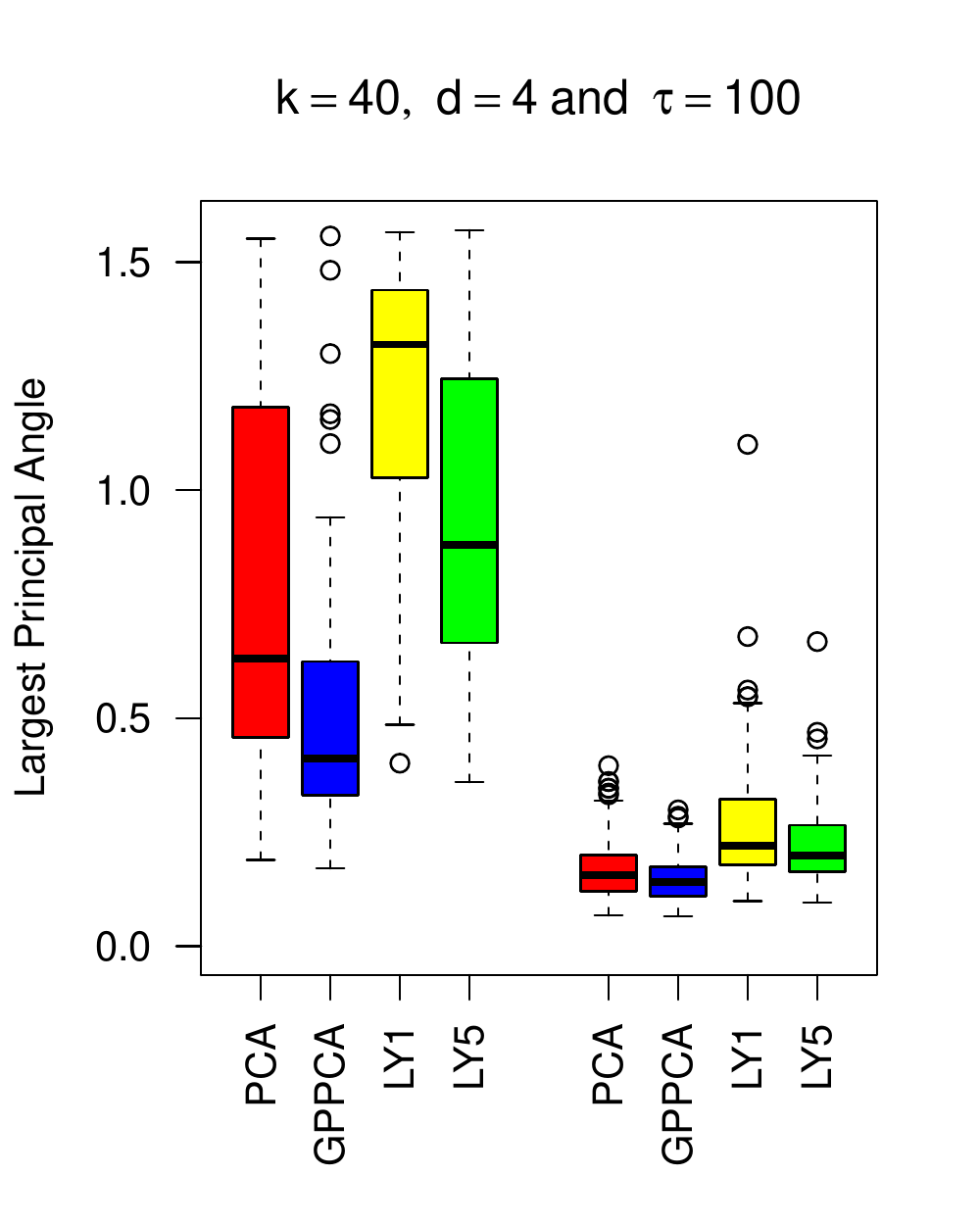}
            \includegraphics[width=.25\textwidth,height=.35\textwidth]{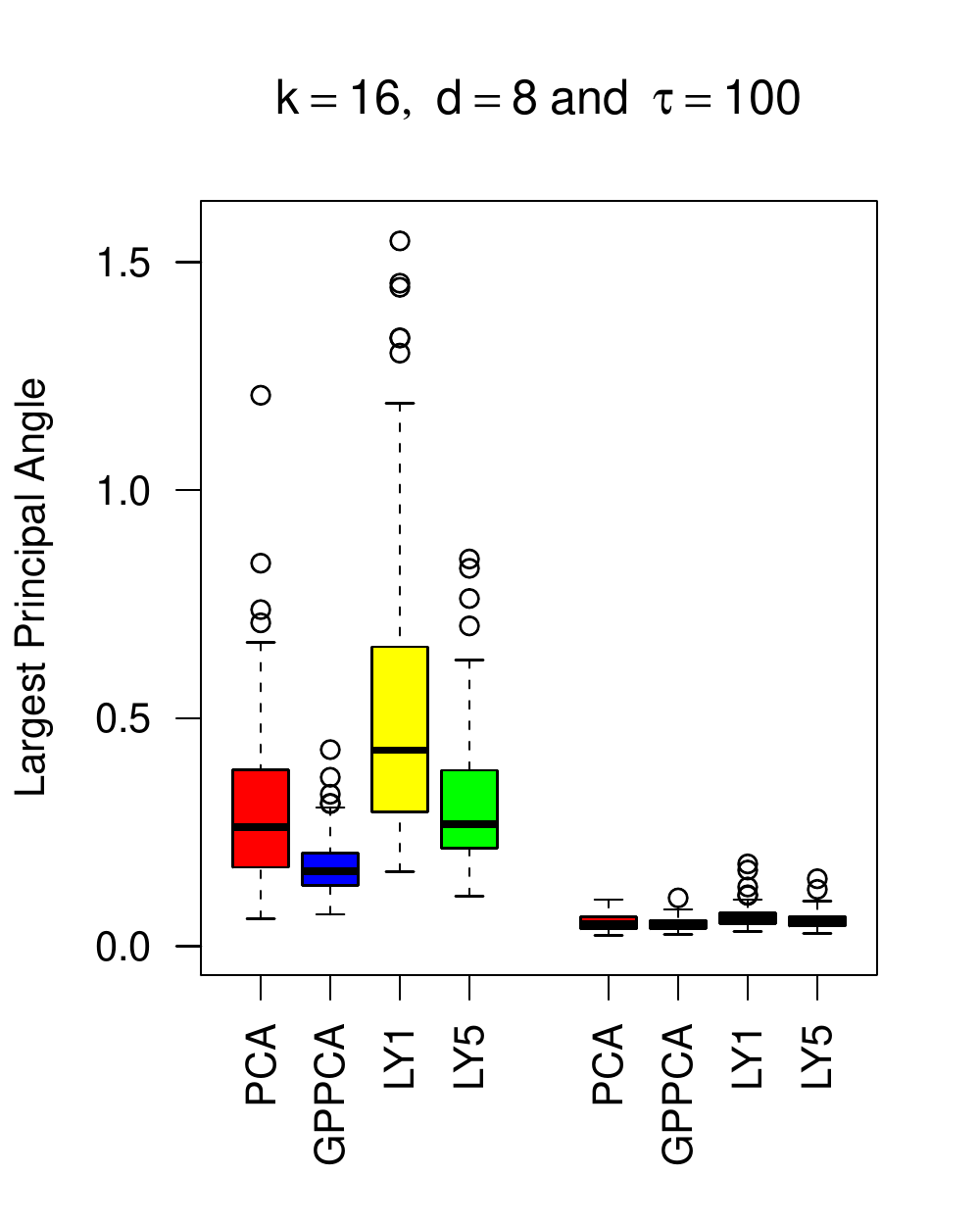}
        \includegraphics[width=.25\textwidth,height=.35\textwidth]{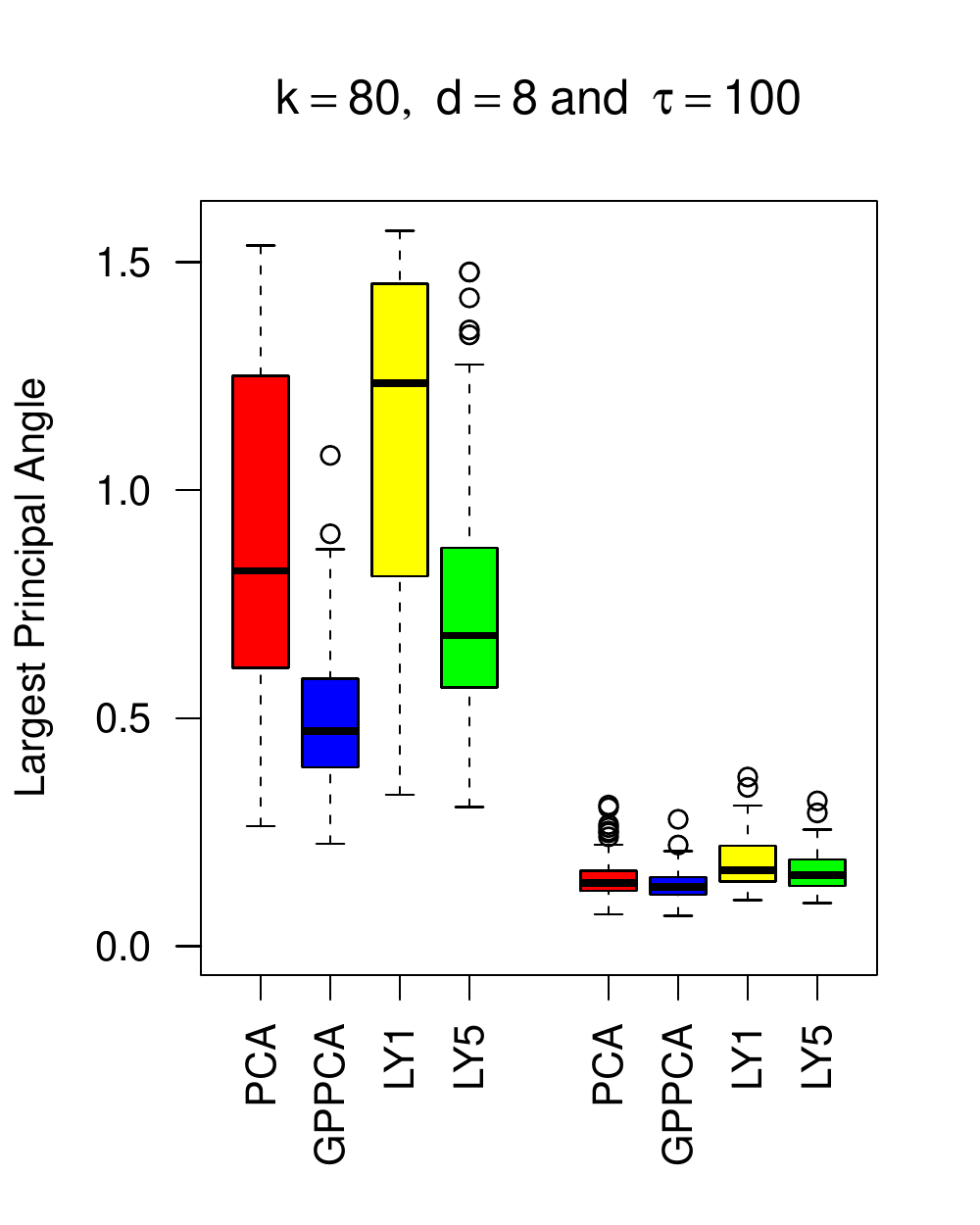}\vspace{-.1in}\\
    \includegraphics[width=.25\textwidth,height=.35\textwidth]{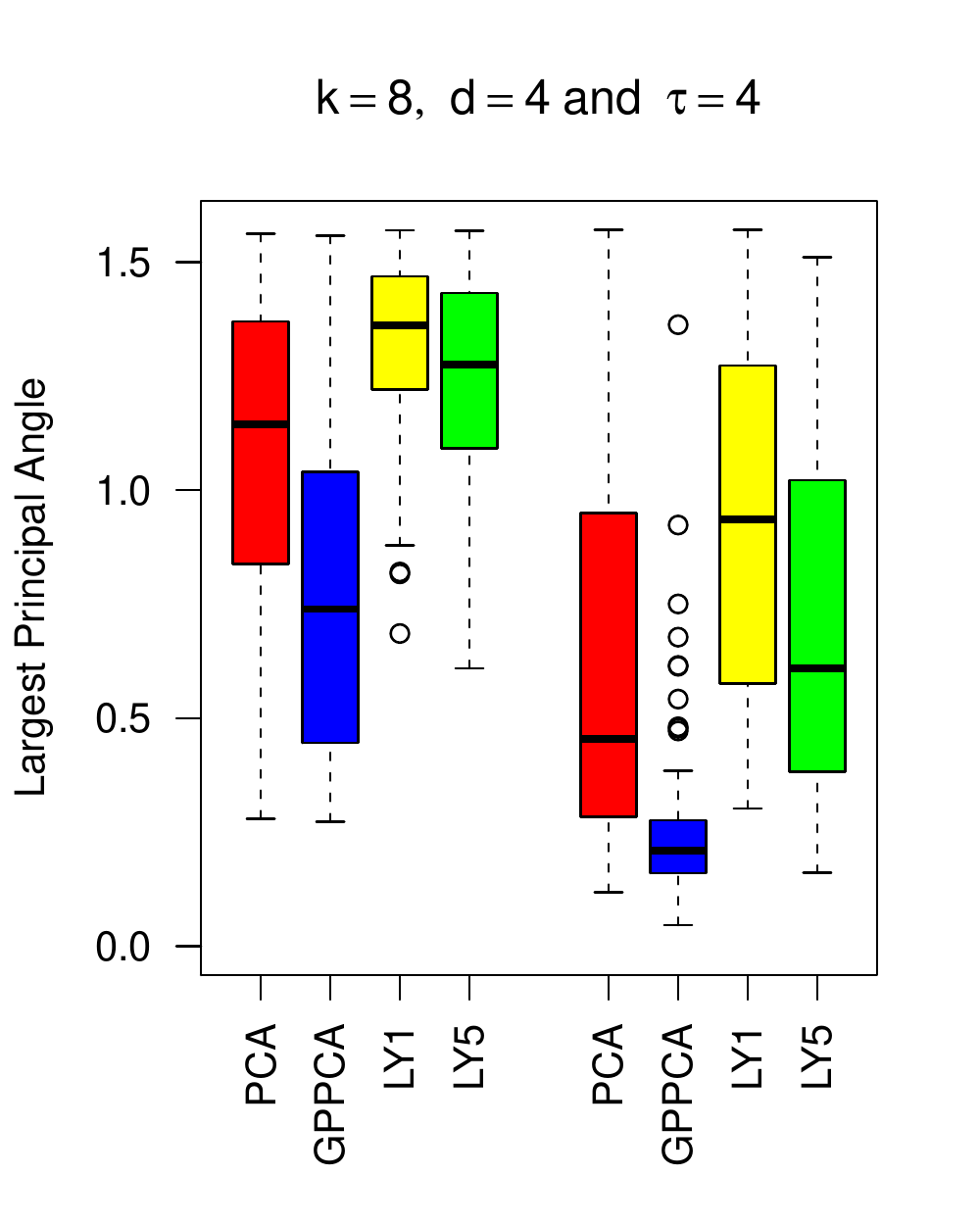}
        \includegraphics[width=.25\textwidth,height=.35\textwidth]{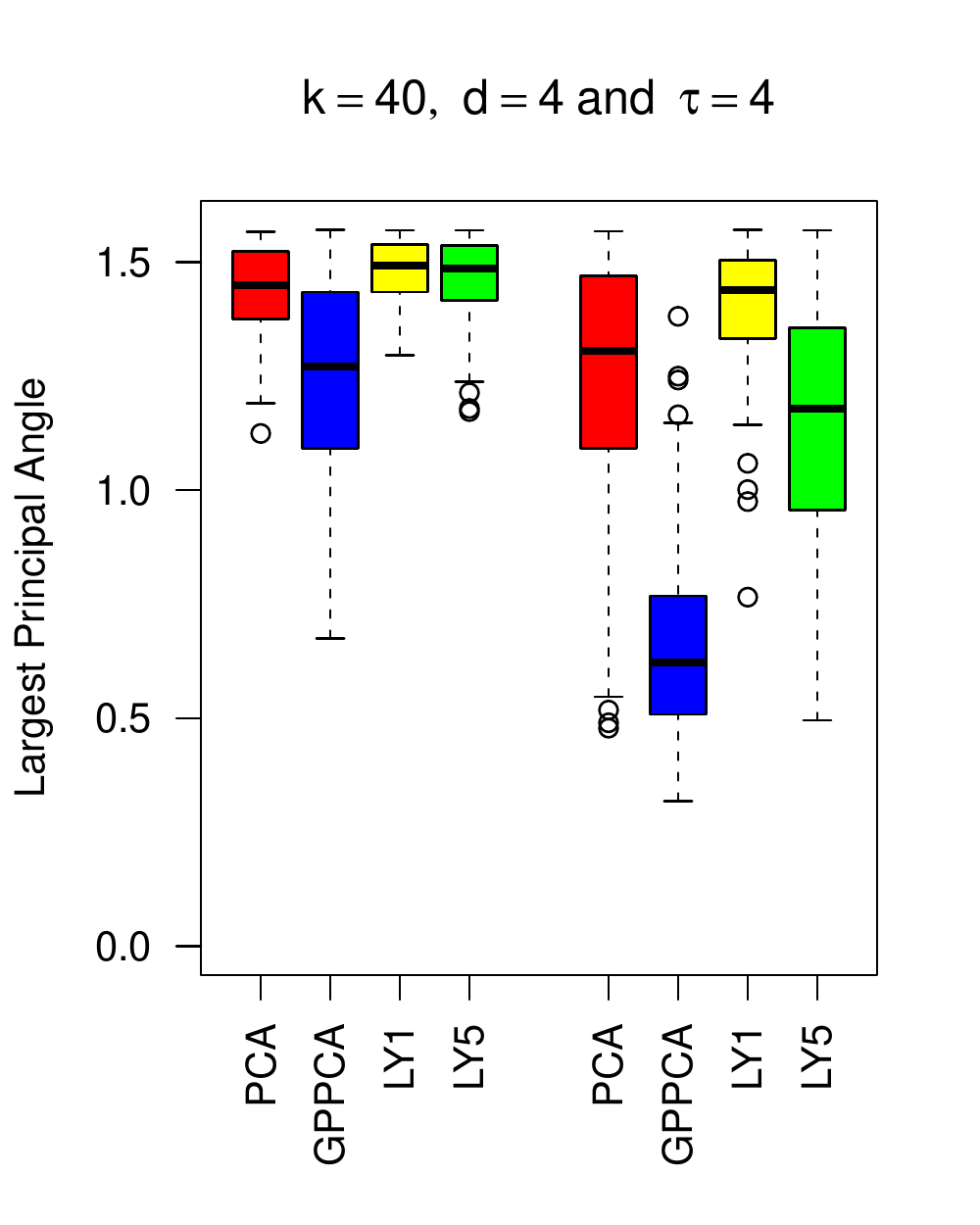}
            \includegraphics[width=.25\textwidth,height=.35\textwidth]{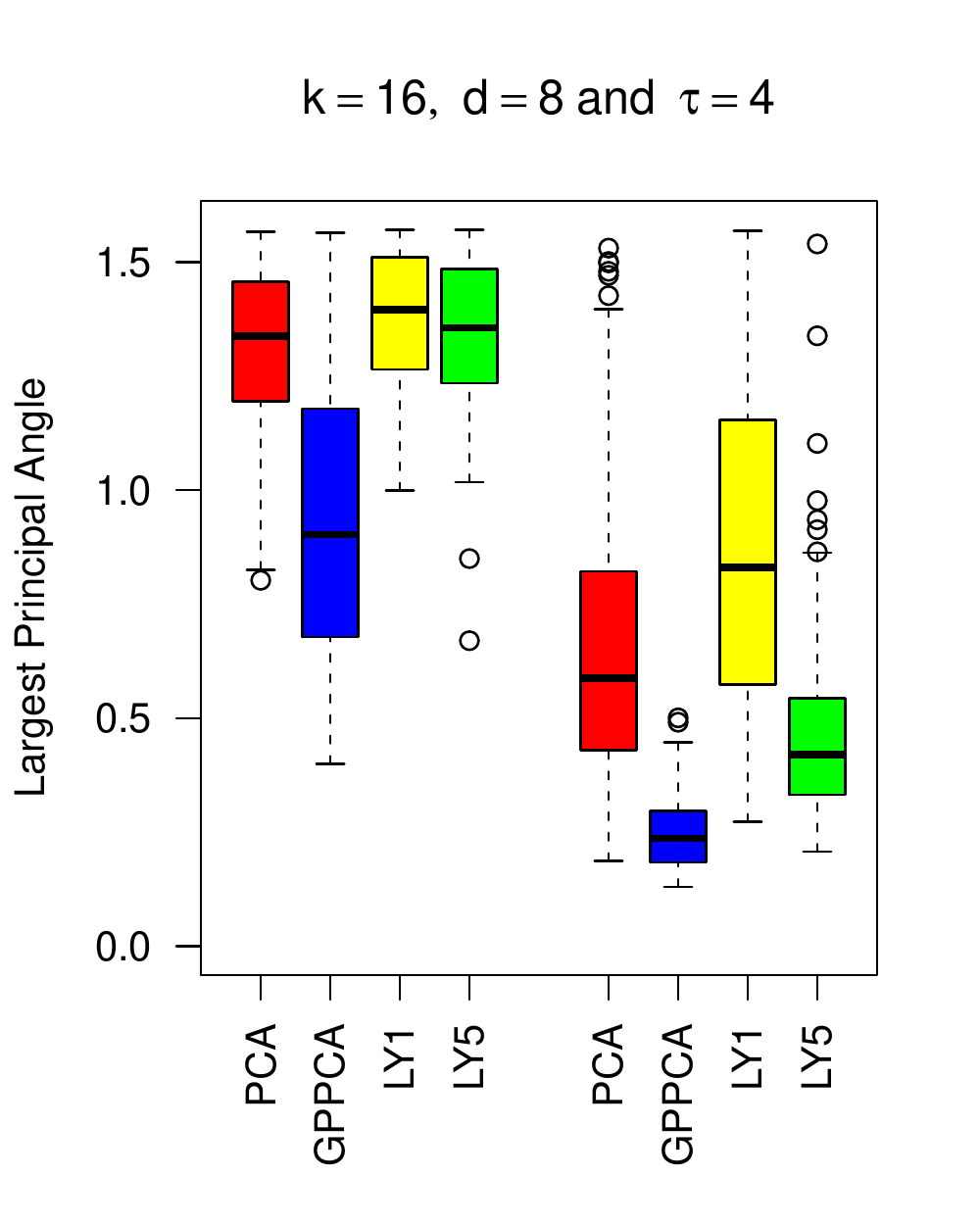}
        \includegraphics[width=.25\textwidth,height=.35\textwidth]{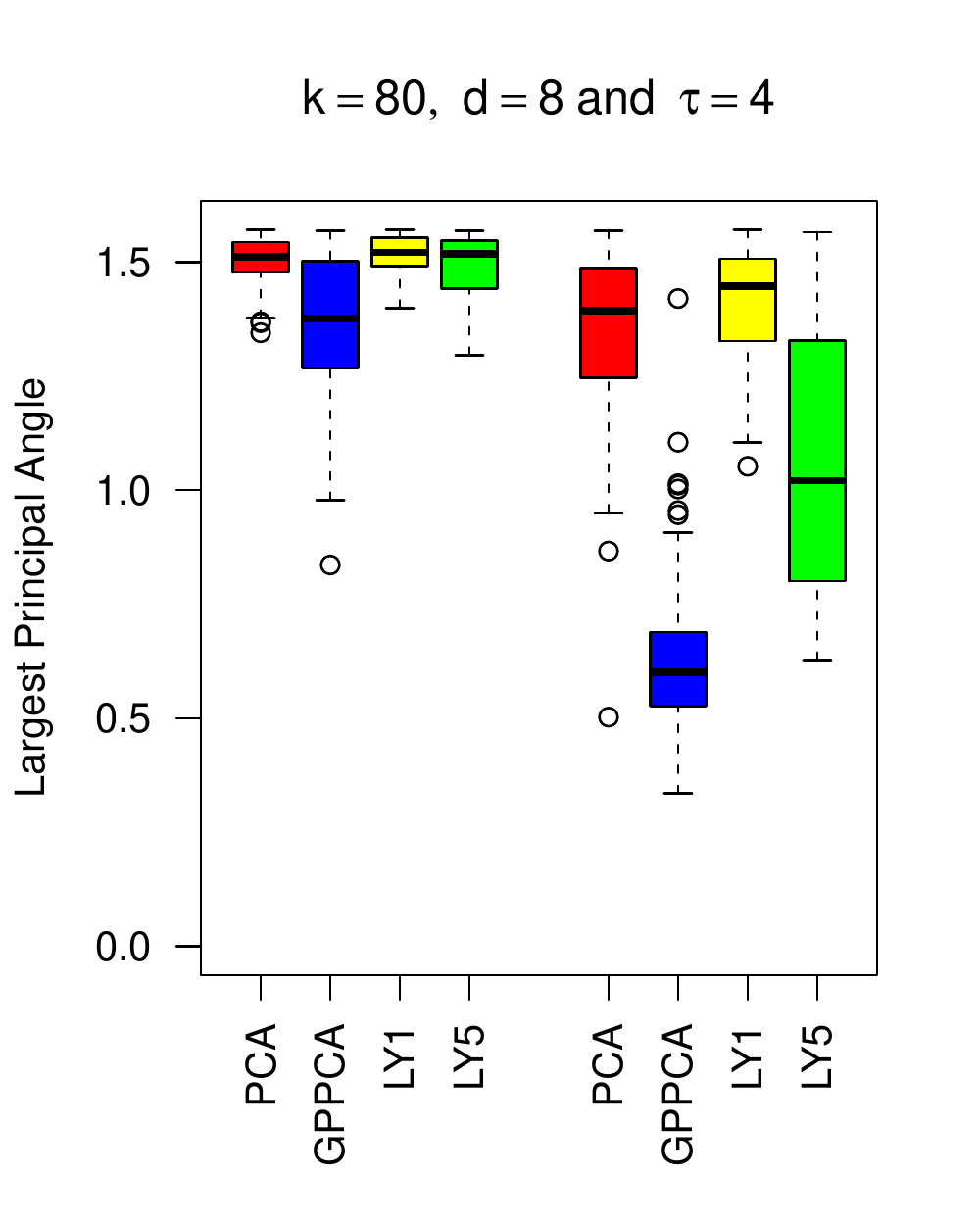}
  \end{tabular}
  \vspace{-.1in}
   \caption{The largest principal angle between the true subspace of the factor loading matrix and the estimation from the four approaches  for Example \ref{eg:shared_cov} (ranging from $[0,\pi/2]$, the smaller the better).  In the first row, the number of the observations of each output variable is assumed to be $n=200$ and $n=400$ for the left four boxplots and right four boxplots in each panel, respectively. In the second row, the number of observations is  assumed to be $n=500$ and $n=1000$ for the left four boxplots and right four boxplots in each panel, respectively. }
\label{fig:simulation_shared_cov_angles}
\end{figure}



\begin{example}[Factors with the same covariance matrix]
The data are sampled from model (\ref{equ:model_1}) with $\bm \Sigma_1=...=\bm \Sigma_d=\bm \Sigma$, where $x_i=i$ for $1\leq i\leq n$, and the kernel function in (\ref{equ:matern_5_2}) is used with $\gamma=100$ and $\sigma^2=1$. 
In each scenario, we simulate the data from $16$ different combinations of  $\sigma^2_0$, $k$, $d$ and $n$. We repeat $N=100$ times for each scenario. The parameters $(\sigma^2_0,  \sigma^2, \gamma)$ are treated as unknown and estimated from the data.
\label{eg:shared_cov}
\end{example}



 In Figure \ref{fig:simulation_shared_cov_angles}, we present the largest principal angle between the true subspace $\mathcal M(\mathbf A)$ and  estimated subspace $\mathcal M(\mathbf {\hat A})$ at different settings of Example 1. The red, blue, yellow and green boxplots are the results from the PCA, GPPCA, LY1 and LY5. In each panel, the sample size gets doubled from the left four boxplots to the right four. The SNR $\tau =\sigma^2/\sigma^2_0$ is assumed to be $100$ and $4$ in the upper panels and lower panels, respectively.

\begin{table}[t]
\begin{center}
\begin{tabular}{lcccc}
  \hline
 $d=4$ and $\tau=100$         &  \multicolumn{2}{c}{k=8} & \multicolumn{2}{c}{k=40}\\
 &  $n=200$ &  $n=400$ & $n=200$ &  $n=400$  \\
  \hline
  PCA            &$5.3\times 10^{-3}$& $5.1\times 10^{-3}$ &$1.4\times 10^{-3}$& $1.1\times 10^{-3}$  \\
  GPPCA          &$\bf 3.3\times 10^{-4}$ & $\bf 2.6\times 10^{-4}$ &$\bf 2.2\times 10^{-4}$ & $\bf 1.3\times 10^{-4}$ \\
  LY1 & $4.6\times 10^{-2}$ & $5.8 \times 10^{-3}$ &$1.5\times 10^{-2}$ &$2.1\times 10^{-3}$ \\
  LY5 &  $3.2\times 10^{-2}$ & $5.5\times 10^{-3}$ & $1.1\times 10^{-2}$&$1.8\times 10^{-3}$\\

  \hline
 $d=8$ and $\tau=100$         &  \multicolumn{2}{c}{k=16} & \multicolumn{2}{c}{k=80}\\
 &  $n=500$ &  $n=1000$ & $n=500$ &  $n=1000$  \\
  \hline
  PCA            &$5.2\times 10^{-3}$& $5.0\times 10^{-3}$ &$1.3\times 10^{-3}$& $1.1\times 10^{-3}$  \\
  GPPCA          &$\bf 2.9\times 10^{-4}$ & $\bf 2.4\times 10^{-4}$ &$\bf 1.9\times 10^{-4}$ & $\bf 1.1\times 10^{-4}$ \\
  LY1 & $1.4\times 10^{-2}$ & $5.1 \times 10^{-3}$ &$5.4\times 10^{-3}$ &$1.2\times 10^{-3}$ \\
  LY5 &  $8.8\times 10^{-3}$ & $5.1\times 10^{-3}$ & $3.9\times 10^{-3}$&$1.2\times 10^{-3}$\\
  \hline

 $d=4$ and $\tau=4$         &  \multicolumn{2}{c}{k=8} & \multicolumn{2}{c}{k=40}\\
 &  $n=200$ &  $n=400$ & $n=200$ &  $n=400$  \\
  \hline
  PCA            &$1.4\times 10^{-1}$& $1.3\times 10^{-1}$ &$4.2\times 10^{-2}$& $3.4\times 10^{-2}$  \\
  GPPCA          &$\bf 5.8\times 10^{-3}$ & $\bf 4.4\times 10^{-3}$ &$\bf 5.3\times 10^{-3}$ & $\bf 3.0\times 10^{-3}$ \\
    LY1 & $2.2\times 10^{-1}$ & $1.7 \times 10^{-1}$ &$7.2\times 10^{-2}$ &$6.4\times 10^{-2}$ \\
  LY5 &  $2.2\times 10^{-1}$ & $1.5\times 10^{-1}$ & $4.8\times 10^{-2}$&$4.1\times 10^{-2}$\\
  \hline
 $d=8$ and $\tau=4$         &  \multicolumn{2}{c}{k=16} & \multicolumn{2}{c}{k=80}\\
 &  $n=500$ &  $n=1000$ & $n=500$ &  $n=1000$  \\
  \hline
  PCA            &$1.4\times 10^{-1}$& $1.3\times 10^{-1}$ &$3.9\times 10^{-2}$& $3.2\times 10^{-2}$  \\
  GPPCA          & $\bf 5.1\times 10^{-3}$ & $\bf 3.9\times 10^{-3}$ &$\bf 4.3\times 10^{-3}$ & $\bf 2.4\times 10^{-3}$ \\
      LY1 & $1.8\times 10^{-1}$ & $1.4 \times 10^{-1}$ &$5.1\times 10^{-2}$ &$3.4\times 10^{-2}$ \\
  LY5 &  $1.7\times 10^{-1}$ & $1.3\times 10^{-1}$ & $4.6\times 10^{-2}$&$3.1\times 10^{-2}$\\

  \hline

\end{tabular}
\end{center}
   \caption{AvgMSE for Example \ref{eg:shared_cov}.}

   \label{tab:AvgMSE_shared_cov}
\end{table}

 Since the covariance of the factor processes is the same in Example \ref{eg:shared_cov},  the estimated $\mathbf A$ by the GPPCA has a closed-form solution given in Theorem \ref{thm:est_A_shared_cov}. For all 16 different scenarios, the GPPCA outperforms the other three methods in terms of having the smallest largest principal angle between   $\mathcal M(\mathbf A)$ and $\mathcal M(\mathbf {\hat A})$. Both PCA and GPPCA can be viewed as maximum likelihood type of approaches under the orthonormality assumption of the factor loading matrix. The difference is that the estimator of $\mathbf A$ by the GPPCA maximizes the marginal likelihood after integrating out the factor processes, whereas the PCA maximizes the likelihood without modeling the factor processes. The principal axes by the PCA are the same as the PPCA which assumes the factors are independently distributed.  As discussed before, the model with independent factors, however, is not a sensible sampling model for the correlated data, such as the multiple time series or multivariate spatial processes. 
 

The performance of all methods improves when the sample size increases or when the SNR increases, shown in Figure \ref{fig:simulation_shared_cov_angles}. The  LY5 estimator  \citep{lam2011estimation} seems to perform slightly better than the PCA when the SNR is smaller. This method is sensible  because the factor loading space $\mathcal M(\mathbf A)$ is spanned by the eigenvectors of $\mathbf M:=\sum^{q_0}_{i=1}\bm {\Sigma}_y(q) \bm {\Sigma}^T_y(q)$ under some conditions. However, this may not be the unique way to represent the  subspace of the factor loading matrix. Thus the estimator based on this argument may not be as efficient as the maximum marginal likelihood approach by the GPPCA, shown in Figure \ref{fig:simulation_shared_cov_angles}.


\begin{figure}[t]
\centering
  \begin{tabular}{c}
    \includegraphics[width=1\textwidth,height=.4\textwidth]{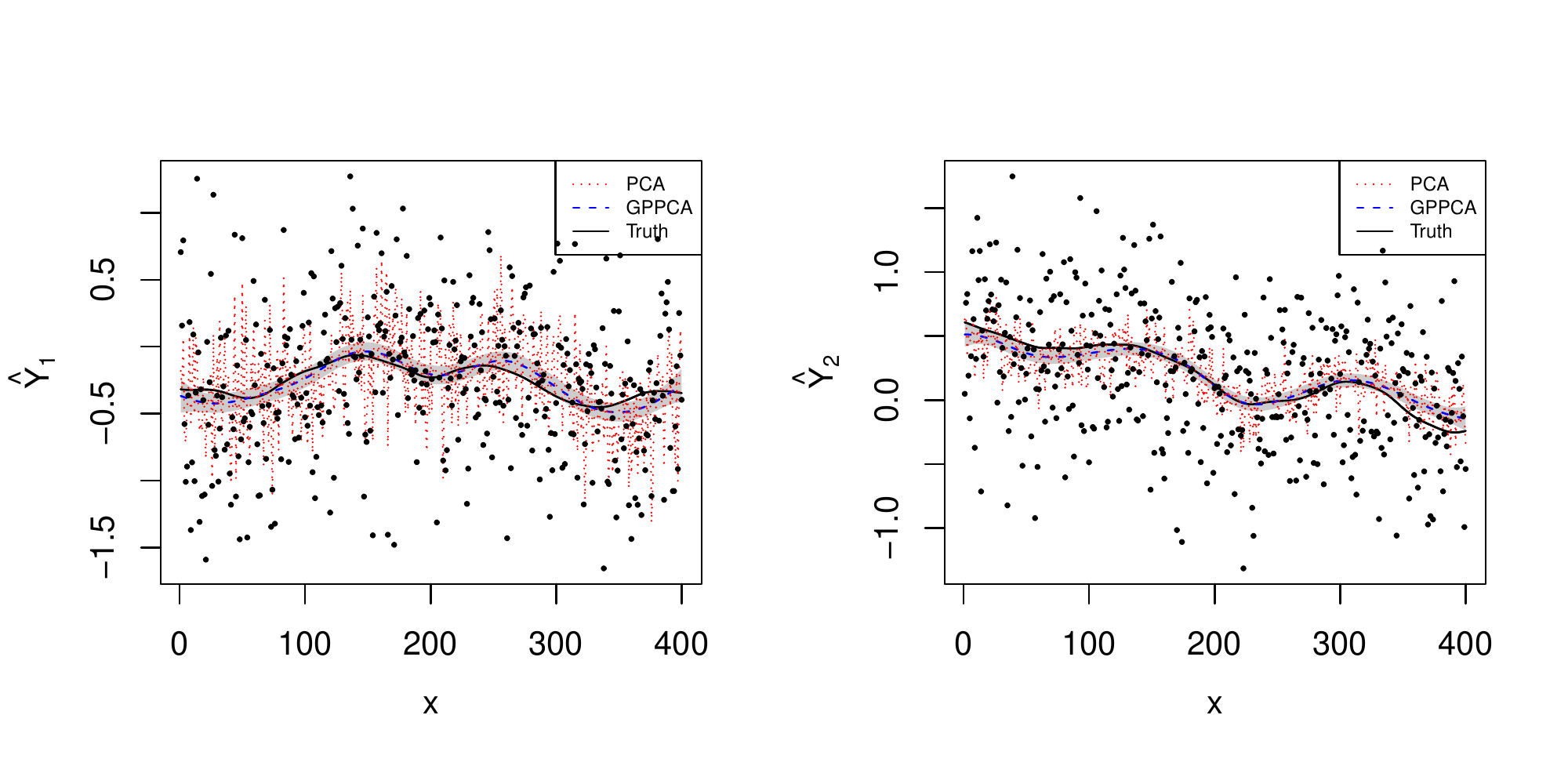}
  \end{tabular}
  \vspace{-.2in}
   \caption{Prediction of the mean of the first two output variables in one experiment with $k=8$, $d=4$, $n=400$ and $\tau=4$. The observations are plotted as black circles and the truth is graphed as the black curves. The estimation by the PCA and GPPCA is graphed as the red dotted curves and blue dashed curves, respectively. The shaded area is the 95\% posterior credible interval by the GPPCA.  }
\label{fig:shared_cov_pred_mean}
\end{figure}


The AvgMSE of the different approaches for Example \ref{eg:shared_cov} is shown in Table \ref{tab:AvgMSE_shared_cov}. The mean squared error of the estimation by the GPPCA is typically a digit or two smaller than the ones by the other approaches. This is because the correlation of the factor processes in the GPPCA is properly modeled, and the kernel parameters are estimated based on the maximum marginal likelihood estimation.



We plot the first two rows of the estimated mean of the output in one experiment from the Example \ref{eg:shared_cov} in Figure \ref{fig:shared_cov_pred_mean}. The estimation of the GPPCA approach is graphed as the blue dashed curves, which is very close to the truth, graphed as the black curves, wheares the estimation by the PCA is graphed as the red dotted curves, which are less smooth and less accurate in predicting the mean of the outputs, because of the noise in the data. The estimators by LY1 and LY5 are similar to those of PCA so we omit them in Figure \ref{fig:shared_cov_pred_mean}. The problem of the PCA (and PPCA) is that the estimation assumes that the factors are independently distributed, which makes the likelihood too concentrated. Hence the variance of the noise is underestimated as indicated by the red curves in Figure \ref{fig:shared_cov_pred_mean}. In comparison, the variance of the noise estimated by the GPPCA is more accurate, which makes predictions by the GPPCA closer to the truth. 





\begin{example}[Factors with different covariance matrices]
The data are sampled from model (\ref{equ:model_1}) where $x_i=i$ for $1\leq i\leq n$. The variance of the noise is $\sigma^2_0=0.25$ and the kernel function is assumed to follow from (\ref{equ:matern_5_2}) with $\sigma^2=1$. The range parameter $\gamma$ of each factor is uniformly sampled from $[10,10^3]$ in each experiment.  We simulate the data from $8$ different combinations of $k$, $d$ and $n$. In each scenario, we repeat $N=100$ times. The parameters in the kernels and the variance of the noise are all estimated from the data.
\label{eg:diff_cov}
\end{example}

\begin{figure}[t]
\centering
  \begin{tabular}{cccc}
    \includegraphics[width=.245\textwidth,height=.35\textwidth]{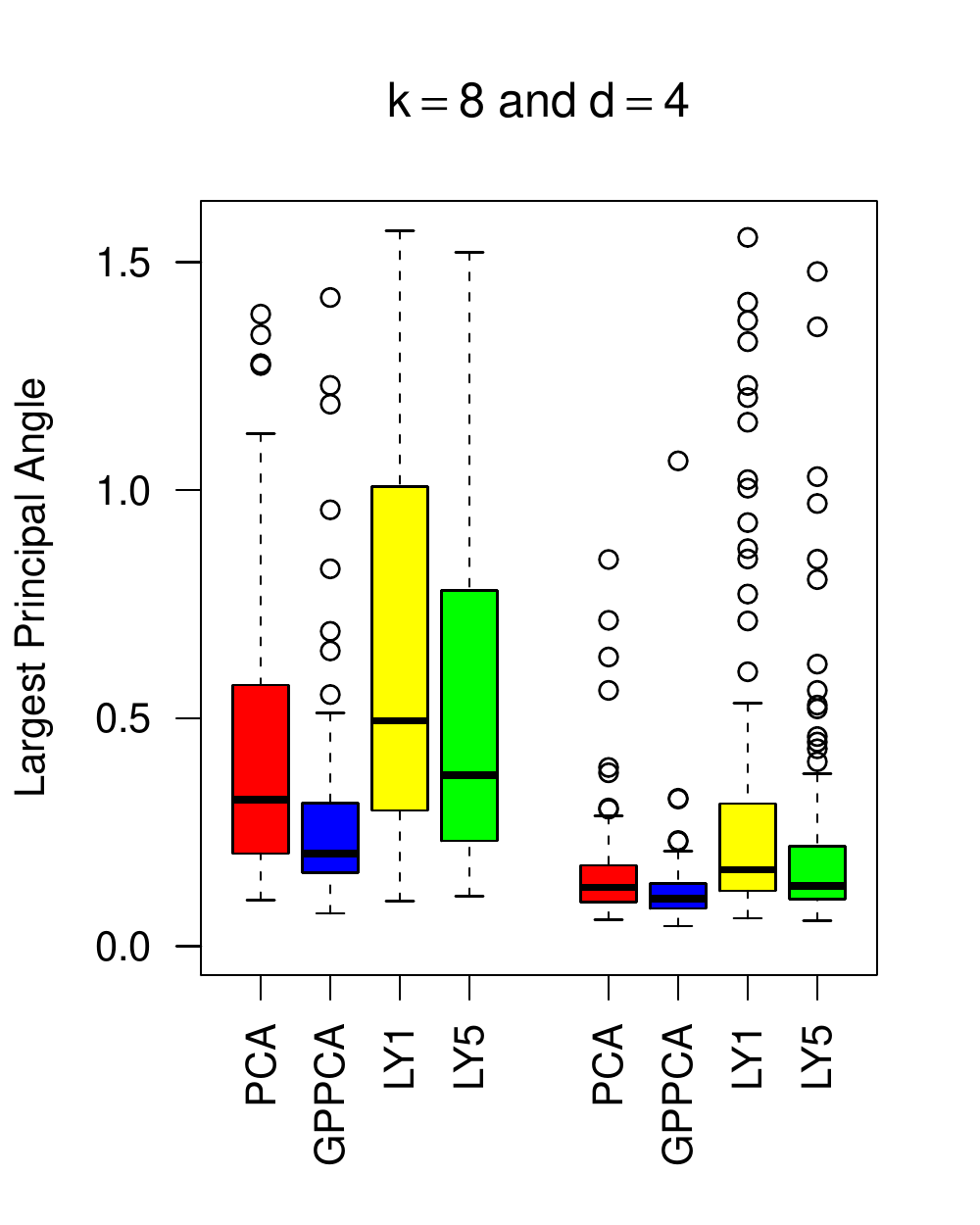}
        \includegraphics[width=.245\textwidth,height=.35\textwidth]{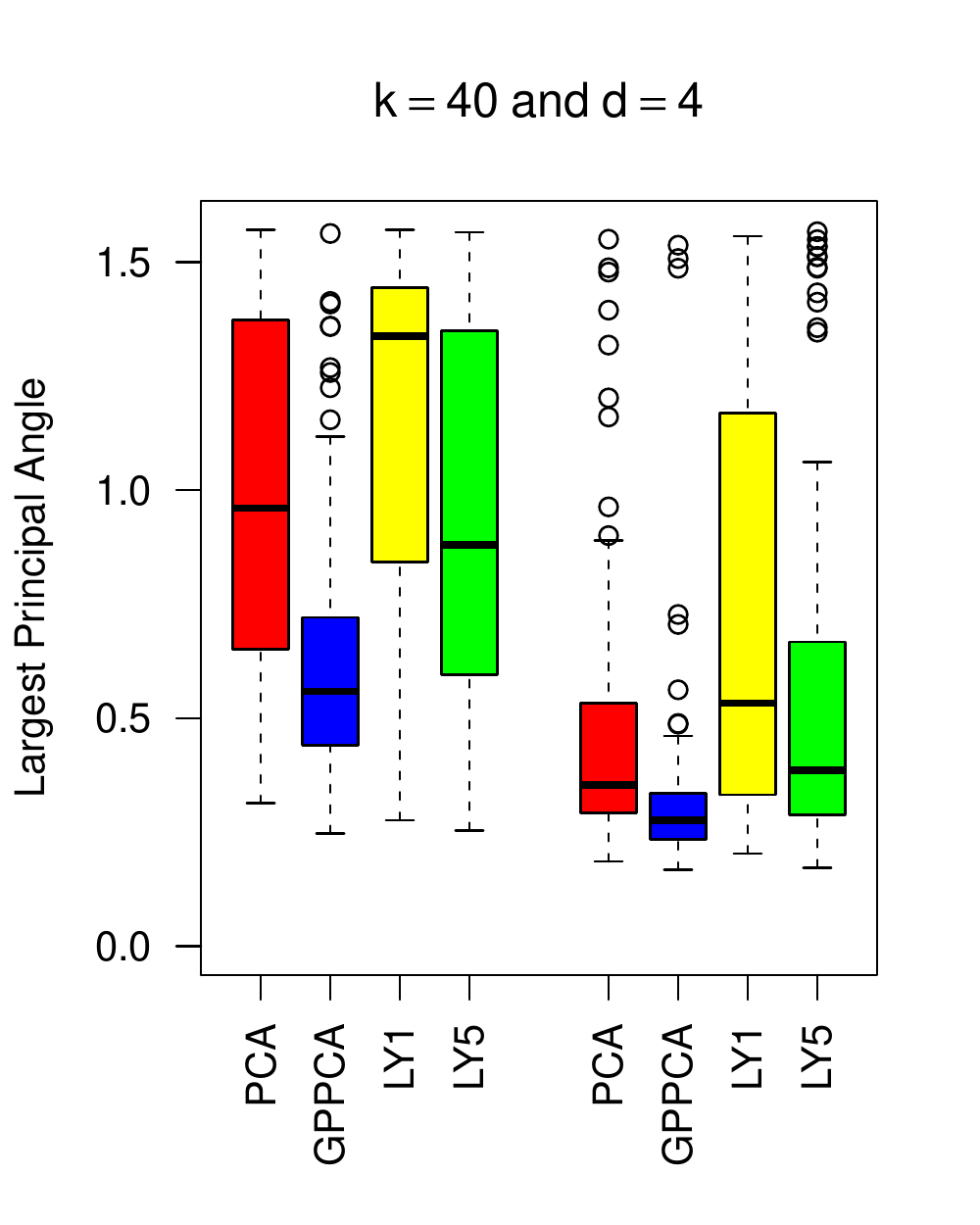}
            \includegraphics[width=.245\textwidth,height=.35\textwidth]{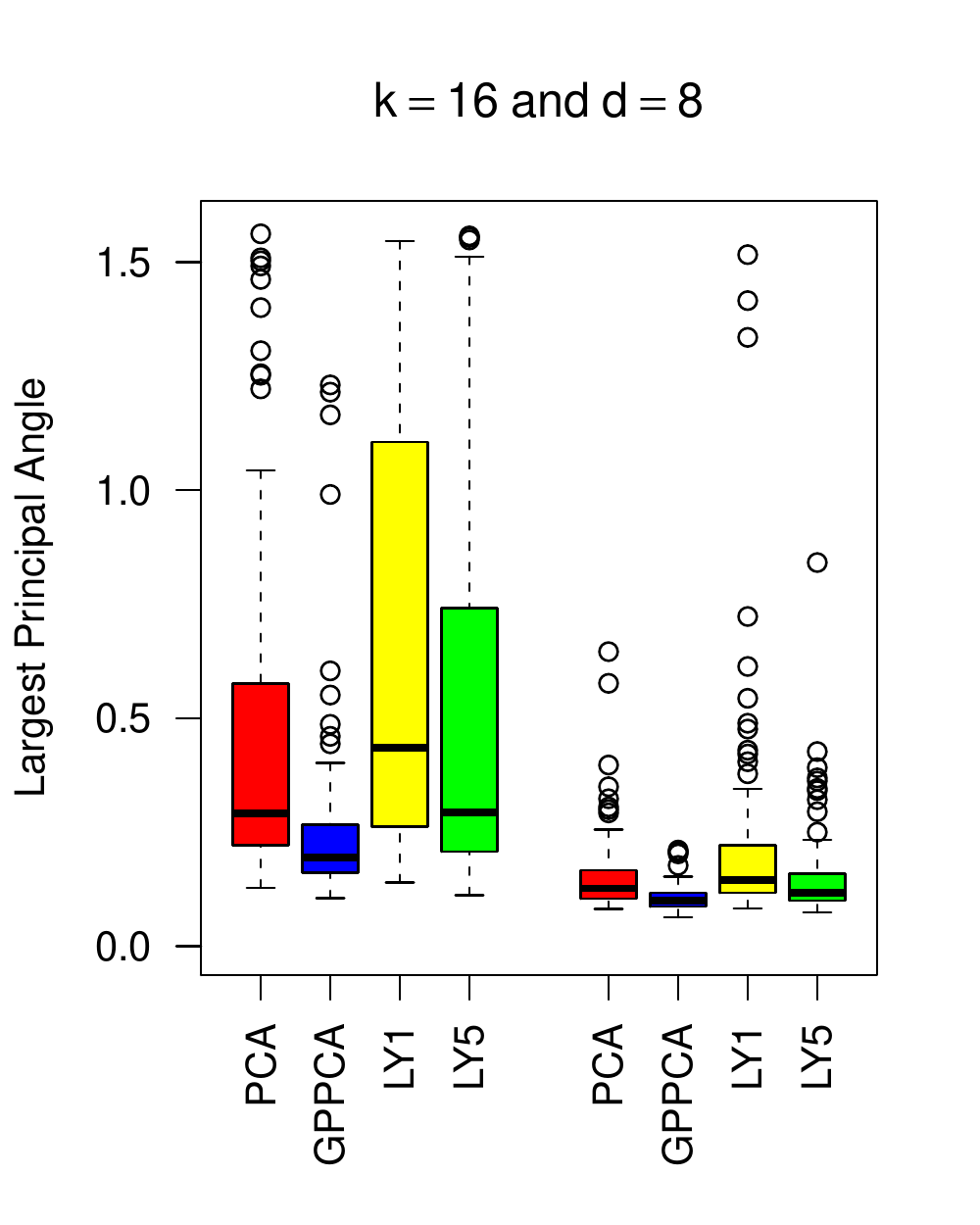}
        \includegraphics[width=.245\textwidth,height=.35\textwidth]{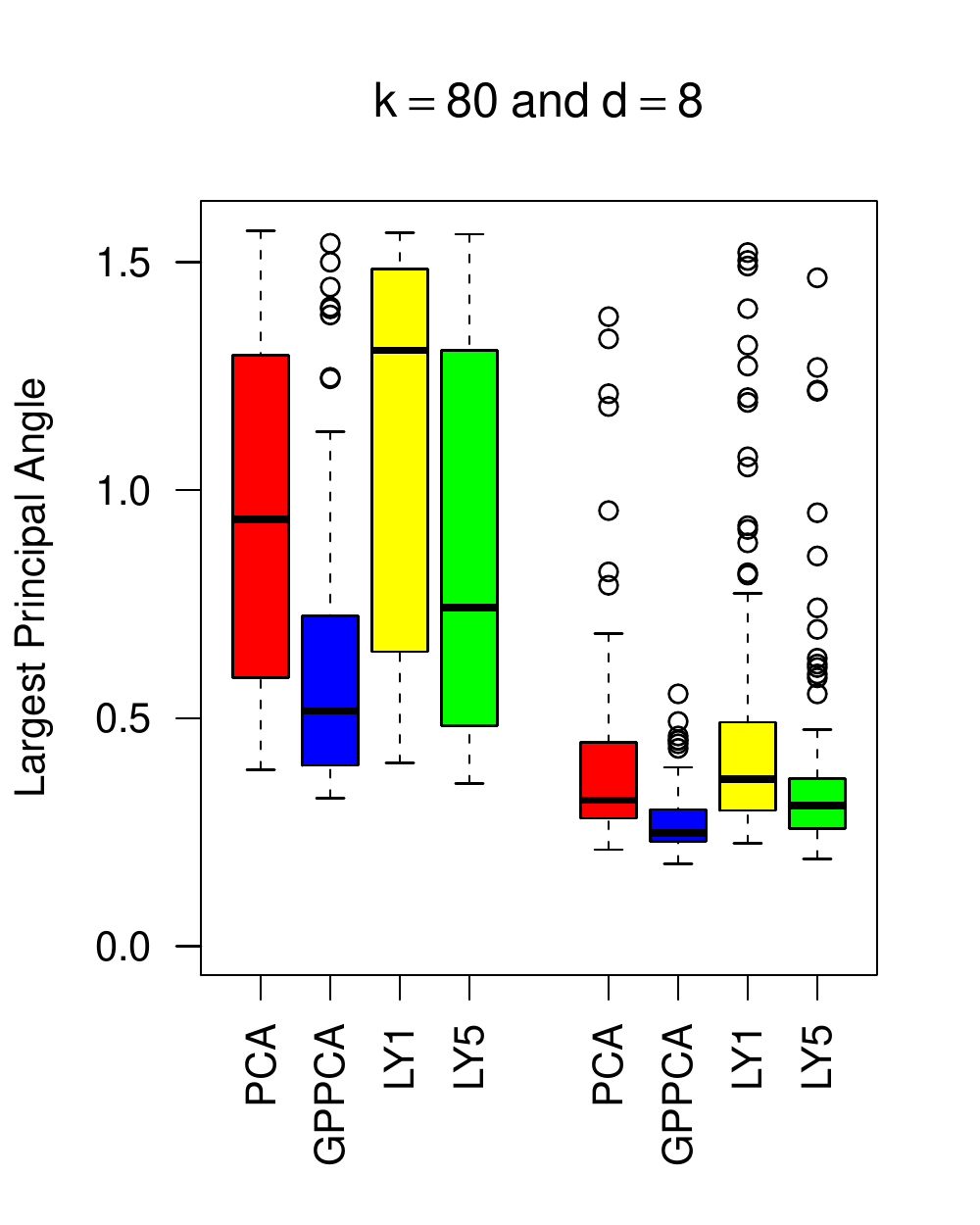}
  \end{tabular}
  \vspace{-.1in}
   \caption{The largest principal angle between the true subspace and the estimated subspace of the four approaches for Example \ref{eg:diff_cov}. The number of observations of each output variable is $n=200$ and $n=400$ for left 4 boxplots and right 4 boxplots in 2 left panels, respectively. The number of observations is $n=500$ and $n=1000$ for left 4 boxplots and right 4 boxplots in 2 right panels, respectively. }
\label{fig:simulation_diff_cov_angles}
\end{figure}

\begin{table}[t]
\begin{center}
\begin{tabular}{lcccc}
  \hline

 $d=4$ and $\tau=4$         &  \multicolumn{2}{c}{k=8} & \multicolumn{2}{c}{k=40}\\
 &  $n=200$ &  $n=400$ & $n=200$ &  $n=400$  \\
  \hline
  PCA            &$1.3\times 10^{-1}$& $1.3\times 10^{-1}$ &$3.8\times 10^{-2}$& $3.0\times 10^{-2}$  \\
  GPPCA          &$\bf 1.4\times 10^{-2}$ & $\bf 4.0\times 10^{-2}$ &$\bf 7.1\times 10^{-3}$ & $\bf 1.1\times 10^{-2}$ \\
    LY1 & $1.6\times 10^{-1}$ & $1.4 \times 10^{-1}$ &$4.9\times 10^{-2}$ &$3.4\times 10^{-2}$ \\
  LY5 &  $1.5\times 10^{-1}$ & $1.3\times 10^{-1}$ & $4.4\times 10^{-2}$&$3.2\times 10^{-2}$\\
  \hline
 $d=8$ and $\tau=4$        &  \multicolumn{2}{c}{k=16} & \multicolumn{2}{c}{k=80}\\
 &  $n=500$ &  $n=1000$ & $n=500$ &  $n=1000$  \\
  \hline
  PCA            &$1.3\times 10^{-1}$& $1.3\times 10^{-1}$ &$3.5\times 10^{-2}$& $2.9\times 10^{-2}$  \\
  GPPCA          & $\bf 1.3\times 10^{-2}$ & $\bf 3.3\times 10^{-2}$ &$\bf 6.0\times 10^{-3}$ & $\bf 8.0\times 10^{-3}$ \\
      LY1 & $1.4\times 10^{-1}$ & $1.3 \times 10^{-1}$ &$3.7\times 10^{-2}$ &$2.9\times 10^{-2}$ \\
  LY5 &  $1.4\times 10^{-1}$ & $1.3\times 10^{-1}$ & $3.4\times 10^{-2}$&$2.8\times 10^{-2}$\\

  \hline

\end{tabular}
\end{center}
   \caption{AvgMSE for Example \ref{eg:diff_cov}.}

   \label{tab:AvgMSE_diff_cov}
\end{table}



Since the covariance matrices are different in Example \ref{eg:diff_cov}, we implement the numerical optimization algrithm on the Stiefel manifold \citep{wen2013feasible} to estimate $\mathbf A$ in Theorem \ref{thm:est_A_diff_cov}. The largest principal angle between $\mathcal M(\mathbf A)$ and  $\mathcal M(\hat {\mathbf A})$ and the AvgMSE in estimating the mean of the output matrix  by different approaches for Example \ref{eg:diff_cov} is given in Figure \ref{fig:simulation_diff_cov_angles} and Table \ref{tab:AvgMSE_diff_cov}, respectively. The estimation by the GPPCA outperforms the other methods based  on both criteria.

\section{Real Data Examples}
\label{sec:real_eg}
We apply the proposed GPPCA approach and compared its performance with other approaches on two real data applications in this section.

\subsection{Emulating multivariate output of the computer models}
We first apply GPPCA for emulating computer models with multivariate output. Computer models or simulators have been developed and used in various scientific, engineering and social applications. Some simulators are computationally expensive (as the numerical solution of a system of the partial different equations (PDEs) is often required and is slow), and some contain multivariate outputs at a set of the input parameters (see e.g. \cite{higdon2008computer,paulo2012calibration,fricker2013multivariate,Gu2016PPGaSP}). Thus, a statistical emulator is often required to approximate the behavior of the simulator. 

We consider the testbed called the `diplomatic and military operations in a non-warfighting domain' (DIAMOND) simulator \citep{taylor2004development}. The DIAMOND simulator models the number of casualties during the second day to sixth day after the earthquake and volcanic eruption in Giarre and Catania. The simulator has 13 input variables, such as the helicopter cruise speed, engineer ground speed,  shelter and food supply capacity at the two places (see Table 1 in \cite{overstall2016multivariate} for a complete list of the input variables). 

We use the same $n=120$ training and $n^*=120$ test outputs in \cite{overstall2016multivariate} to compare different methods. We focus on the out-of-sample prediction criteria:
\begin{eqnarray}
\text{RMSE}&=&\sqrt{\frac{\sum^{k}_{j=1}\sum^{n^*}_{i=1}(\hat Y^*_j(\mathbf x^*_i )-  Y^*_j(\mathbf x^*_i  ))^2 }{kn^*}},\, \label{equ:RMSE} \\
{P_{CI}(95\%)} &=& \frac{1}{k{n^{*}}}\sum\limits_{j = 1}^{k} {\sum\limits_{i = 1}^{n^{*}} 1\{Y^*_j(\mathbf x^*_i  )\in C{I_{ij}}(95\% )\}}\,,\label{equ:PCI} \\
{L_{CI}(95\%)} &=& \frac{1}{{k{n^{*}}}}\sum\limits_{j = 1}^{k} \sum\limits_{i = 1}^{{n^{*}}} {\Length\{C{I_{ij}}(95\% )\} } \,, \label{equ:LCI}
\end{eqnarray}
where  $Y^*_j(\mathbf x^*_i)$ is the $j$th coordinate of the held-out test output vector at the  $i$th  test input $\mathbf x^*_i$ for  $1\leq i\leq n^*$ and $1\leq j\leq k^*$. $C{I_{ij}}(95\% )$ is the $95\%$ predictive credible interval and $\Length\{C{I_{ij}}(95\% )\}$ is the length of the $95\%$ predictive credible interval of the $Y^*_j(\mathbf x^*_i)$. A method with a small out-of-sample RMSE, ${P_{CI}(95\%)}$ being close to nominal $95\%$ level, and  a small $L_{CI}(95\%)$ is considered precise in prediction and uncertainty quantification.

\begin{table}[t]
\begin{center}
\begin{tabular}{lcccccc}
  \hline

Method & Mean function & Kernel &  RMSE &   ${P_{CI}(95\%)} $ &  ${L_{CI}(95\%)} $  \\
  \hline
 GPPCA  & Intercept& Gaussian kernel            &$3.33\times 10^{2}$& $0.948$ &$1.52\times 10^{3}$  \\
  GPPCA &Selected covariates& Gaussian kernel  &$3.18\times 10^{2}$& $0.957$ &$1.31\times 10^{3}$  \\
 GPPCA &Intercept& Mat{\'e}rn kernel           &$2.82\times 10^{2}$& $0.962$ &$1.22\times 10^{3}$  \\
  GPPCA &Selected covariates& Mat{\'e}rn kernel  &$2.74\times 10^{2}$& $0.957$ &$1.18\times 10^{3}$  \\
\hline
    Ind GP  & Intercept& Gaussian kernel            &$3.64\times 10^{2}$& $0.918$ &$1.18\times 10^{3}$  \\
   Ind GP &Selected covariates& Gaussian kernel  &$4.04\times 10^{2}$& $0.918$ &$1.17\times 10^{3}$  \\
  Ind GP &Intercept& Mat{\'e}rn kernel           &$3.40\times 10^{2}$& $0.930$ &$0.984\times 10^{3}$  \\
   Ind GP &Selected covariates& Mat{\'e}rn kernel  &$3.31\times 10^{2}$& $0.927$ &$0.967\times 10^{3}$  \\
  \hline
    Multi GP  & Intercept & Gaussian kernel            &$3.63\times 10^{2}$& $0.975$ &$1.67\times 10^3$   \\
   Multi GP & Selected covariates & Gaussian kernel  &$3.34\times 10^{2}$& $0.963$ &  $1.54 \times  10^3$ \\
      Multi GP & Intercept & Mat{\'e}rn kernel  &$3.01\times 10^{2}$& $0.962$ &  $1.34 \times  10^3$ \\
      Multi GP & Selected covariates & Mat{\'e}rn kernel  &$3.05\times 10^{2}$& $0.970$ &  $1.50 \times  10^3$ \\

\hline

\end{tabular}
\end{center}
   \caption{Emulation of the DIAMOND simulator by different models. The first four rows show the predictive performance by the GPPCA with different mean structure and kernels. The middle four rows give the predictive performance by Ind GP with the same mean structure and kernels, as used in the GPPCA. The 9th and 10th rows show the emulation result of two best models in \cite{overstall2016multivariate} using Gaussian kernel for the same held-out test output, whereas the last two rows give the result of the same model with the Mat{\'e}rn kernel in (\ref{equ:matern_5_2}).  The RMSE is $1.08\times 10^5$ using the mean of the training output to predict. }

   \label{tab:humanity_model}
\end{table}

We compare the prediction performance of the GPPCA, the independent Gaussian processes (Ind GP) and multivariate Gaussian process (Multi GP) on the held-out test output. The Ind GP builds a GP to emulate each coordinate of the output vector separately. The Multi GP in \cite{overstall2016multivariate} proposes a separable model, where the covariance of the output is a Kronecker product of the covariance matrix of the output vector at the same input, and the correlation matrix of the any output variable at different inputs. The parameters of Multi GP are estimated by the MLE using the code provided in \cite{overstall2016multivariate} and the parameters in Ind GP are estimated by the posterior mode using ${\tt RobustGaSP}$ ${\sf R}$ package \citep{gu2018robustgasp}. 


We use a product kernel for all models where each kernel is assumed the same for each input dimension. The Gaussian kernel is assumed in \cite{overstall2016multivariate} and we also include results using the Mat{\'e}rn kernel in (\ref{equ:matern_5_2}) for comparison. In \cite{overstall2016multivariate}, the model with the least RMSE is the one using the Gaussian kernel and a set of selected covariates. We find the 11th input (food capacity in Catania) is positively correlated with the outputs. Thus for the GPPCA and Ind GP, we explore the predictive performance of the models with the mean basis function being $\mathbf h(\mathbf x)=(1, x_{11})$. For GPPCA, we assume the range parameters in the kernels are shared for the latent factor processes, while the variance parameters are allowed to be different.

 \begin{figure}[t]
\centering
  \begin{tabular}{cc}
  \raisebox{5mm}[0pt][0pt]{
   \includegraphics[width=.49\textwidth,height=.31\textwidth]{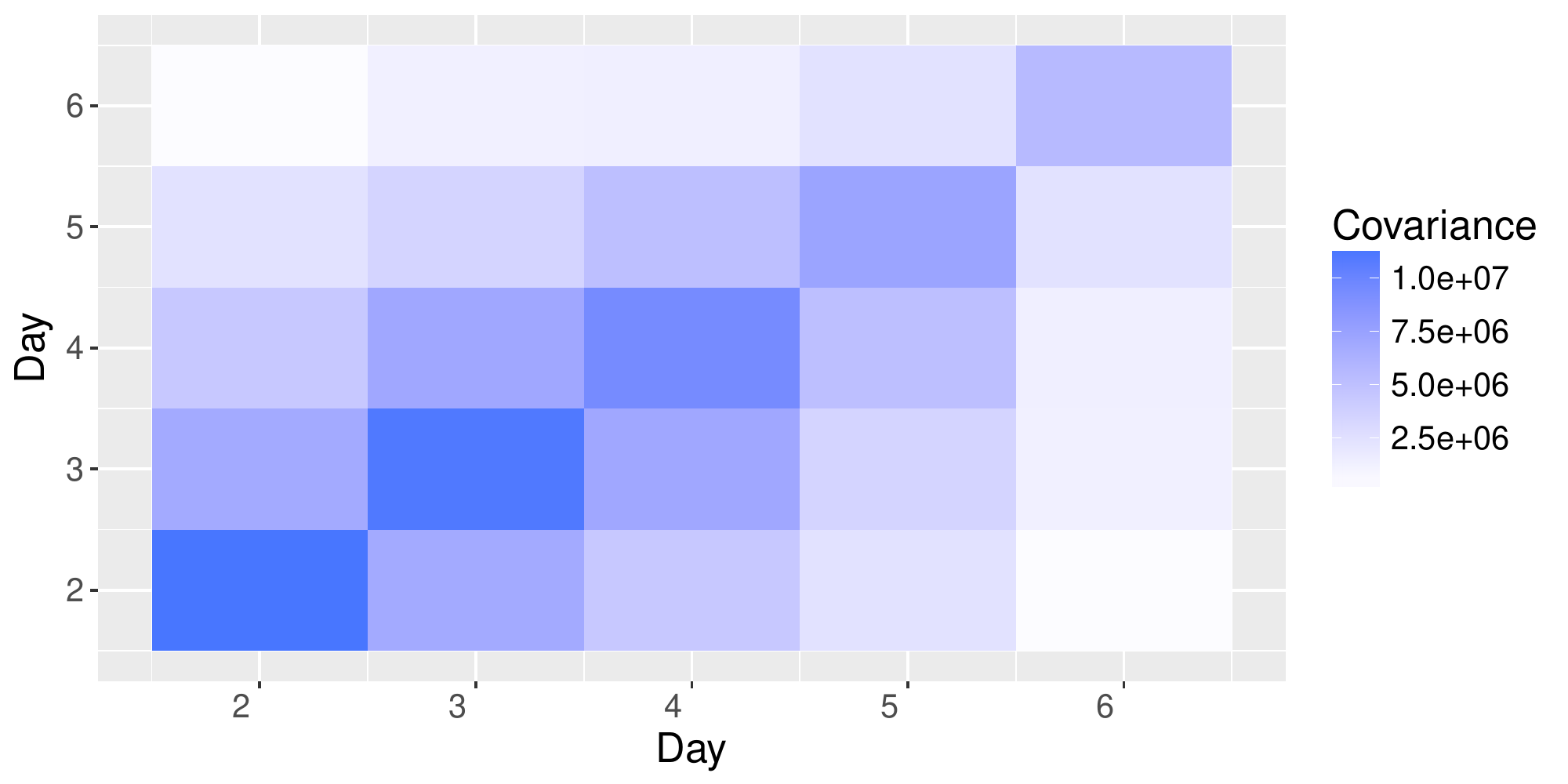}
   }
      \includegraphics[width=.49\textwidth,height=.4\textwidth]{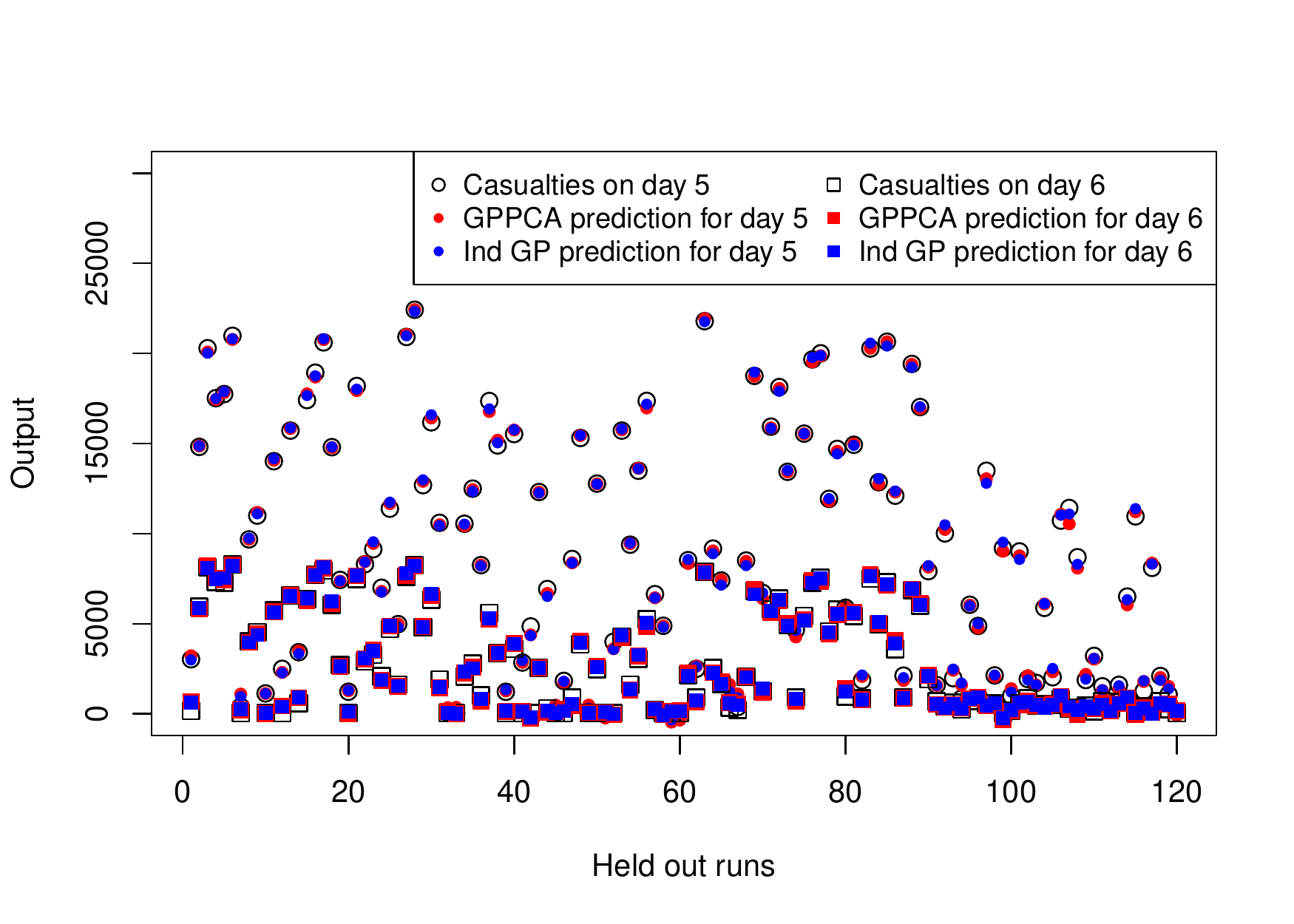}
    
  \end{tabular}
  \vspace{-.1in}
   \caption{The estimated  covariance of the casualties at the different days after the catastrophe by the GPPCA is graphed in the left panel.  The held-out test output, the prediction by the GPPCA and Independent GPs with the mean basis $\mathbf h(\mathbf x)=(1,x_{11} )$ and Mat{\'e}rn kernel for the fifth day and sixth day are graphed in the right panel.   }
\label{fig:humanity_model}
\end{figure}

The predictive RMSE of different models are shown in Table \ref{tab:humanity_model}. Overall, all three approaches are precise in prediction, as the predictive RMSE is less than $1\%$ of the RMSE using the mean to predict. Compared to the other two approaches, the GPPCA has the smallest out-of-sample RMSE on each combination of the kernel function and mean function among three approaches. The nominal $95\%$ predictive interval covers around $95\%$ of held-out test output with relatively short average length of the predictive interval. The predictive interval from Multi GP covers more than $95\%$ of the held-out test output, but the average length of the interval is the highest. The Ind GP has the shortest length of the  predictive interval, but it covers less than  $95\%$ of the held-out test output using any kernel or mean function. The held-out test output on the fifth and sixth day and the prediction by Ind PG and GPPCA are graphed in the right panel in Figure \ref{fig:humanity_model}, both of which seem to be accurate. 
 






In GPPCA,  the estimated covariance matrix of the casualties at the different days is $\mathbf {\hat A}  \bm {\hat   \Lambda}  \mathbf {\hat A} + \hat \sigma^2_0 \mathbf I_k $, where $\bm {\hat   \Lambda}$ is a diagonal matrix where the $i$th term is  $\hat \sigma^2_i$ (the estimated variance of  the $i$th factor).  This covariance matrix is shown in the left panel in Figure \ref{fig:humanity_model}. We found that the estimated covariance between any two days is positive. This is sensible as the short food capacity, for example, is associated with the high casualties for all following days after the catastrophe. We also noticed that the estimated correlation of  the output at the two consecutive days is larger, though we do not enforce a time-dependent structure (such as the autoregressive model in \cite{liu2009dynamic,farah2014bayesian}). The GPPCA is a more general model as the output does not have to be time-dependent, and the estimated covariance between the output variables captures the time dependence in the example. 






\subsection{Gridded temperature}
\label{subsec:gridded_temperature}
In this subsection, we consider global gridded temperature anomalies from U.S. National Oceanic and Atmospheric Administration (NOAA), available at:

\href{ftp://ftp.ncdc.noaa.gov/pub/data/noaaglobaltemp/operational}{ftp://ftp.ncdc.noaa.gov/pub/data/noaaglobaltemp/operational}

\noindent This dataset records the global gridded monthly anomalies of the  air and marine temperature from Jan 1880 to near present with $5^{\circ}\times 5^{\circ}$ latitude-longitude  resolution \citep{SShen2017}. 

A proportion of the temperature measurements is missing in the data set, which is also a common scenario in other climate data set. As many scientific studies may rely on the full data set, we first compare different approaches on  interpolation,  using the monthly temperature anomalies at $k=1,639$ spatial grid boxes in the past 20 years. We hold out the $24,000$ randomly sampled measurements on $k^*=1,200$ spatial grid boxes in $n^*=20$ months as the test data set. The rest $15,336$ measurements are used as the training data. We evaluate the interpolation performance of different methods based on the RMSE, ${P_{CI}(95\%)} $, and ${L_{CI}(95\%)}$ on the test data set.

\begin{table}[t]
\begin{center}
\begin{tabular}{lcccccc}
  \hline
Method  & measurement error &  RMSE &   ${P_{CI}(95\%)} $ &  ${L_{CI}(95\%)} $  \\
  \hline
     GPPCA, $d=50$  &estimated    &$0.386$& $0.870$ &$1.02$  \\
  GPPCA, $d=100$  & estimated           &$0.320$& $0.772$ &$0.563$  \\
 GPPCA, $d=50$  &fixed    &$0.385$& $0.933$ &$1.33$  \\
  GPPCA, $d=100$  & fixed           &$0.314$& $0.977$ &$1.44$  \\
  \hline
    PPCA, $d=50$  &estimated    &$0.620$& $0.677$ &$1.08$  \\
  PPCA, $d=100$  & estimated           &$0.602$& $0.525$ &$0.803$  \\
 PPCA, $d=50$  &fixed    &$0.617$& $0.765$ &$1.32$  \\
  PPCA, $d=100$  & fixed           &$0.585$& $0.819$ &$1.400$  \\
\hline
 Temporal model   &estimated    &$0.937$& $0.944$ &$2.28$  \\
  Spatial model &estimated    &$0.560$& $0.942$ &$2.23$  \\
  Spatio-temporal model &  estimated   &$0.492$ &$0.957$ & $2.10$ \\
\hline
 Temporal regression by RF  &estimated    &$0.441$ & /&/   \\
 Spatial regression by RF  &estimated    &$0.391$& / & /  \\
\hline

\end{tabular}
\end{center}
   \caption{Out of sample prediction of the temperature anomalies by different approaches. The first four rows give the predictive performance by the GPPCA with different latent factors, estimated and fixed variance of the measurement error, whereas the latter four rows record the results by the PPCA. The predictive performance by the temporal, spatial and spatio-temporal smoothing methods  are given in the 9th and 10th rows. The last two rows give the predictive RMSE by regression using the random forest (RF) algorithm. }

   \label{tab:temperature}
\end{table}

The  predictive performance by the GPPCA using the predictive distribution in (\ref{equ:cond_pred_with_mean}) is shown in the first four rows of Table \ref{tab:temperature}. Here the number of grid boxes is $k=1639$, and the temporal correlation of the temperature measurements at different months are parameterized by the Mat{\'e}rn kernel in (\ref{equ:matern_5_2}). We model the intercept and monthly change rate at each location by assuming the mean basis function $\mathbf h(x)=(1,x)$, where $x$ is an integer from $1$ to $240$ to denote the month of an observation. We explore the cases with $d=50$ and $d=100$ latent factor processes where the covariance in each latent process is assumed to be the same.  In this dataset, the average recorded variance of the measurement error is around $0.1$. We implement the scenarios with an estimated variance or a fixed variance of the measurements. In the fifth to the eighth rows, we show the predictive performance of the PPCA with the same number of latent factors. In the ninth and tenth rows, we show the results by a spatial model and a temporal model both based on the  Mat{\'e}rn kernel, separately for the observations in each spatial grid box and in each month, respectively. The ${\tt RobustGaSP}$ ${\sf R}$ package \citep{gu2018robustgasp} is used to fit the GP regression with the estimated nuggets, and the mean basis function is assumed to be  $\mathbf h(x)=(1,x)$ when fitting GP regression for the monthly measurements.  The predictive performance by a spatio-temporal model that use a product Mat{\'e}rn kernel function is  shown in the eleventh row. In the last two rows in Table \ref{tab:temperature}, we consider two regression schemes based on the random forest algorithm (\cite{breiman2001random}). The first scheme treats the observations in each spatial grid box as independent measurements, whereas the second scheme treats the observations in each month as independent measurements. The  modeling fitting details of these approaches are given in Appendix D. 


First, we find that GPPCA has the lowest out-of-sample RMSE among all the methods we considered. When the number of factors increases, both the PPCA and GPPCA seem to perform better in terms of RMSEs. However, the estimation by the GPPCA is more precise. This is because the temporal correlation and linear trend are modeled and estimated in the GPPCA, whereas the PPCA is a special case of GPPCA with the independent monthly measurements. This result is achieved with the simplest setting in GPPCA, that is when the covariance of the factor processes is assumed to be the same. In this case, the estimation of the factor loadings has a closed form expression. Assuming different parameters in the factor processes and use other kernel functions may further improve the precision in prediction. Furthermore, when the variance of the measurement error is estimated, the predictive credible interval by either the PPCA or GPPCA is too short, resulting in less than $95\%$ of the data covered by $95\%$ predictive interval. When the variance of the noise is fixed to be $0.1$ (the variance of the measurement error), around $95\%$ of the held-out data are covered in the nominal $95\%$ predictive interval in the GPPCA, but not in the PPCA.



\begin{figure}[t]
\centering
  \begin{tabular}{ccc}
      \includegraphics[width=.33\textwidth,height=.3\textwidth]{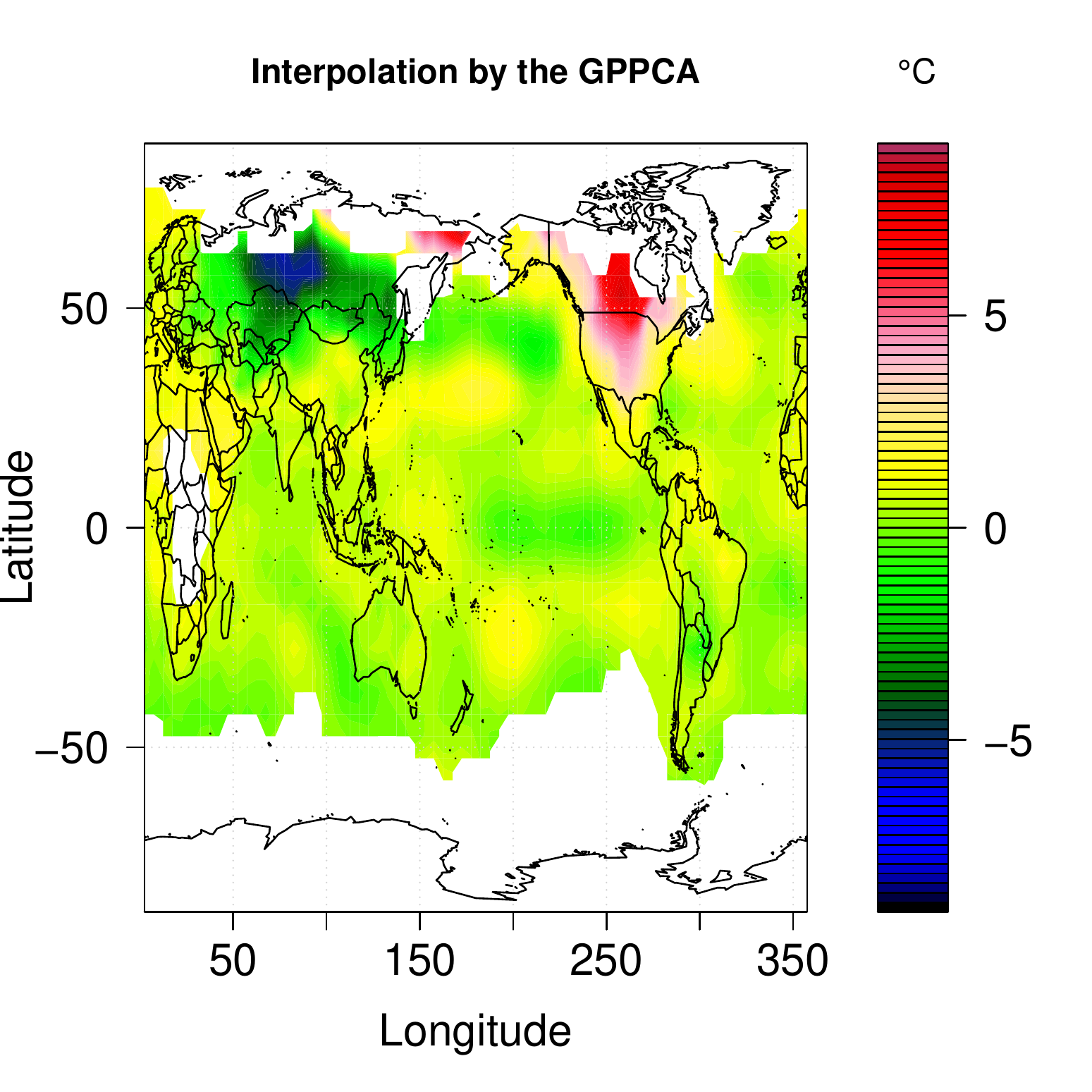} 
    \includegraphics[width=.33\textwidth,height=.3\textwidth]{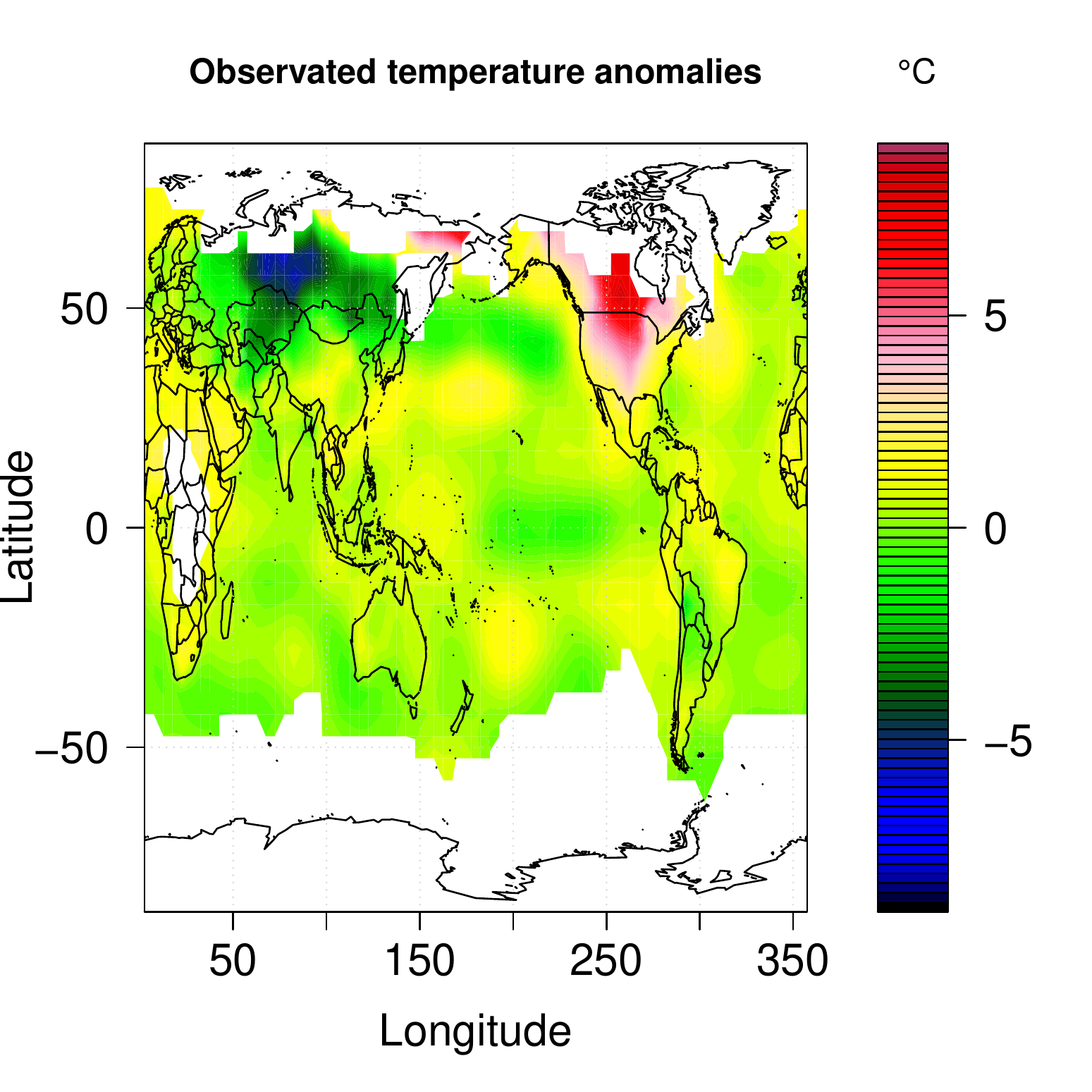} 
          \includegraphics[width=.33\textwidth,height=.3\textwidth]{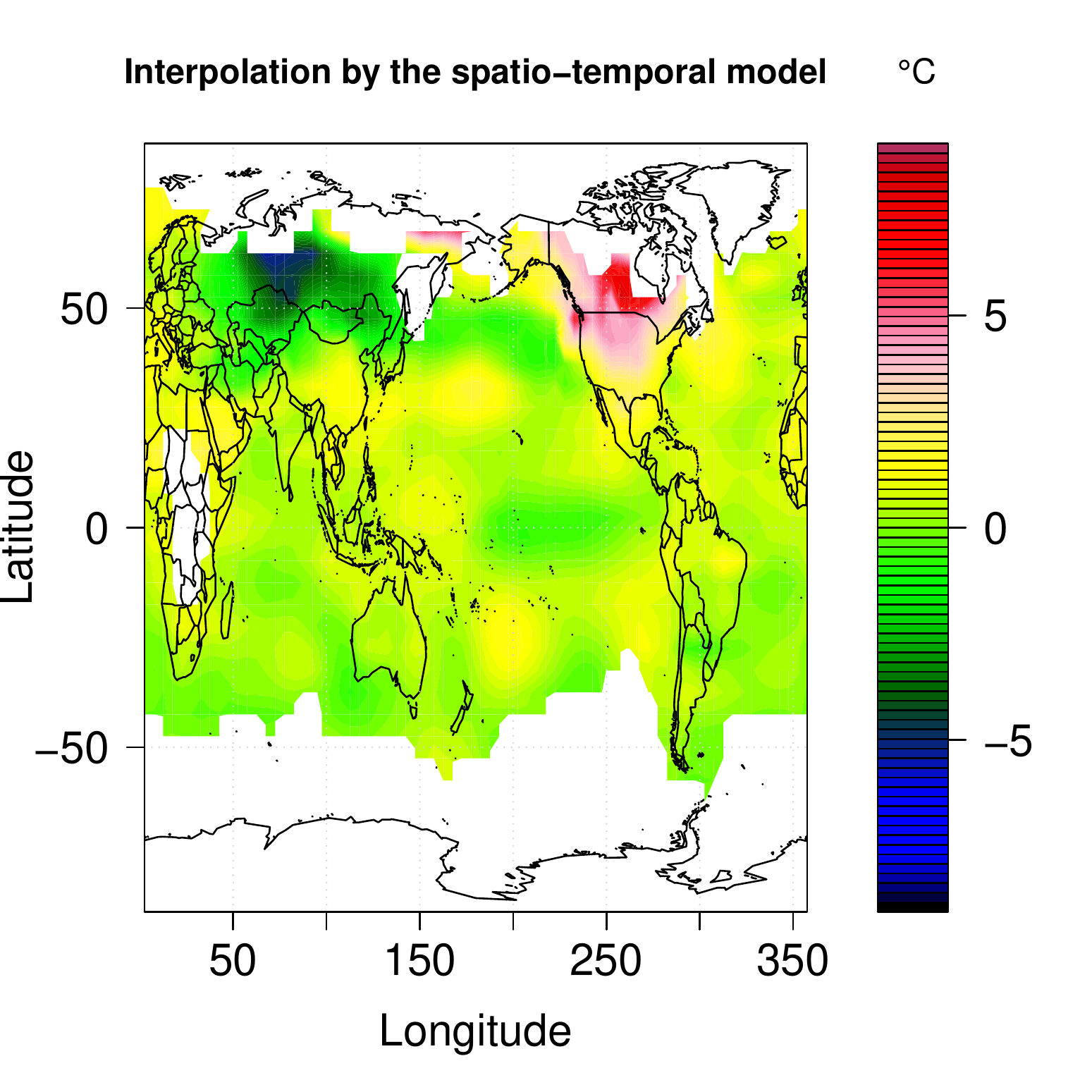} 
  \end{tabular}
  \vspace{-.1in}
   \caption{Interpolation of the temperature anomalies in November 2016. The real temperature anomalies in November 2016 is graphed in the middle panel. The interpolated temperature anomalies by the GPPCA and spatio-temporal model are graphed in the left and right panels, respectively. The number of training  and test observations are 439 and 1200, respectively. The out-of-sample RMSE of the GPPCA and spatio-temporal model is $0.314$ and $0.747$, respectively. }
\label{fig:temperature_2016_Nov_prection}
\end{figure}

The  spatial smoothing approach by GP and spatial regression by RF have smaller predictive errors than its temporal counterparts, indicating the spatial correlation may be larger than the temporal correlation in the data. Combining both the spatial and temporal information seems to be more accurate than using only the spatial or temporal information. However, the spatio-temporal model  is not as accurate as the GPPCA. We plot the interpolated temporal anomalies in November 2016 by the GPPCA (with the variance of the measurement error fixed to be 0.1) and the spatio-temporal model in the left and right panels in Figure   \ref{fig:temperature_2016_Nov_prection}, respectively. Compared with the observed temperature anomalies shown in the middle panel, the GPPCA interpolation is more precise than the spatial smoothing method at the locations where the temperature anomalies changes rapidly, e.g. the region between the U.S. and Canada, and the east region in Russia. We should acknowledge that the implemented spatio-temporal model is not the only choice. Other spatio-temporal models may be applicable, yet fitting these models may be more computationally expensive.

Note that the missing values are typically scattered in different rows and columns of the observation matrix in practice. One of the future directions is to extend the GPPCA to include the columns of the data matrix with missing values to improve the estimation of the factor loading matrix and the predictive distribution of the missing values, based on 
expectation-maximization algorithm, or the Markov chain Monte Carlo algorithm if one can specify the full posterior distributions of the factor loading matrix and the parameters. Besides, We should also emphasize that we do not utilize the spatial distance in the GPPCA. This makes the GPPCA suitable for other interpolation and matrix completion tasks when there is no distance information between the output variables. 








\begin{figure}[t]
\centering
  \begin{tabular}{cc}
      \includegraphics[width=.48\textwidth,height=.43\textwidth]{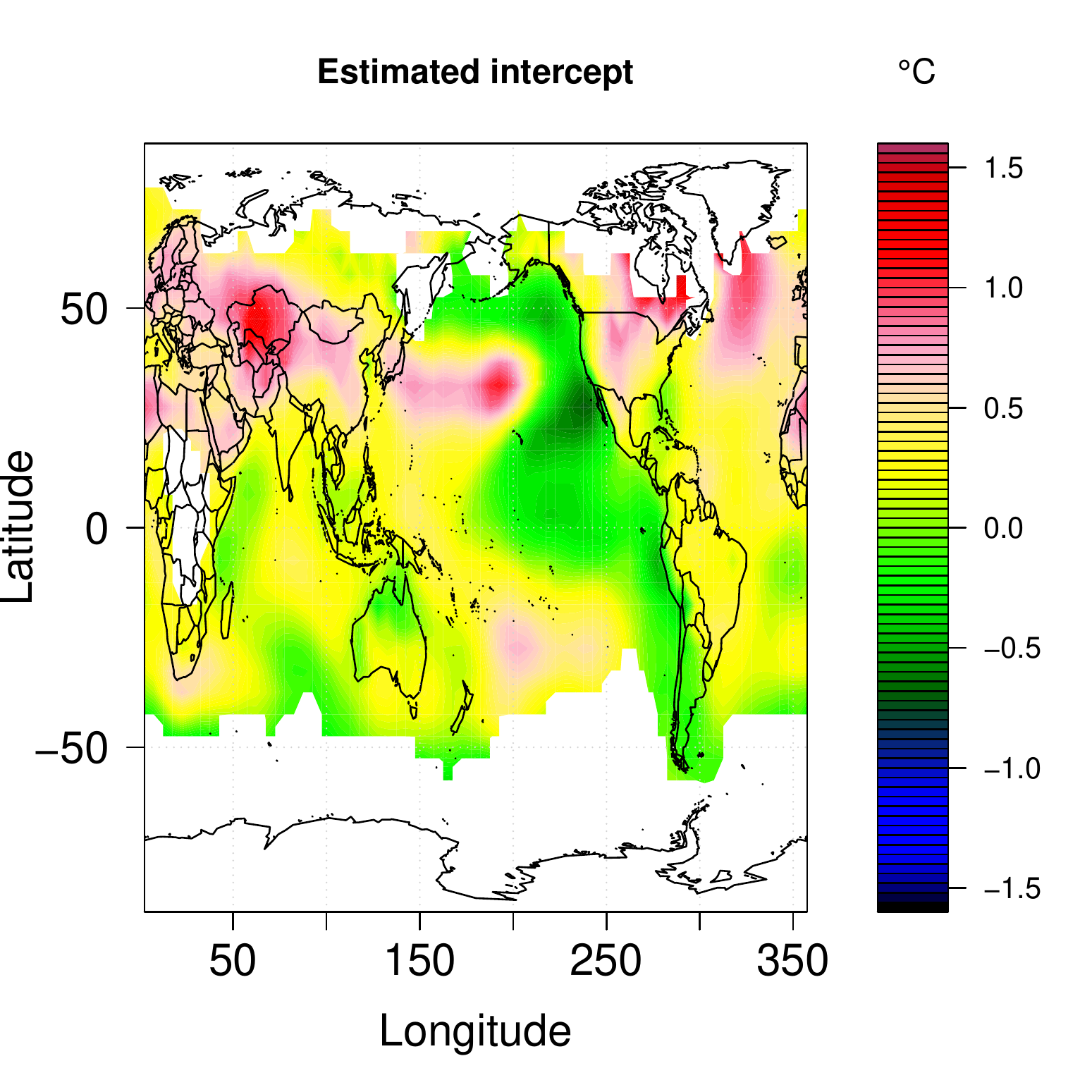} 
    \includegraphics[width=.48\textwidth,height=.43\textwidth]{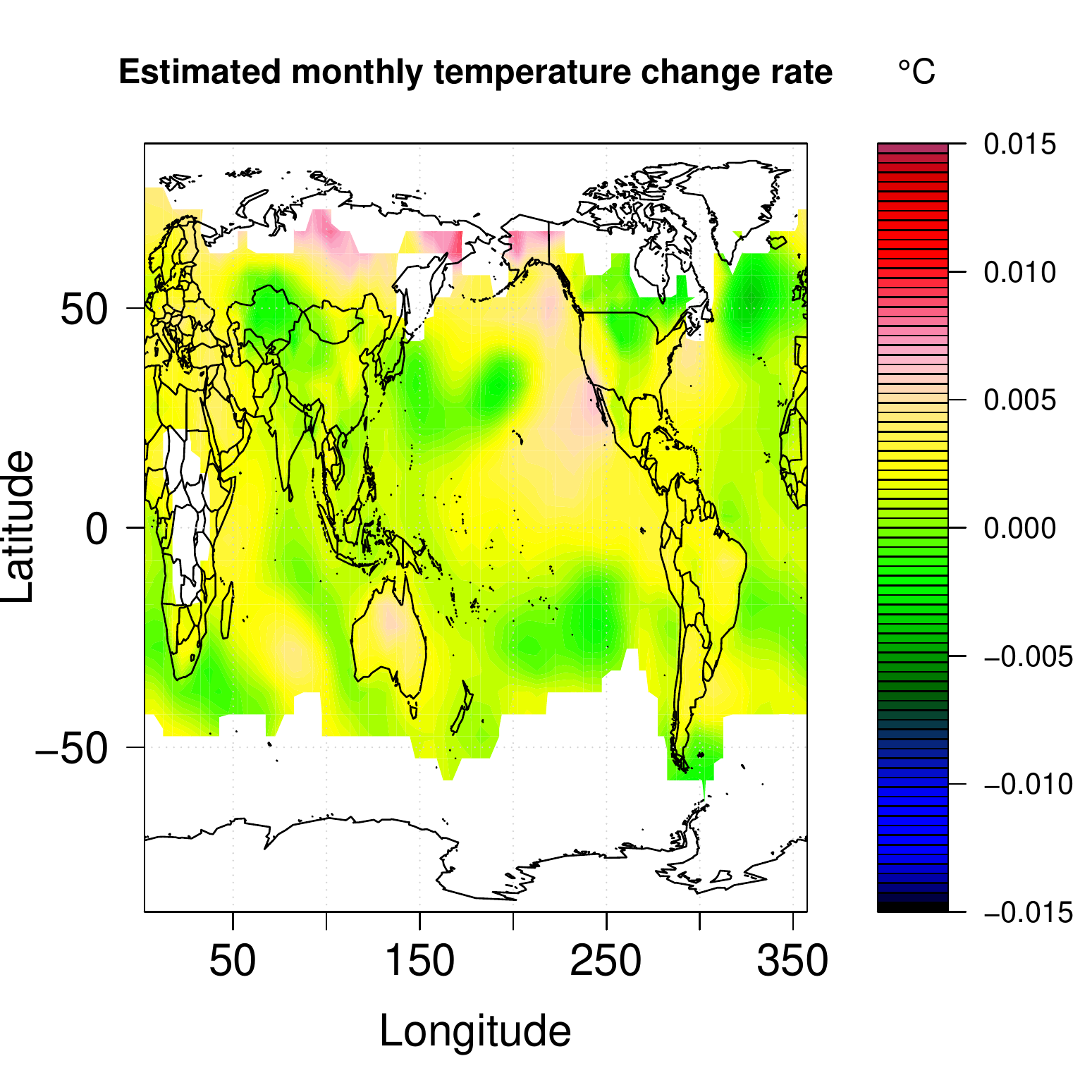} 

  \end{tabular}
  \vspace{-.1in}
   \caption{Estimated intercept and monthly change rate of the temperature anomalies by the GPPCA using the monthly temperature anomalies between January 1999 and December 2018.  }
\label{fig:temperature_April_theta_est}
\end{figure}

  The estimated trend parameters $\bm {\hat \Theta}$ by the GPPCA are shown in Figure \ref{fig:temperature_April_theta_est}. Based on the last twenty years' data, the average annual increase of the temperature is at the rate of around $0.02 ~^o C$. The areas close to the north pole seems to have the most rapid increase rate. Among the rest of the areas, the south west part and the north east part of the U.S. also seem to increase slightly faster than the other areas. Note we only use the observations from the past 20 years for demonstration purpose. A study based on a longer history of measurements may give a clearer picture of the change in global temperature. 
  






\section{Concluding remarks}
\label{sec:conclusion}
In this paper, we have introduced the GPPCA, as an extension of the PPCA for the latent factor model with the correlated factors. By allowing data to infer the covariance structure of the factors,  the estimation of the factor loading matrix and the predictive distribution of the output variables both become more accurate by the GPPCA, compared to the ones by the  PCA and other approaches. This work also highlights the scalable computation achieved by a closed-form expression of the inverse covariance matrix in the marginal likelihood. In addition, we extend our approach to include additional covariates in the mean function and we manage to marginalize out the regression parameters to obtain a closed-form expression of the marginal likelihood when estimating the  factor loading matrix. 



There are several future directions related to this work. First of all, the factor loading matrix, as well as other parameters in the kernel functions and the variance of the noise, is estimated by the maximum marginal likelihood estimator, where the uncertainty in the parameter estimation is not expressed in the predictive distribution of the output variables. A full Bayesian approach may provide a better way to quantify the uncertainty in the predictive distribution. Secondly, we assume the number of the latent factors is known in this work. A consistent way to identify the number of latent factors is often needed in practice. Thirdly, the convergence rate of the predictive distribution and the estimation of the subspace of the factor loading matrix of the GPPCA both need to be explored. The numerical results shown in this work seem to be encouraging towards this direction.  Furthermore, when the covariances of the factor processes are not the same, the numerical optimization algorithm that preserves the orthogonal constraints \citep{wen2013feasible} is implemented for the marginal maximum likelihood estimator of the factor loading matrix. The convergence of this algorithm is an interesting direction to explore. A fast algorithm for the optimization problem in Theorem \ref{thm:est_A_diff_cov} will also be crucial for some computationally intensive applications. Finally, here we use kernels to parameterize the covariance of the factor processes for demonstrative purposes. The GPPCA automatically apply to many other models of the latent factors, as long as the likelihood of a factor follows a multivariate normal distribution.  It is interesting to explore the GPPCA in other factor models and applications.

\acks{The authors thank the editor and three anonymous referees for their comments that substantially improve this article. Shen's research is partially supported by the Simons Foundation Award 512620 and the National Science Foundation (NSF DMS 1509023).}


\newpage

\appendix
\section*{Appendix A: Auxiliary facts}
\begin{enumerate}[1.]
\item Let $\mathbf A$ and $\mathbf B$ be matrices, 
\[ (\mathbf A \otimes \mathbf B)^T=  (\mathbf A^T  \otimes \mathbf B^T); \]
further assuming $\mathbf A$ and $\mathbf B$ are invertible, 
\[ (\mathbf A \otimes \mathbf B)^{-1}=  \mathbf A^{-1} \otimes  \mathbf B^{-1}. \]
\label{item:matrix_kronecker_trans_inv}
\item Let $\mathbf A$, $\mathbf B$, $\mathbf C$ and $\mathbf D$ be the matrices such that the products $\mathbf A\mathbf C$ and $\mathbf B\mathbf D$ are matrices,  
\[(\mathbf A \otimes \mathbf B)(\mathbf C\otimes \mathbf D)= (\mathbf A \mathbf C) \otimes  (\mathbf B \mathbf D).\]
\label{item:matrix_kronecker_otimes}
\item For matrices $\mathbf A$,  $\mathbf B$ and  $\mathbf C$, 
\begin{align*}
(\mathbf C^T \otimes \mathbf A) \mbox{vec}(\mathbf B)= \mbox{vec}(\mathbf A \mathbf B \mathbf C);
\end{align*}
further assuming $\mathbf A^T \mathbf B$ is a matrix, 
\begin{align*}
\tr(\mathbf A^T \mathbf B )= \mbox{vec}(\mathbf A)^T \mbox{vec}(\mathbf B).
\end{align*}
\label{item:vectorization}
\item For any invertible $n \times n$ matrix $\mathbf C$,
\[|\mathbf C+\mathbf A \mathbf B|=|\mathbf C||\mathbf I_n + \mathbf B \mathbf C^{-1} \mathbf A |. \]
\label{item:determinant}
\end{enumerate}
\section*{Appendix B: Proofs}
We first give some notations for the vectorization used in the proofs. 
Let $\mathbf A_v=[\mathbf I_n \otimes \mathbf a_1,..., \mathbf I_n \otimes \mathbf a_d]$ and $\mathbf Z_{vt}= \mbox{vec}(\mathbf Z^T)$. Let $\bm \Sigma_v$ be a $nd\times nd$ block diagonal matrix where the $l$th diagonal block is $\bm \Sigma_l$, for $l=1,...,d$.

\begin{proof}[Proof of Equation (\ref{equ:Y_v_cov})]
Vectorize the observations in model (\ref{equ:model_1}), one has 
\[\mathbf Y_v=  \mathbf A_v \mathbf Z_{vt} +\bm \epsilon_v \]
where $\mathbf Z_{vt}\sim \N(\mathbf 0, \, \bm \Sigma_v )$ and $\bm \epsilon_v\sim N(\mathbf 0, \sigma^2 \mathbf I_{nk})$. Using the fact \ref{item:matrix_kronecker_trans_inv} and fact \ref{item:matrix_kronecker_otimes},  $\mathbf A_v \mathbf Z_{vt} \sim \MN(\mathbf 0, \, \bm \Sigma_{A_v Z_{vt}} )$, where 
\begin{equation}
\bm \Sigma_{A_v Z_{vt}} = \mathbf A_v \bm \Sigma_v \mathbf A^T_v = [\bm \Sigma_1 \otimes \mathbf a_1,..., \bm \Sigma_d \otimes \mathbf a_d]  \mathbf A^T_v = \sum^{d}_{l=1} \bm \Sigma_l \otimes (\mathbf a_l \mathbf a^T_l)
  \label{equ:Sigma_AZ}
\end{equation}
 for $l=1,...,d$. Marginalizing out $\mathbf Z_{vt}$, one has Equation (\ref{equ:Y_v_cov}). 
\end{proof}

\begin{proof}[Proof of Lemma~\ref{lemma:Y_v_inv_cov}]
By  (\ref{equ:Y_v_cov}) and (\ref{equ:Sigma_AZ}), one has 
\[\mathbf Y_v \mid \mathbf A, \sigma^2_0, \bm \Sigma_{1},..., \bm \Sigma_{d} \sim \MN\left(\mathbf 0, \,  \mathbf A_v \bm \Sigma_v \mathbf A^T_v+ \sigma^2_0 \mathbf I_{nk} \right).\] 
The precision matrix is
\begin{align*}
&(\mathbf A_v \bm \Sigma_v \mathbf A^T_v+ \sigma^2_0 \mathbf I_{nk})^{-1} \\
=& \sigma^{-2}_0\mathbf I_{nk}-\mathbf A_v \frac{( \sigma^2_0 \bm \Sigma^{-1}_v+ \mathbf A^T_v \mathbf A_v )^{-1}}{\sigma^2_0} \mathbf A^T_v \\
=&\sigma^{-2}_0\mathbf I_{nk}-\mathbf A_v \frac{( \sigma^2_0 \bm \Sigma^{-1}_v+ \mathbf I_{nd} )^{-1}}{\sigma^2_0} \mathbf A^T_v \\
=&\sigma^{-2}_0 \left\{ \mathbf I_{nk}- [ (\sigma^2_0 \bm \Sigma^{-1}_1+ \mathbf I_{n})^{-1} \otimes \mathbf a_1,...,(\sigma^2_0 \bm \Sigma^{-1}_d+ \mathbf I_{n})^{-1} \otimes \mathbf a_d ] \mathbf A^T_v \right\}\\
=& \sigma^{-2}_0  \left( \mathbf I_{nk}- \sum^{d}_{l=1}  (\sigma^2_0 \bm \Sigma^{-1}_l+ \mathbf I_{n})^{-1} \otimes \mathbf a_l \mathbf a^T_l \right)
\end{align*}
where the first equality follows from the Woodbury identity; the second equality is by Assumption \ref{assumption:A}; the third equality is by fact \ref{item:matrix_kronecker_otimes}; and the four equality is by  fact \ref{item:matrix_kronecker_trans_inv} and fact \ref{item:matrix_kronecker_otimes}, from which the results follow immediately. 
\end{proof}

\begin{proof}[Proof of Theorem ~\ref{thm:est_A_shared_cov}]

When $\bm \Sigma_1=...=\bm \Sigma_d=\bm \Sigma$, by the fact \ref{item:vectorization}, the likelihood of $\mathbf A$ is 
\begin{align*}
L(\mathbf A \mid \mathbf Y,  \sigma^2_0,  \bm \Sigma)&\propto \exp\left(- \frac{\mathbf { Y}^T_v\left( \mathbf I_{nk}-  (\sigma^2_0 \bm \Sigma^{-1}+ \mathbf I_{n})^{-1} \otimes \sum^{d}_{l=1} \mathbf a_l \mathbf a^T_l \right) \mathbf { Y}_v}{2\sigma^{2}_0} \right)\\
&\propto \exp\left(- \frac{ \mathbf { Y}^T_v \mathbf { Y}_v - \mathbf { Y}^T_v \mbox{vec}( \mathbf A \mathbf A^T\mathbf { Y}\bm (\sigma^2_0 \bm \Sigma^{-1}+ \mathbf I_{n})^{-1} ) }{2\sigma^{2}_0}\right) \\
&\propto \etr\left(- \frac{ \mathbf { Y}^T \mathbf { Y} - \mathbf { Y}^T \mathbf A \mathbf A^T\mathbf { Y}\bm (\sigma^2_0 \bm \Sigma^{-1}+ \mathbf I_{n})^{-1}  }{2\sigma^{2}_0 } \right) \\
&\propto \etr\left(- \frac{ \mathbf { Y}^T \mathbf { Y} -  \mathbf A^T\mathbf { Y}\bm (\sigma^2_0 \bm \Sigma^{-1}+ \mathbf I_{n})^{-1} \mathbf { Y}^T \mathbf A }{2\sigma^{2}_0 } \right), 
\end{align*}
where $\etr(\cdot):=\exp(\tr(\cdot))$. 

Maximizing the likelihood as a function of $\mathbf A$ is equivalent to  the optimization problem: 
\begin{equation}
\mathbf {\hat A}= \argmax_{\mathbf A} \mbox{tr}(\mathbf A^T \mathbf G \mathbf A  ) \quad s.t. \quad \mathbf A^T \mathbf A=\mathbf I_d,
\label{equ:trace_opt}
\end{equation}
where $\mathbf G=  \mathbf Y  \left( \mathbf I_n +\sigma^2_0\bm \Sigma^{-1} \right)^{-1}  \mathbf Y^T $. This optimization in (\ref{equ:trace_opt}) is a trace optimization problem (\cite{kokiopoulou2011trace}). By the Courant-Fischer-Weyl min-max principal (\cite{saad1992numerical}), $\mbox{tr}(\mathbf A^T \mathbf G \mathbf A  ) $ is maximized when $\mathbf {\hat A}=\mathbf U \mathbf R$, with $\mathbf U$ being the orthonormal basis of the eigenspace  associated with the $d$ largest eigenvalue of $\mathbf G$ and $ \mathbf R$ is any arbitrary rotation matrix. In this case, $\mbox{tr}(\mathbf {\hat A}^T \mathbf G  \mathbf {\hat A} )= \mbox{tr} (\mathbf U \bm \Lambda \mathbf U^T)=\sum^{d}_{l=1} \lambda_l$, where $\bm \Lambda$ is a diagonal matrix of the $d$ largest eigenvalue $\lambda_l$  of $\mathbf G $, for $l=1,...,d$. 
\end{proof}

\begin{proof}[Proof of Theorem ~\ref{thm:est_A_diff_cov}]
Under Assumption \ref{assumption:A}, by fact \ref{item:vectorization}, the likelihood for $\mathbf A$ is
\begin{align*}
L(\mathbf A \mid \mathbf Y, \sigma^2_0, \bm \Sigma_1,..., \bm \Sigma_d)&\propto \exp\left(- \frac{\mathbf {\ Y}^T_v\left( \mathbf I_{nk}- \sum^{d}_{l=1}  (\sigma^2_0 \bm \Sigma^{-1}_l+ \mathbf I_{n})^{-1} \otimes \mathbf a_l \mathbf a^T_l  \right) \mathbf {\ Y}_v}{2\sigma^{-2}_0} \right)\\
&\propto\etr\left(-\frac{\mathbf {\ Y}^T \mathbf {\ Y}- \mathbf {\ Y}^T\sum^d_{l=1} \mathbf a_l \mathbf a^T_l   \mathbf {\ Y} (\sigma^2_0 \bm \Sigma^{-1}_l+ \mathbf I_{n})^{-1} }{2 \sigma^2_0}\right)\\
&\propto\etr\left(-\frac{\mathbf {\ Y}^T \mathbf {\ Y}- \sum^d_{l=1}  \mathbf a^T_l   \mathbf {\ Y} (\sigma^2_0 \bm \Sigma^{-1}_l+ \mathbf I_{n})^{-1} \mathbf {\ Y}^T \mathbf a_l }{2 \sigma^2_0}\right),
\end{align*}
from which the result follows. 

\end{proof}

\begin{proof}[Proof of Equation (\ref{equ:profile_lik})]


From the proof of Theorem ~\ref{thm:est_A_diff_cov}, one has
\begin{align}
L(\sigma^2_0, \mid \mathbf Y, \bm \Sigma_1,..., \bm \Sigma_d,\mathbf A)&\propto (\sigma^2_0)^{-nk/2} \etr\left(-\frac{\mathbf {\ Y}^T \mathbf {\ Y}- \sum^d_{l=1}  \mathbf a^T_l   \mathbf {\ Y} (\sigma^2_0 \bm \Sigma^{-1}_l+ \mathbf I_{n})^{-1} \mathbf {\ Y}^T \mathbf a_l }{2 \sigma^2_0}\right).
\label{equ:lik_sigma_2_0}
\end{align}
 Equation (\ref{equ:hat_sigma_2_0}) follows immediately by maximizing  (\ref{equ:lik_sigma_2_0}).

We now turn to show the profile likelihood in (\ref{equ:profile_lik}). Under Assumption \ref{assumption:A}
\begin{align}
&p(\mathbf Y \mid \bm \tau, \bm \gamma,  \mathbf { A}, \sigma^2_0  ) \nonumber \\
=&\int p(\mathbf Y \mid  \mathbf { A}, \sigma^2_0,\mathbf Z  ) p(\mathbf Z \mid \bm \tau, \bm \gamma) d\mathbf Z \nonumber\\
=& \int (2\pi \sigma^2_0)^{-\frac{nk}{2} } \etr\left(-\frac{ (\mathbf Y- \mathbf A \mathbf Z)^T (\mathbf Y- \mathbf A \mathbf Z)  }{\sigma^2_0} \right) (2\pi)^{-\frac{nd}{2} }    \prod^d_{l=1} |\bm \Sigma_l|^{-\frac{1}{2} } \exp\left(-\frac{1}{2} \sum^d_{l=1} \mathbf Z^T_l \bm \Sigma^{-1}_l  \mathbf Z_l\right)  d\mathbf Z \nonumber\\
=&  (2\pi \sigma^2_0)^{-\frac{nk}{2}  }   \prod^d_{l=1} \left|\bm \Sigma_l/\sigma^2_0+ \mathbf I_k \right|^{-1/2} \exp\left(-\frac{S^2}{2\sigma^2_0} \right) \label{equ:last_line_pf}
\end{align}
where $S^2=\tr(  \mathbf {Y}^T \mathbf {Y} )-\sum^d_{l=1} \mathbf { a}^T_l \mathbf {Y}  (  \tau^{-1}_l\mathbf R^{-1}_l  +\mathbf I_n)^{-1} \mathbf {Y}^T  \mathbf {a}_l$. Equation (\ref{equ:profile_lik}) follows by plugging $\mathbf {\hat A}$ and $\hat \sigma^2_0$ into (\ref{equ:last_line_pf}).  
\end{proof}

\begin{proof}[Proof of Theorem~\ref{thm:prediction}]
Denote the parameters $ \hat {\bm \theta}=({ \bm {\hat \gamma},  \mathbf { \hat A},  \bm {\hat  \sigma}^2, \hat \sigma^2_0})$. Denote $\bm \hat \Sigma$ as the estimated $\bm \Sigma_v$ by plugging the estimated parameters.  We first compute the posterior distribution of $(\mathbf Z_{vt} \mid \mathbf Y_v, \, \bm {\hat \theta})$. From Equation (\ref{equ:Y_v_cov}), 
\begin{align*}
p(\mathbf Z_{vt} \mid \mathbf Y_v, \bm {\hat \theta} )&\propto \exp\left( \frac{ (\mathbf Y_v -\mathbf {\hat A}_v \mathbf Z_{vt} )^T (\mathbf Y_v -\mathbf {\hat A}_v \mathbf Z_{vt})}{2\hat \sigma^2_0}\right) \exp\left( -\frac{1}{2} \mathbf Z^T_{vt} \bm {\hat \Sigma}^{-1}_v \mathbf Z_{vt} \right) \\
&\propto \exp\left\{ -\frac{1}{2}(\mathbf Z_{vt} -\mathbf {\hat Z}_{vt} )^T \left( \frac{\mathbf {\hat A}^T_v \mathbf {\hat A}_v}{\hat \sigma^2_0 }+\bm {\hat \Sigma}_v^{-1} \right) (\mathbf Z_{vt} -\mathbf {\hat Z}_{vt})\right\},
\end{align*}
where $\mathbf {\hat Z}_{vt} =  ({\mathbf {\hat A}^T_v \mathbf {\hat A}_v} + \hat \sigma^2_0 \bm {\hat \Sigma_v}^{-1})^{-1}  \mathbf {\hat A}^T_v \mathbf Y_v$ from which we have
\begin{equation}
\mathbf Z_{vt} \mid \mathbf Y_v, \bm {\hat \theta} \sim \mbox{MN}\left(\mathbf {\hat Z}_{vt} , \,  \left(\frac{\mathbf {\hat A}^T_v \mathbf {\hat A}_v}{\hat {\sigma}^2_0 }+\bm {\hat \Sigma_v}^{-1}\right)^{-1} \right).
\label{equ:Z_posterior} 
\end{equation}

Note $\mathbf {\hat A}^T_v \mathbf {\hat A}_v=\mathbf I_{nd} $. Using fact \ref{item:matrix_kronecker_otimes} and fact \ref{item:vectorization},  one has 
\begin{align}
\mathbf {\hat Z}_{vt}&= \left( {\begin{array}{*{20}{c}}
\left(\hat \sigma^2_0 \bm {\hat  \Sigma}^{-1}_1 +{\mathbf I_n}\right)^{-1}\otimes  \mathbf {\hat a}^T_1  \\
\vdots \\
\left(\hat \sigma^2_0 \bm {\hat  \Sigma}^{-1}_d +{\mathbf I_n}\right)^{-1} \otimes  \mathbf {\hat a}^T_d \\
 \end{array} } \right) \mbox{vec}({\mathbf Y})
 = \left( {\begin{array}{*{20}{c}}
\mbox{vec}\left(\mathbf {\hat a}^T_1 \mathbf Y \left(\hat \sigma^2_0 \bm {\hat  \Sigma}^{-1}_1 +{\mathbf I_n} \right)^{-1}\right)  \\
\vdots \\
\mbox{vec}\left(\mathbf {\hat a}^T_d \mathbf Y \left(\hat \sigma^2_0 \bm {\hat  \Sigma}^{-1}_d +{\mathbf I_n}\right)^{-1}\right)   \\
 \end{array} } \right) \nonumber\\
 &=\mbox{vec} \left( {\begin{array}{*{20}{c}}
\mathbf {\hat a}^T_1 \mathbf Y \left(\hat \sigma^2_0 \bm {\hat  \Sigma}^{-1}_1 +{\mathbf I_n} \right)^{-1}  \\
\vdots \\
\mathbf {\hat a}^T_d \mathbf Y \left(\hat \sigma^2_0 \bm {\hat  \Sigma}^{-1}_d +{\mathbf I_n}\right)^{-1}   \\
 \end{array} } \right)^T :=\mbox{vec}(\mathbf {\hat Z}^T). 
 \label{equ:Z_hat_posterior}
 \end{align}

Now we are ready to derive the predictive mean and predictive variance. First 
\begin{align*}
\E[\mathbf Y(\mathbf x^*) \mid \mathbf Y,  \bm {\hat \theta}]&=\E [\E[ \mathbf Y(\mathbf x^*)  \mid \mathbf Y, \mathbf Z(\mathbf x^*), \bm {\hat \theta} ]] = \E[\mathbf {\hat A} \mathbf Z(\mathbf x^*) \mid \mathbf Y, \bm {\hat \theta}]\\
&= \mathbf {\hat A} \E[ \E[\mathbf Z(\mathbf x^*) \mid \mathbf Y, \mathbf Z, \bm {\hat \theta}]]=\mathbf {\hat A} \mathbf {\hat Z}(\mathbf x^*)
\end{align*} 
with the $l$th term of $\mathbf {\hat Z}(\mathbf x^*)$ 
\begin{align*}
 {\hat Z}_l(\mathbf x^*)&= \bm {\hat \Sigma}_l(\mathbf x^*) \bm \Sigma^{-1}_l \E[ \mathbf Z^T_l \mid \mathbf Y, \bm {\hat \theta}] \\
 &=\bm {\hat \Sigma}^T_l(\mathbf x^*) \bm {\hat \Sigma}_l^{-1} ( \bm {\hat \Sigma}_l^{-1}+\hat \sigma^2_0 \mathbf I_n)^{-1} \mathbf Y^T\mathbf {\hat a}_l  \\
  &=\bm {\hat \Sigma}^T_l(\mathbf x^*)  ( \hat \sigma^2_0 \mathbf I_n+\bm {\hat \Sigma}_l)^{-1} \mathbf Y^T \mathbf {\hat a}_l.
\end{align*}
where the first equality is from the property of multivariate normal distribution and the second equality is from (\ref{equ:Z_hat_posterior}). 

Secondly, we have
\begin{align*}
&\V[\mathbf Y^* \mid \mathbf Y,   \bm {\hat \theta} ] \\
=&  \E[\V[\mathbf Y^* \mid \mathbf Y,   \bm {\hat \theta}, \mathbf Z(\mathbf x^*)]]+\V[\E[\mathbf Y^* \mid \mathbf Y,   \bm {\hat \theta}, \mathbf Z(\mathbf x^*)]] \\
=& \hat \sigma^2_0 \mathbf I_k +   \V[ \mathbf {\hat A} \mathbf Z(\mathbf x^*) \mid \mathbf Y,   \bm {\hat \theta}] \\
=&\hat \sigma^2_0 \mathbf I_k + \mathbf {\hat A} [   \E[\V[\mathbf Z(\mathbf x^*) \mid \mathbf Y,   \bm {\hat \theta}, \mathbf Z]]+\V[\E[\mathbf Z(\mathbf x^*) \mid \mathbf Y,   \bm {\hat \theta}, \mathbf Z]] ]\mathbf {\hat A}^T=\hat \sigma^2_0 \mathbf I_k + \hat \sigma^2_0 \mathbf {\hat A} \mathbf {\hat D}(\mathbf x^*)  \mathbf {\hat A}^T   \\
\end{align*} 
with $\mathbf {\hat D}(\mathbf x^*)= \frac{1}{\hat \sigma^2_0}(\E[\V[\mathbf Z(\mathbf x^*) \mid \mathbf Y,   \bm {\hat \theta}, \mathbf Z]]+\V[\E[\mathbf Z(\mathbf x^*) \mid \mathbf Y,   \bm {\hat \theta}, \mathbf Z]])$.

Note that $\E[\V[\mathbf Z(\mathbf x^*) \mid \mathbf Y,   \bm {\hat \theta}, \mathbf Z]]$ is $k\times k$ diagonal matrix where the $l$th diagonal term is $\sigma^2_l\hat K_l(\mathbf x^*, \mathbf x^*) - \bm {\hat \Sigma}^T_l(\mathbf x^*)\bm{\hat \Sigma}^{-1}_l  \bm {\hat \Sigma}_l(\mathbf x^*)$, and $\V[\E[\mathbf Z(\mathbf x^*) \mid \mathbf Y,   \bm {\hat \theta}, \mathbf Z]]$ is another  $k\times k$ diagonal matrix where the $i$th diagonal term is $\sigma^2_l \hat K_l(\mathbf x^*, \mathbf x^*) - \bm {\hat \Sigma}^T_l(\mathbf x^*)( \hat \sigma^2_0\mathbf I_n+\bm{\hat \Sigma}_l)^{-1}  \bm {\hat \Sigma}_l(\mathbf x^*)    $. Thus, by the Woodbury matrix identity, $\mathbf {\hat D}(\mathbf x^*)$ is a diagonal matrix where the $i$th term is $\hat \sigma^2_l K_l(\mathbf x^*, \mathbf x^*) - \bm {\hat \Sigma}^T_l(\mathbf x^*)( \hat \sigma^2_0\mathbf I_n+\bm{\hat \Sigma}_l)^{-1}  \bm {\hat \Sigma}_l(\mathbf x^*)$ for $l=1,...,d$.

%

\end{proof}

\begin{proof}[Proof of Lemma \ref{lemma:Y_mean_structure}]
  Denote $\tilde {\mathbf M}= \mathbf M/\sigma^2_0$. Using the prior $\pi(\mathbf B)\propto 1$, we first marginalizing out $\mathbf B$ and the marginal density becomes 
\begin{align*}
p( \mathbf Y \mid \mathbf Z, \mathbf A,  \sigma^2_0,  \bm \Sigma_1,...,\bm \Sigma_d)\propto (\sigma^2_0)^{-k(n-q)/2} \etr \left(- \frac{(\mathbf Y- \mathbf A \mathbf Z)\mathbf {\tilde M} ( \mathbf Y- \mathbf A \mathbf Z)^T }{2} \right).
\end{align*}
Denote $\mathbf Y_{vt}=vec(\mathbf Y^T)$. By  Fact \ref{item:vectorization}, we have 
\begin{align}
&p( \mathbf Y, \mathbf Z \mid \mathbf A,  \sigma^2_0,  \bm \Sigma_{1:d} ) \nonumber \\
\propto & (\sigma^2_0)^{-\frac{k(n-q)}{2}} \prod^d\limits_{l=1} | \bm \Sigma_l|^{-\frac{1}{2}} \etr \left(- \frac{(\mathbf Y- \mathbf A \mathbf Z)\mathbf {\tilde M} ( \mathbf Y- \mathbf A \mathbf Z)^T+\sum^d_{l=1}\mathbf Z^T_l \bm \Sigma^{-1}_l  \mathbf Z_l }{2} \right). \nonumber \\
\propto& (\sigma^2_0)^{-\frac{k(n-q)}{2}} \prod^d\limits_{l=1} | \bm \Sigma_l|^{-\frac{1}{2}}  \etr \left(-\frac{\mathbf Y \mathbf {\tilde M}\mathbf Y^T  }{2} \right) \nonumber \\ 
   & \quad \quad   \times  \exp\left\{-\frac{ \mathbf Z^T_{vt} (\mathbf I_d \otimes \mathbf {\tilde M})  \mathbf Z_{vt} -2\mathbf Z^T_{vt} (\mathbf A^T \otimes \mathbf {\tilde M} ) \mathbf Y_{vt} + \mathbf Z^T_{vt} \bm \Sigma^{-1}_v \mathbf Z_{vt}  }{2} \right\}  
   \label{equ:dist_Y_Z}
\end{align}
where $\mathbf Z_{vt}= \mbox{vec}(\mathbf Z^T)$ and $\bm \Sigma_v $ is an $nd \times nd$ block diagonal matrix, where the $l$th diagonal block is $\bm \Sigma_l$, $l=1,...,d$. Marginalizing out $\mathbf Z$, one has 

\begin{align*}
&p( \mathbf Y \mid \mathbf A,  \sigma^2_0,  \bm \Sigma_{1:d} )\\
   \propto &(\sigma^2_0)^{-\frac{k(n-q)}{2}} \prod^d\limits_{l=1} | \tilde {\mathbf M}\bm \Sigma_l+\mathbf I_n |^{-\frac{1}{2}} \etr \left(-\frac{\mathbf Y \mathbf {\tilde M}\mathbf Y^T  }{2} \right)  \\ 
& \quad \quad  \times  \exp\left\{ -\frac{1}{2} \mathbf Y^T_{vt} (\mathbf A^T \otimes \tilde {\mathbf M}  )^T ( \mathbf I_d \otimes \mathbf {\tilde M}+ \bm \Sigma^{-1}_v   )^{-1}(\mathbf A^T \otimes \tilde {\mathbf M})  \mathbf Y_{vt} \right\} \\
\propto&(\sigma^2_0)^{-\frac{k(n-q)}{2}} \prod^d\limits_{l=1} | \tilde {\mathbf M}\bm \Sigma_l+\mathbf I_n |^{-\frac{1}{2}} \etr \left(-\frac{\mathbf Y \mathbf {\tilde M}\mathbf Y^T  }{2} \right)  \\
&\quad \quad  \times  \exp\left\{ -\frac{1 }{2} \mathbf Y^T_{vt} \left( \sum^d_{l=1} (\mathbf a_l \otimes \tilde {\mathbf M}) (\tilde{\mathbf M}+ \bm \Sigma^{-1}_l )^{-1} (\mathbf a^T_l \otimes \tilde {\mathbf M})\right)  \mathbf Y_{vt} \right\} \\
\propto&(\sigma^2_0)^{-\frac{k(n-q)}{2}} \prod^d\limits_{l=1} | \tilde {\mathbf M}\bm \Sigma_l+\mathbf I_n |^{-\frac{1}{2}} \etr \left(-\frac{\mathbf Y \mathbf {\tilde M}\mathbf Y^T  }{2} \right) \\
&\quad \quad   \times  \exp\left\{ -\frac{1 }{2} \mathbf Y^T_{vt} \left( \sum^d_{l=1} (\mathbf a_l \mathbf a^T_l)  \otimes\tilde {\mathbf M}  (\tilde{\mathbf M}+ \bm \Sigma^{-1}_l )^{-1} \tilde {\mathbf M}\right) \mathbf Y_{vt} \right\} \\
\propto &(\sigma^2_0)^{- (\frac{k(n-q)}{2}) } \prod^d\limits_{l=1} |  {\mathbf M}\tau_l \mathbf K_l+\mathbf I_n |^{-\frac{1}{2}} \exp\left\{ -\frac{ \tr(\mathbf Y  {\mathbf M} \mathbf Y^T)-\sum^d_{l=1} \mathbf a^T_l \mathbf Y  {\mathbf M}  ({\mathbf M}+ \tau^{-1}_l \bm K_l^{-1} )^{-1}  {\mathbf M} \mathbf Y^T \mathbf a_l}{2\sigma^2_0}\right\}
\end{align*}
Note  for any $l=1,...,d$
\begin{align*}
|  \tau_l{\mathbf M} \mathbf K_l+\mathbf I_n |=& |  \tau_l \mathbf K_l +\mathbf I_n | |\mathbf I_n -(\tau_l \mathbf K_l +\mathbf I_n)^{-1} \tau_l \mathbf K_l \mathbf H(\mathbf H^T \mathbf H)^{-1}\mathbf H^T |\\
=& |  \tau_l \mathbf K_l +\mathbf I_n | |\mathbf H^T \mathbf H|^{-1} | \mathbf H^T \mathbf H-\mathbf H^{T} ( (\tau_l \mathbf K_l)^{-1} +\mathbf I_n)^{-1} \mathbf H|\\
=& |  \tau_l \mathbf K_l +\mathbf I_n |  |\mathbf H^T \mathbf H|^{-1} |\mathbf H^T (\tau_l \mathbf K_l +\mathbf I_n)^{-1} \mathbf H |,
\end{align*}
where the first equation is by the definition of $\mathbf M$; the second equation is based on Fact \ref{item:determinant};  the third equation is by the Woodbury matrix identity. 
Further maximizing over $\sigma^2_0$ and we have the result. 

\end{proof}

The following  lemma is needed to prove Theorem \ref{thm:prediction_with_mean}. 
\begin{lemma}
Let $\tilde {\mathbf M}= \frac{1}{\sigma^2_0} (\mathbf I_n- \mathbf H(\mathbf H^T \mathbf H)^{-1} \mathbf H^T )$, where $\mathbf H$ is a $n\times q$ matrix with $n>q$, and $\mathbf H^T \mathbf H$ is a $q \times q$ matrix with rank $q$. Further let $\tilde {\bm \Sigma}= \bm \Sigma +\sigma^2_0 \mathbf I_n$, where both $\bm \Sigma$ and $\tilde {\bm \Sigma}$ have full rank. One has 
\begin{equation}
(\mathbf H^T \mathbf H)^{-1} \mathbf H^T ( \mathbf I_n - \bm \Sigma(\mathbf {\tilde M} \bm \Sigma+ \mathbf I_n )^{-1}  \mathbf {\tilde M} )=  (\mathbf H^T \tilde {\bm \Sigma}^{-1} \mathbf H)^{-1} \mathbf H^T \tilde {\bm \Sigma}^{-1}
\label{equ:identity_useful1}
\end{equation} 
\label{lemma:identity_useful1}
\end{lemma}
\begin{proof}
Denote $\bm \Sigma_0= \frac{1}{\sigma^2_0}\bm \Sigma$.  We start from the right hand side:
\begin{align*}
 &(\mathbf H^T \tilde {\bm \Sigma}^{-1} \mathbf H)^{-1} \mathbf H^T \tilde {\bm \Sigma}^{-1} \\
=&  \left( {\mathbf H^T \mathbf H}-  \mathbf H^T (\bm  \Sigma^{-1}_0 +\mathbf I_n)^{-1}\mathbf H  \right)^{-1} \mathbf H^T ( \bm \Sigma_0 + \mathbf I_n)^{-1} \\
=& \left\{ (\mathbf H^T \mathbf H)^{-1}- (\mathbf H^T \mathbf H)^{-1} \mathbf H^T ( \mathbf H(\mathbf H^T \mathbf H)^{-1} \mathbf H^T -\bm  \Sigma^{-1}_0 -\mathbf I_n )^{-1} \mathbf H (\mathbf H^T \mathbf H)^{-1} \right\}\mathbf H^T ( \bm \Sigma_0 + \mathbf I_n)^{-1} \\
=& (\mathbf H^T \mathbf H)^{-1}\mathbf H^T ( \bm \Sigma_0 + \mathbf I_n)^{-1} +(\mathbf H^T \mathbf H)^{-1} \mathbf H^T (\mathbf M+\bm \Sigma^{-1}_0)^{-1}\mathbf H (\mathbf H^T \mathbf H)^{-1} \mathbf H^T  ( \bm \Sigma_0 + \mathbf I_n)^{-1} \\
=& (\mathbf H^T \mathbf H)^{-1}\mathbf H^T \left\{ \mathbf I_n -(\bm \Sigma^{-1}_0 + \mathbf I_n)^{-1}+ (\mathbf M+\bm \Sigma^{-1}_0)^{-1}\mathbf H (\mathbf H^T \mathbf H)^{-1} \mathbf H^T  ( \bm \Sigma_0 + \mathbf I_n)^{-1} \right\}\\
=& (\mathbf H^T \mathbf H)^{-1}\mathbf H^T \left\{ \mathbf I_n- (\mathbf M+\bm \Sigma^{-1}_0)^{-1}(  (\mathbf M+\bm \Sigma^{-1}_0)(\bm \Sigma^{-1}_0 + \mathbf I_n)^{-1} -\mathbf H (\mathbf H^T \mathbf H)^{-1} \mathbf H^T  ( \bm \Sigma_0 + \mathbf I_n)^{-1})  \right\} \\
=& (\mathbf H^T \mathbf H)^{-1}\mathbf H^T \left\{ \mathbf I_n- (\mathbf M+\bm \Sigma^{-1}_0)^{-1}( \mathbf I_n - \mathbf H (\mathbf H^T \mathbf H)^{-1}\mathbf H^T )  \right\} \\
=&(\mathbf H^T \mathbf H)^{-1} \mathbf H^T ( \mathbf I_n - \bm \Sigma(\mathbf {\tilde M} \bm \Sigma+ \mathbf I_n )^{-1}  \mathbf {\tilde M} ),
\end{align*}
where we repeatedly use the Woodbury matrix identity. 
\end{proof}

\begin{proof}[Proof of Theorem \ref{thm:prediction_with_mean}]
Denote $\bm {\hat \Theta}=(\mathbf {\hat A},  \bm {\hat \gamma}, \bm \hat {\sigma}^2, { \hat \sigma}^2_0)$. From Equation (\ref{equ:dist_Y_Z}) in the proof of Lemma \ref{lemma:Y_mean_structure}, one has 
\begin{equation}
\mathbf Z_{vt} \mid \mathbf Y, \bm {\hat \Theta}  \sim \MN(\mathbf {\hat Z}_{vt}, \, \bm {\hat  \Sigma}_{\mathbf Z_{vt}} ), 
\label{equ:Z_vt_with_mean}
\end{equation}
where $\mathbf {\hat Z}_{vt}=\mbox{vec}(\hat {\bm \Sigma}_1 (  \mathbf M\hat{\bm \Sigma}_1 + \hat \sigma^2_0 \mathbf I_n )^{-1} \mathbf M \mathbf Y^T \mathbf {\hat a}_1,...,\hat {\bm \Sigma}_d ( \mathbf M \hat{\bm \Sigma}_d + \hat \sigma^2_0 \mathbf I_n )^{-1} \mathbf M \mathbf Y^T \mathbf {\hat a}_d )$ and $\bm {\hat  \Sigma}_{\mathbf Z_{vt}}$ is a $dn\times dn$ block diagonal matrix where the $l$th $n\times n$ diagonal block is $\hat \sigma^2_0 \hat {\bm \Sigma}_l (  \mathbf M\hat{\bm \Sigma}_l + \hat \sigma^2_0 \mathbf I_n )^{-1}$. 

It is also easy to obtain 
\begin{equation}
\mathbf B \mid \mathbf Y, \mathbf Z ,\sigma^2_0 \sim N((\mathbf H^T \mathbf H)^{-1} \mathbf H^T(\mathbf Y^T- \mathbf Z^T \mathbf A^T), \sigma^2_0 \mathbf I_k \otimes (\mathbf H^T \mathbf H)^{-1} ).
\label{equ:B_Y_Z}  
\end{equation}

Denote $z(\mathbf x^*)= (z_1(\mathbf x^*),...,z_d(\mathbf x^*) )^T$ the factors at input $\mathbf x^*$. First  the mean
\begin{align*}
\bm {\hat \mu}_M^*(\mathbf x^*)&=\E[\mathbf Y(\mathbf x^*) \mid \mathbf Y,  \bm {\hat \Theta} ]\\
&=\E[\E[ \mathbf Y(\mathbf x^*) \mid \mathbf Y, \mathbf B, \mathbf z(\mathbf x^*), \bm {\hat \Theta} ]]\\
&=\E[  (\mathbf h(\mathbf x^*) \mathbf B)^T+ \mathbf {\hat A} \mathbf z(\mathbf x^*)  \mid  \mathbf Y,  \bm {\hat \Theta}] \\
&=\E[\E[  (\mathbf h(\mathbf x^*) \mathbf B)^T+ \mathbf {\hat A} \mathbf z(\mathbf x^*)  \mid  \mathbf Y,  \bm {\hat \Theta} ,\mathbf Z]] \\ 
&= \E[ (\mathbf Y -\mathbf {\hat A} \mathbf { Z} ) \mathbf H(\mathbf H^T \mathbf H)^{-1}\mathbf h^T(\mathbf x^*)+ \mathbf {\hat A} \mathbf {\tilde z}(\mathbf x^*)\mid  \mathbf Y,  \bm {\hat \Theta}]\\
\end{align*}
where $\mathbf  {\tilde z}(\mathbf x^*)$ is a $d$-dimensional vector where the each term is $ \hat {\bm \Sigma}^T_l(\mathbf x^*) \hat{\bm \Sigma}^{-1}_l\mathbf Z^T_l$ for $l=1,...,d$. From (\ref{equ:Z_vt_with_mean}), noting $\mathbf Z_{vt}=\mbox{vec}(\mathbf Z^T)$, one has $\E[\mathbf Z \mid \mathbf Y, \bm {\hat \Theta}]= (\mathbf {\hat Z}^T_{1,M},...,\mathbf {\hat Z}^T_{d,M})^T$, with $\mathbf {\hat Z}_{l,M}=\mathbf a^T_l \mathbf Y \mathbf M (  \hat{\bm \Sigma}_l \mathbf M+ \hat \sigma^2_0 \mathbf I_n )^{-1} \hat {\bm \Sigma}_l  $ is a $1\times n$ vector, from which we have proved that equation (\ref{equ:hat_mu}) holds. 

\begin{align*}
\bm {\hat \Sigma}^*_M(\mathbf x^*)=&\V[\mathbf Y(\mathbf x^*) \mid \mathbf Y,  \bm {\hat \Theta} ] \\
=& \V[\E[ \mathbf Y(\mathbf x^*) \mid \mathbf Y, \mathbf B, \mathbf z(\mathbf x^*), \bm {\hat \Theta}]]+\E[\V[\mathbf Y(\mathbf x^*) \mid \mathbf Y, \mathbf B, \mathbf z(\mathbf x^*), \bm {\hat \Theta} ]]\\
=&\V[(\mathbf h(\mathbf x^*) \mathbf B)^T+ \mathbf {\hat A} \mathbf z(\mathbf x^*) \mid  \mathbf Y]+\sigma^2_0 \mathbf I_k \\
=&\V[ \E[(\mathbf h(\mathbf x^*) \mathbf B)^T+ \mathbf {\hat A} \mathbf z(\mathbf x^*) \mid  \mathbf Y, \mathbf Z ]]+ \E[ \V[(\mathbf h(\mathbf x^*) \mathbf B)^T+ \mathbf {\hat A} \mathbf z(\mathbf x^*) \mid  \mathbf Y, \mathbf Z ]]+ \sigma^2_0 \mathbf I_k \\
=& \V[(\mathbf Y -\mathbf {\hat A} \mathbf { Z} ) \mathbf H(\mathbf H^T \mathbf H)^{-1}\mathbf h^T(\mathbf x^*)+ \mathbf {\hat A} \mathbf {\tilde z}(\mathbf x^*) \mid  \mathbf Y]+\mathbf {\hat A}\V[ \mathbf z(\mathbf x^*)\mid \mathbf Y, \mathbf Z ]\mathbf {\hat A}^T  \\
& \quad  +\sigma^2_0 \mathbf I_k\otimes (1+\mathbf h^T(\mathbf x^*) (\mathbf H^T \mathbf H )^{-1} \mathbf h(\mathbf x^*)) \\
=& \mathbf {\hat A} \mathbf D_M(\mathbf x^*) \mathbf {\hat A} +  \sigma^2_0 \mathbf I_k\otimes (1+\mathbf h^T(\mathbf x^*) (\mathbf H^T \mathbf H )^{-1} \mathbf h(\mathbf x^*))
\end{align*}
where $\mathbf D_M(\mathbf x^*)$ is a diagonal matrix where the $l$th diagonal term is 
\begin{align}
D_{l,M}&=( \bm {\hat \Sigma}^T_l(\mathbf x^*) \bm {\hat \Sigma}^{-1}_l- \mathbf h(\mathbf x^*) (\mathbf H^T \mathbf H )^{-1} \mathbf H^T )  (  \mathbf { \tilde M} +\hat{\bm \Sigma}^{-1}_l )^{-1}(\bm {\hat \Sigma}^T_l(\mathbf x^*) \bm {\hat \Sigma}^{-1}_l- \mathbf h(\mathbf x^*) (\mathbf H^T \mathbf H )^{-1} \mathbf H^T )^T \nonumber \\
    &\quad \quad + ( \sigma^2_l\hat K_l(\mathbf x^*, \mathbf x^*)- \bm {\hat \Sigma}^T_l(\mathbf x^*) \bm {\hat \Sigma}^{-1}_l  \bm {\hat \Sigma}_l(\mathbf x^*)  )
    \end{align}
    
    We write $D_{l,M}+\sigma^2_0\mathbf h(\mathbf x^*)(\mathbf H^T \mathbf H)^{-1}\mathbf h^T(\mathbf x^*)$ as the following three terms. First, one has
    \begin{align}
    &\mathbf h(\mathbf x^*) (\mathbf H^T \mathbf H)^{-1} \mathbf H^T (\mathbf { \tilde M} +\hat{\bm \Sigma}^{-1}_l )^{-1}  \mathbf H ( \mathbf H^T \mathbf H)^{-1}\mathbf h^T(\mathbf x^*) +\sigma^2_0\mathbf h(\mathbf x^*)(\mathbf H^T \mathbf H)^{-1}\mathbf h^T(\mathbf x^*) \nonumber \\
    =& \sigma^2_0  \mathbf h(\mathbf x^*) \left\{(\mathbf H^T \mathbf H)^{-1}-  (\mathbf H^T \mathbf H)^{-1} \mathbf H^T \left(  \mathbf H(\mathbf H^T \mathbf H)^{-1} \mathbf H^T- \mathbf I_n - \hat \sigma^2_0 \hat{\bm \Sigma}^{-1}_l \right)^{-1}  \mathbf H ( \mathbf H^T \mathbf H)^{-1} \right\} \mathbf h^T(\mathbf x^*)  \nonumber  \\
    =& \sigma^2_0  \mathbf h(\mathbf x^*)  \left\{ \mathbf H^T \mathbf H- \mathbf H^T\left(\mathbf I_n +\hat \sigma^2_0 \hat{\bm \Sigma}^{-1}_l \right)^{-1} \mathbf H \right\}^{-1} \mathbf h^T(\mathbf x^*) \nonumber \\
    =&   \mathbf h(\mathbf x^*)  \left\{ \mathbf H^T \left( {\hat{\bm \Sigma}_l}+ {\hat \sigma^2_0} \mathbf I_n \right)^{-1}\mathbf H \right\}^{-1} \mathbf h^T(\mathbf x^*),  
    \label{equ:term1}
    \end{align}
    where the third and fourth equality is based on the Woodbury matrix identity. 
    
    Note
      \begin{align}
     &(\mathbf { \tilde M} +\hat{\bm \Sigma}^{-1}_l )^{-1}= \left(\frac{\mathbf I_n}{\hat \sigma^2_0} +\hat{\bm \Sigma}^{-1}_l -\frac{\mathbf H (\mathbf H^T \mathbf H )^{-1} \mathbf H^T}{\hat \sigma^2_0} \right)^{-1}\nonumber\\
      =& \left(\frac{\mathbf I_n}{\hat \sigma^2_0} +\hat{\bm \Sigma}^{-1}_l\right)^{-1} - \left(\frac{\mathbf I_n}{\hat \sigma^2_0} +\hat{\bm \Sigma}^{-1}_l\right)^{-1} \mathbf H \left\{ {\hat \sigma^2_0} \mathbf H^T \mathbf H - \mathbf H^T \left( \frac{\mathbf I_n}{\hat \sigma^2_0} + \hat{\bm \Sigma}^{-1}_l \right)^{-1} \mathbf H \right\}^{-1} \mathbf H^T   \left(\frac{\mathbf I_n}{\hat \sigma^2_0} +\hat{\bm \Sigma}^{-1}_l\right)^{-1} \nonumber \\
      =&\left(\frac{\mathbf I_n}{\hat \sigma^2_0} +\hat{\bm \Sigma}^{-1}_l\right)^{-1}  -\left({\mathbf I_n} + {\hat \sigma^2_0}\hat{\bm \Sigma}^{-1}_l\right)^{-1} \mathbf H \left\{\mathbf H^T \left(\hat{\bm \Sigma}_l+ {\hat \sigma^2_0}{\mathbf I_n} \right)^{-1}  \mathbf H \right\}^{-1} \mathbf H^T \left({\mathbf I_n} + {\hat \sigma^2_0}\hat{\bm \Sigma}^{-1}_l\right)^{-1}, \nonumber
        \end{align}        
      by Woodbury matrix identity,  one has 
              \begin{align}
              &( \bm {\hat \Sigma}^T_l(\mathbf x^*) \bm {\hat \Sigma}^{-1}_l  ) (\mathbf { \tilde M} +\hat{\bm \Sigma}^{-1}_l )^{-1}    ( \bm {\hat \Sigma}^T_l(\mathbf x^*) \bm {\hat \Sigma}^{-1}_l  )^T -\bm {\hat \Sigma}^T_l(\mathbf x^*) \bm {\hat \Sigma}^{-1}_l  \bm {\hat \Sigma}_l(\mathbf x^*) \nonumber \\
              =&- \bm {\hat \Sigma}^T_l(\mathbf x^*)     \bm {\tilde \Sigma}^{-1}_l   \bm {\hat \Sigma}_l(\mathbf x^*)  - \bm {\hat \Sigma}^T_l(\mathbf x^*)  \bm {\tilde \Sigma}^{-1}_l \mathbf H  ( \mathbf H^T \bm {\tilde \Sigma}^{-1}_l \mathbf H)^{-1}  \mathbf H^T \bm {\tilde \Sigma}^{-1}_l   \bm {\hat \Sigma}_l(\mathbf x^*).
                  \label{equ:term2}
              \end{align}
              
            Third, one has
            \begin{align}
            & \mathbf h(\mathbf x^*) (\mathbf H^T \mathbf H)^{-1} \mathbf H^T  (\mathbf { \tilde M} +\hat{\bm \Sigma}^{-1}_l )^{-1} \bm {\hat \Sigma}^{-1}_l  \bm {\hat \Sigma}_l(\mathbf x^*)  \nonumber \\ 
            =&\mathbf h(\mathbf x^*) (\mathbf H^T \mathbf H)^{-1} \mathbf H^T \left\{\mathbf I_n-\bm {\hat \Sigma}_l ( \mathbf { \tilde M} \bm {\hat \Sigma}_l+ \mathbf I_n )^{-1} \mathbf { \tilde M}  \right\}\bm {\hat \Sigma}_l(\mathbf x^*)   \nonumber \\ 
            =& \mathbf h(\mathbf x^*) (\mathbf H^T \tilde {\bm \Sigma}^{-1}_l \mathbf H)^{-1} \mathbf H^T \tilde {\bm \Sigma}^{-1}_l \bm {\hat \Sigma}_l(\mathbf x^*).
            \label{equ:term3}
             \end{align}
where the first equation is from the Woodbury matrix identity and the second equation is from Lemma \ref{lemma:identity_useful1}. 

From equation (\ref{equ:term1}), (\ref{equ:term2}) and (\ref{equ:term3}), we have shown that equation (\ref{equ:hat_Sigma}) holds.

        \end{proof}

 \section*{Appendix C: Simulated examples when models are misspecified}
We discuss two numerical examples where the latent factor model is misspecified. First, we let the Assumption \ref{assumption:A} be violated. In both examples, we assume that each entry of the factor loading matrix is sampled independently from a uniform distribution, hence not constrained in the Stiefel manifold. The second misspecification comes from the misuse of the kernel function in the factor processes. In reality, the smoothness of the true process may be unknown, therefore the use of any particular type of kernels may lead to an under-smoothing or over-smoothing scenario. Moreover, the factor may be an unknown deterministric function, rather than a sample from a Gaussian process. All these possible misspecifications will be discussed using the following Examples \ref{eg:misspeicifed_model_1} and  \ref{eg:misspeicifed_model_2}.

\begin{figure}[t]
\centering
  \begin{tabular}{cc}
    \includegraphics[width=.5\textwidth,height=.4\textwidth]{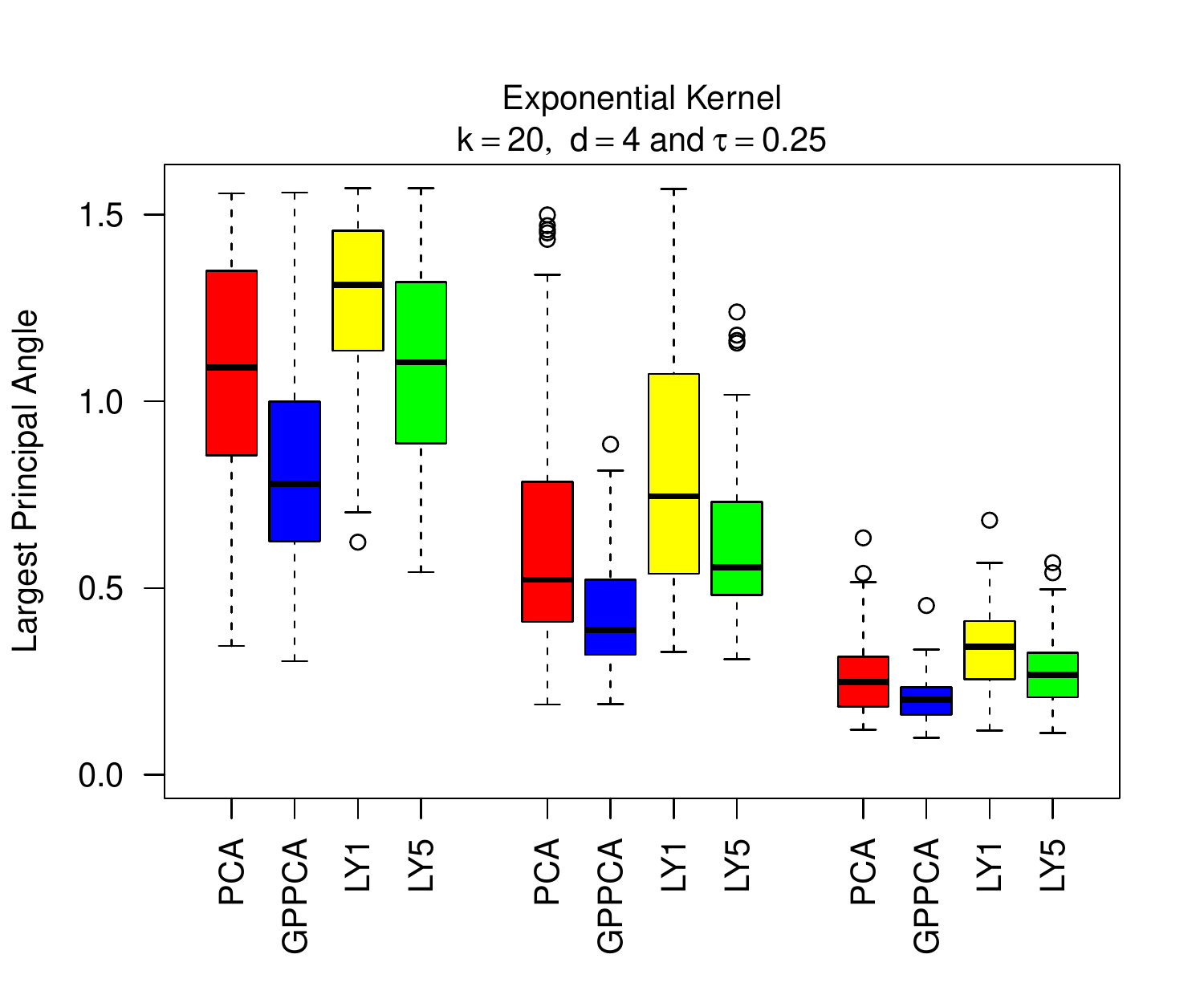}
        \includegraphics[width=.5\textwidth,height=.4\textwidth]{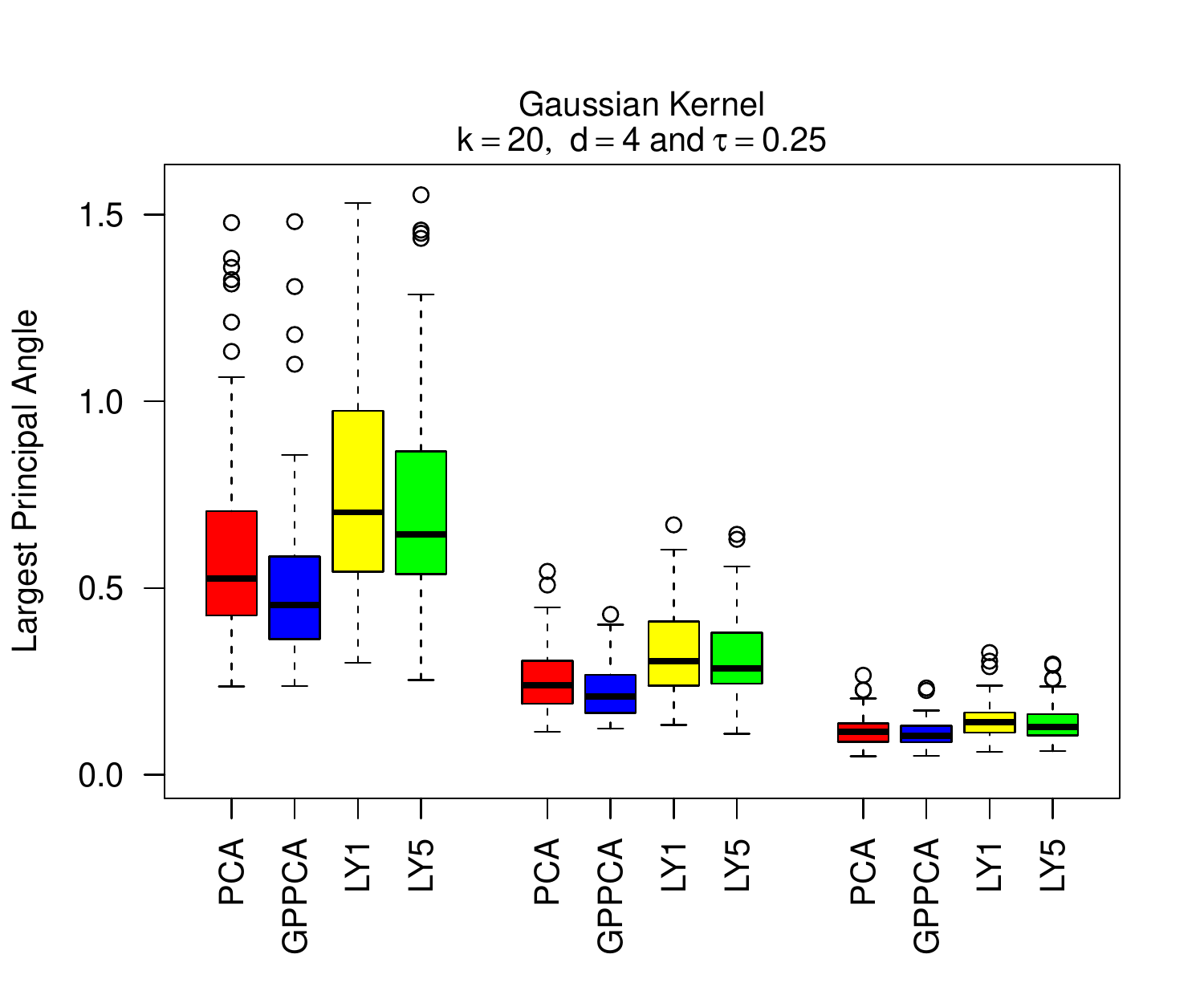}
  \end{tabular}
  \vspace{-.1in}
   \caption{The largest principal angle between the estimated subspace of four approaches and the true subspace for Example \ref{eg:misspeicifed_model_1}. The number of observations are assumed to be $n=100$, $n=200$ and $n=400$ for left 4 boxplots, middle 4 boxplots and right 4 boxplots in both panels, respectively. The kernel in simulating the data is assumed to be the exponential kernel in the left panel, whereas the kernel is assumed to be the Gaussian kernel in the right panel.}
\label{fig:simulation_misspecified_1_angles}
\end{figure}

\begin{example}[Unconstrained factor loadings and misspecified kernel functions]
 The data are sampled from model (\ref{equ:model_1}) with $\bm \Sigma_1=...=\bm \Sigma_d=\bm \Sigma$ and $x_i=i$ for $1\leq i\leq n$. Each entry of the factor loading matrix is assumed to be uniformly sampled from $[0,1]$ independently (without the orthogonal constraints in (\ref{equ:A})). The exponential kernel and the Guassian kernel are assumed in generating the data with different combinations of  $\sigma^2_0$ and $n$, while in the GPPCA, we still use the Mat{\'e}rn kernel function in (\ref{equ:matern_5_2}) for the estimation. We assume $k=20$, $d=4$, $\gamma=100$ and $\sigma^2=1$ in sampling the data. We repeat $N=100$ times for each scenario.  All the kernel parameters and the noise variance are treated as unknown and estimated from the data.
\label{eg:misspeicifed_model_1}
\end{example}

The largest principal angles between   $\mathcal M({\mathbf A})$ and $\mathcal M(\hat {\mathbf A})$ of the four approaches for Example \ref{eg:misspeicifed_model_1} are plotted in Figure  \ref{fig:simulation_misspecified_1_angles}. Even though the factor loading matrix is not constrained on the Stiefel manifold and the kernels are misspecified in GPPCA, GPPCA still has a better performance than other approaches in all scenarios. The PCA is an extreme case of the GPPCA where the range parameter of the kernel is estimated to be zero, meaning that the covariance of the factor process is an identity matrix. 

Another interesting finding is that all methods seem to perform better when the Gaussian kernel is used in simulating the data, even if the SNR of the simulation using a Gaussian kernel is smaller. This is because the variation of the factors is much larger when the Gaussian kernel is used, which makes the effect of the noise relatively small. In both cases, the GPPCA seems to be efficient in estimating the subspace of the factor loading matrix. 

Furthermore, since only the linear subspace of the factor loading matrix  is identifiable, rather than the factor loading matrix, the estimation of the factor loadings without the orthogonal constraints is also accurate by the GPPCA. Note the interpretation of the estimated variance parameter in the kernel by the GPPCA changes, because each column of $\mathbf A$ is not orthonomal in generating the data. 


\begin{table}[t]
\begin{center}
\begin{tabular}{lcccccc}
  \hline

         &  \multicolumn{3}{c}{exponential kernel and  $\tau=4$} & \multicolumn{3}{c}{Gaussian kernel and  $\tau=1/4$}\\
 &  $n=100$ & $n=200$ &  $n=400$ & $n=100$ & $n=200$ &  $n=400$  \\
  \hline
  PCA            &$7.4\times 10^{-2}$& $6.1\times 10^{-2}$ &$5.4\times 10^{-2}$& $1.1 \times 10^{0}$ &$8.9\times 10^{-1}$ &$8.4\times 10^{-1}$ \\
  GPPCA          &${\bf 3.1\times 10^{-2}}$ & $\bf 2.6\times 10^{-2}$ &$\bf 2.4\times 10^{-2}$ & $\bf 7.2\times 10^{-1}$ &$\bf 6.6\times 10^{-1}$ &$\bf 6.2\times 10^{-1}$  \\
    LY1 & $1.5\times 10^{-1}$ & $8.2 \times 10^{-1}$ &$5.7\times 10^{-2}$ &$1.3 \times 10^{0}$ &$1.0\times 10^{0}$ &$8.6\times 10^{-1}$ \\
  LY5 &  $1.3\times 10^{-1}$ & $7.3\times 10^{-1}$ & $5.6\times 10^{-2}$ & $1.3\times 10^{0}$ &$1.0 \times 10^{0}$ &$8.6\times 10^{-1}$ \\
  \hline

\end{tabular}
\end{center}
   \caption{AvgMSE for Example \ref{eg:misspeicifed_model_1}.}

   \label{tab:misspeicifed_model_1}
\end{table}

The AvgMSE of the four approaches for Example \ref{eg:misspeicifed_model_1} is shown in Table \ref{tab:misspeicifed_model_1}. The estimation of the GPPCA is more accurate than the other approaches. Because of the larger variation in the factor processes with the Gaussian kernel, the corresponding variation in the mean of the output is also larger than the one when the exponential kernel is used. Consequently, all approaches have larger estimated errors for the case with the Gaussian kernel.

We show an example when the factor is generated from a deterministic function.

\begin{example}[Unconstrained factor loadings and deterministic factors]
The data are sampled from model (\ref{equ:model_1}) with each latent factor being a deterministic function
\[Z_l(x_i)=cos(0.05\pi \theta_l x_i)\]
where $\theta_l \overset{i.i.d.}{\sim} \mbox{unif}(0,1)$ for $l=1,...,d$, with $x_i=i$ for $1\leq i\leq n$, $\sigma^2_0=0.25$, $k=20$ and $d=4$. Four cases are tested with the sample size $n=100$, $n=200$, $n=400$ and $n=800$. 
\label{eg:misspeicifed_model_2}
\end{example}

\begin{figure}[t]
\centering
  \begin{tabular}{c}
    \includegraphics[width=1\textwidth,height=.4\textwidth]{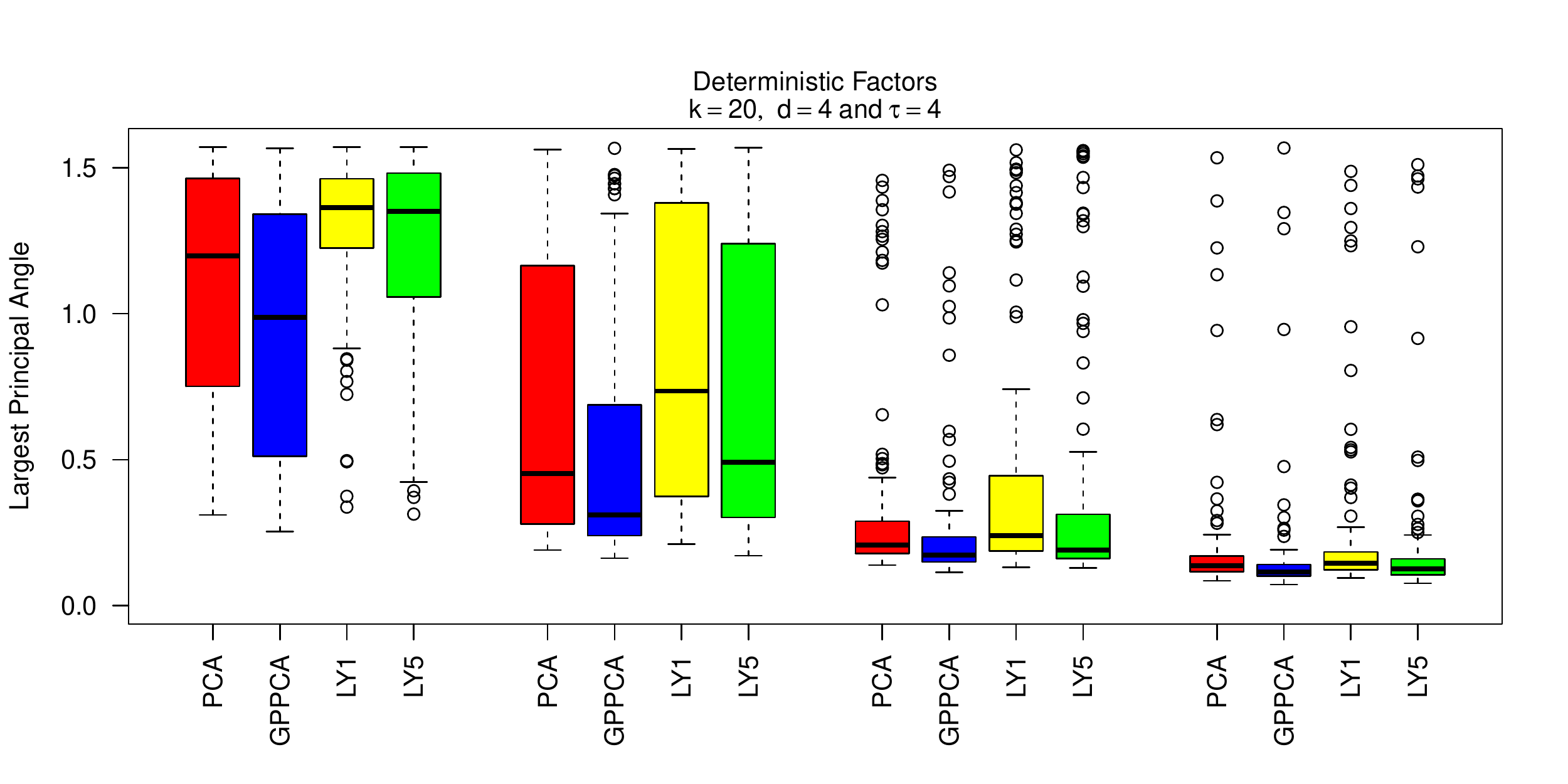} 
  \end{tabular}
  \vspace{-.1in}
   \caption{The largest principal angle between the estimated subspace of the loading matrix and the true subspace for Example \ref{eg:misspeicifed_model_2}. From the left to the right, the number of observations is assumed to be $n=100$, $n=200$, $n=400$ and  $n=800$ for each 4 boxplots, respectively.  }
\label{fig:simulation_misspecified_deterministic_angles}
\end{figure}

\begin{table}[t]
\begin{center}
\begin{tabular}{lcccccc}
  \hline

   &  $n=100$ & $n=200$ &  $n=400$ & $n=800$   \\
  \hline
  PCA            &$7.0\times 10^{-2}$& $6.0\times 10^{-2}$ &$5.4\times 10^{-2}$& $5.2\times 10^{-2}$  \\
  GPPCA          &$\bf 1.4\times 10^{-2}$ & $\bf 9.2\times 10^{-3}$ &$\bf 6.7\times 10^{-3}$ & $\bf 5.5\times 10^{-3}$ \\
    LY1 & $9.8\times 10^{-1}$ & $7.6 \times 10^{-1}$ &$6.3\times 10^{-2}$ &$5.7\times 10^{-2}$ \\
  LY5 &  $9.3\times 10^{-2}$ & $7.3\times 10^{-2}$ & $6.2\times 10^{-2}$&$5.6\times 10^{-2}$\\
  Ind GP &  $2.0\times 10^{-2}$ & $1.9\times 10^{-2}$ & $1.7\times 10^{-2}$&$1.7\times 10^{-2}$\\
  PP GP &  $2.0\times 10^{-2}$ & $1.9\times 10^{-2}$ & $1.8\times 10^{-2}$&$1.8\times 10^{-2}$\\

  \hline
\end{tabular}
\end{center}
   \caption{AvgMSE for Example \ref{eg:misspeicifed_model_2}.}

   \label{tab:misspeicifed_model_2}
\end{table}

For the GPPCA, we assume the covariance is shared for each factor and the Mat{\'e}rn kernel in (\ref{equ:matern_5_2}) is used for  Example \ref{eg:misspeicifed_model_2}. The largest principal angle between $\mathcal M(\mathbf A)$ and $\mathcal M(\mathbf {\hat A})$ of the four approaches is given in Figure \ref{fig:simulation_misspecified_deterministic_angles}.  When the number of observations increases, all four methods estimate $\mathcal M(\mathbf A)$  more accurately, even though the factors are no longer sampled from Gaussian processes. Note the  reproducing kernel Hilbert space attached to the Gaussian process with the Mat{\'e}rn kernel contains those functions in the Sobolev space that are squared integrable up to the order $2$ \citep{gu2018theoretical}, while the deterministic function to generate the factors in Example \ref{eg:misspeicifed_model_2} is infinitely integrable. The GPPCA is the most precise in estimating $\mathcal M(\mathbf A)$ among the four approaches in this scenario. 



The AvgMSE of the different approaches in estimating the mean of the output of the Example \ref{eg:misspeicifed_model_2} is given in Table \ref{tab:misspeicifed_model_2}.  We also include two more approaches, namely the independent Gaussian processes (Ind GP) and the parallel partial Gaussian processes (PP GP). The Ind GP approach treats each output variable independently and the mean of the output is estimated by the predictive mean in the Gaussian process regression \citep{rasmussen2006gaussian}. The PP GP approach also models each output variable independently by a Gaussian process, whereas the covariance fucntion is shared for $k$ independent Gaussian processes and estimated based on all data \citep{Gu2016PPGaSP}.

As shown in Table \ref{tab:misspeicifed_model_2}, the estimation by the GPPCA is the most accurate among six approaches. The estimation by the Ind GP and PP GP perform similarly and they seem to perform better than the estimation by the PCA, LY1 and LY5. One interesting finding in Table 4 is that the  AvgMSE by the GPPCA seems to decrease faster than those of the Ind GP and PP GP, when the sample size increases. This numerical result may shed some lights on the convergence rate of the GPPCA in the nonparametric regression problem. 

 \section*{Appendix D:  Model fitting for the gridded temperature data}
 
 For the GPPCA, we consider the model \eqref{equ:model_with_mean}, where the input $x$ is an integer time point ranging from 1 to 240. The mean function is assumed as $\mathbf h(x)=(1, x)$ to model the trend of the temperature anomalies over time. For the case with estimated variance, the parameters are estimated by maximizing the marginal likelihood in (\ref{equ:marginal_with_mean}) using the matrix of temperature anomalies $\mathbf Y$ with $k=1639$ rows and $n=220$ columns. The marginal likelihood with known variance is derived by plugging the variance value, instead of integrating it out with a prior. We use Equation (\ref{equ:cond_pred_with_mean}) to compute the predictive distribution by the GPPCA, where $ \mathbf Y_1(x^*)$ is a $439\times 20$ matrix of the temperature anomalies at the $439$ spatial locations (which has the entire observations over the whole 240 months). The  $1200\times 20$ matrix  $\mathbf Y_2(x^*)$ is the held-out temperature anomalies for testing.
 
 
 For the PPCA, we first subtract the mean of each location of the $1639\times 220$ output matrix of normalized temperature anomalies. We then estimate the factor loading matrix by Equation (7) in \cite{tipping1999probabilistic}. The predictive distribution of the test output by the PPCA was obtained in a similar fashion as the Equation \eqref{equ:cond_pred_with_mean} in the GPPCA and the empirical mean for the test output was added back for comparison. The PPCA does not incorporate the temporal correlation and linear trend in the model.
 
The temporal model is constructed by a GaSP separately for each test location.  The Mat{\'e}rn kernel in (\ref{equ:matern_5_2}) and the linear trend $\mathbf h(x)=(1, x)$ are assumed for the temporal model. The spatial model uses a GaSP with a constant mean separately for each test month.  The product kernel in  (\ref{equ:K_l}) is assumed for the two-dimensional input (latitude and longitude) and the Mat{\'e}rn kernel in (\ref{equ:matern_5_2}) is used for each subkernel. The range and the nugget parameters in the temporal model and spatial model are estimated using the ${\tt RobustGaSP}$ ${\sf R}$ package  (\cite{gu2018robustgasp}), and the predictions are also obtained by this package.  
 
 
 The temporal regression by the random forest are trained separately for each location. For each test month, the $439$ observations of that month are used as the responses and  the  $439 \times 220$ output on the other months for the same locations are used as the covariates. The regression parameters of this temporal regression capture the temporal dependence of the output between the test month and the training months. The $1200\times 200$ matrix of the temperature anomalies at the test locations and observed time points are used as the test input. The spatial regression by the random forest uses $220$ observations of a test location as responses and the $220\times 439$ matrix of the temperature anomalies of the observed locations are used as the input. The $20\times 439$ matrix of the temperature anomalies at the observed locations and test time points are used as the test input. The ${\tt randomForest}$  ${\sf R}$ package \citep{liaw2002classification} is used for training models and compute predictions.

 The spatio-temporal model assumes a 3 dimensional product kernel in (\ref{equ:K_l})  for both time points and locations, and the Mat{\'e}rn kernel in (\ref{equ:matern_5_2}) is used as the subkernel for each input variables. Note that if we use the whole training output, the computational order of inverting the covariance matrix is $O(N^3)$, where $N=369360$ is the total number of inputs, which is  computationally challenging. When the output can be written as an $n_1\times n_2$ matrix, the likelihood corresponds to a matrix normal distribution, where two kernel functions model the correlation between rows and between columns of the output. The computational order of the matrix normal distribution is the maximum of $O(n_1^3)$ and $O(n_2^3)$. We choose the $439\times 240$ output of temperature anomalies at the locations with the whole observations to estimate the parameters. The constant mean is assumed for each location. The MLE is used for estimating the range parameters in kernel, nugget, mean and variance parameters. After plugging in the parameters, the predictive distribution of the test data is used for predictions. Though only $439\times 240$ observations are used for estimating the parameters due to the computational conveinience, all $369360$ training output is used for computing the predictive distribution of the test output.

 
 



\vskip 0.2in
\bibliography{References_2018}

\begin{thebibliography}{44}
\providecommand{\natexlab}[1]{#1}
\providecommand{\url}[1]{\texttt{#1}}
\expandafter\ifx\csname urlstyle\endcsname\relax
  \providecommand{\doi}[1]{doi: #1}\else
  \providecommand{\doi}{doi: \begingroup \urlstyle{rm}\Url}\fi

\bibitem[Absil et~al.(2006)Absil, Edelman, and Koev]{absil2006largest}
P.-A. Absil, Alan Edelman, and Plamen Koev.
\newblock On the largest principal angle between random subspaces.
\newblock \emph{Linear Algebra and its applications}, 414\penalty0
  (1):\penalty0 288--294, 2006.

\bibitem[Alvarez et~al.(2012)Alvarez, Rosasco, and
  Lawrence]{alvarez2011kernels}
Mauricio~A Alvarez, Lorenzo Rosasco, and Neil~D. Lawrence.
\newblock Kernels for vector-valued functions: A review.
\newblock \emph{Foundations and Trends{\textregistered} in Machine Learning},
  4\penalty0 (3):\penalty0 195--266, 2012.

\bibitem[Bai(2003)]{bai2003inferential}
Jushan Bai.
\newblock Inferential theory for factor models of large dimensions.
\newblock \emph{Econometrica}, 71\penalty0 (1):\penalty0 135--171, 2003.

\bibitem[Bai and Ng(2002)]{bai2002determining}
Jushan Bai and Serena Ng.
\newblock Determining the number of factors in approximate factor models.
\newblock \emph{Econometrica}, 70\penalty0 (1):\penalty0 191--221, 2002.

\bibitem[Bayarri et~al.(2009)Bayarri, Berger, Calder, Dalbey, Lunagomez, Patra,
  Pitman, Spiller, and Wolpert]{Bayarri09}
M.~J. Bayarri, James~O. Berger, Eliza~S. Calder, Keith Dalbey, Simon Lunagomez,
  Abani~K. Patra, E.~Bruce Pitman, Elaine~T. Spiller, and Robert~L. Wolpert.
\newblock Using statistical and computer models to quantify volcanic hazards.
\newblock \emph{Technometrics}, 51:\penalty0 402--413, 2009.

\bibitem[Berger et~al.(2001)Berger, De~Oliveira, and
  Sans{\'o}]{berger2001objective}
James~O Berger, Victor De~Oliveira, and Bruno Sans{\'o}.
\newblock Objective bayesian analysis of spatially correlated data.
\newblock \emph{Journal of the American Statistical Association}, 96\penalty0
  (456):\penalty0 1361--1374, 2001.

\bibitem[Berger et~al.(2009)Berger, Bernardo, and Sun]{berger2009formal}
James~O Berger, Jos{\'e}~M Bernardo, and Dongchu Sun.
\newblock The formal definition of reference priors.
\newblock \emph{The Annals of Statistics}, 37\penalty0 (2):\penalty0 905--938,
  2009.

\bibitem[Bj{\"o}rck and Golub(1973)]{bjorck1973numerical}
A.~Bj{\"o}rck and Gene~H Golub.
\newblock Numerical methods for computing angles between linear subspaces.
\newblock \emph{Mathematics of computation}, 27\penalty0 (123):\penalty0
  579--594, 1973.

\bibitem[Breiman(2001)]{breiman2001random}
Leo Breiman.
\newblock {R}andom forests.
\newblock \emph{Machine learning}, 45\penalty0 (1):\penalty0 5--32, 2001.

\bibitem[Conti and O'Hagan(2010)]{conti2010bayesian}
Stefano Conti and Anthony O'Hagan.
\newblock Bayesian emulation of complex multi-output and dynamic computer
  models.
\newblock \emph{Journal of statistical planning and inference}, 140\penalty0
  (3):\penalty0 640--651, 2010.

\bibitem[Farah et~al.(2014)Farah, Birrell, Conti, and
  Angelis]{farah2014bayesian}
Marian Farah, Paul Birrell, Stefano Conti, and Daniela~De Angelis.
\newblock Bayesian emulation and calibration of a dynamic epidemic model for
  {A/H1N1} influenza.
\newblock \emph{Journal of the American Statistical Association}, 109\penalty0
  (508):\penalty0 1398--1411, 2014.

\bibitem[Fricker et~al.(2013)Fricker, Oakley, and
  Urban]{fricker2013multivariate}
Thomas~E Fricker, Jeremy~E Oakley, and Nathan~M Urban.
\newblock Multivariate {G}aussian process emulators with nonseparable
  covariance structures.
\newblock \emph{Technometrics}, 55\penalty0 (1):\penalty0 47--56, 2013.

\bibitem[Gelfand et~al.(2004)Gelfand, Schmidt, Banerjee, and
  Sirmans]{gelfand2004nonstationary}
Alan~E Gelfand, Alexandra~M Schmidt, Sudipto Banerjee, and C.~F. Sirmans.
\newblock Nonstationary multivariate process modeling through spatially varying
  coregionalization.
\newblock \emph{Test}, 13\penalty0 (2):\penalty0 263--312, 2004.

\bibitem[Gelfand et~al.(2010)Gelfand, Diggle, Guttorp, and
  Fuentes]{gelfand2010handbook}
Alan~E Gelfand, Peter Diggle, Peter Guttorp, and Montserrat Fuentes.
\newblock \emph{Handbook of spatial statistics}.
\newblock CRC Press, 2010.

\bibitem[Gu(2019)]{gu2018jointly}
Mengyang Gu.
\newblock Jointly robust prior for {G}aussian stochastic process in emulation,
  calibration and variable selection.
\newblock \emph{Bayesian Analysis}, 14\penalty0 (3):\penalty0 857--885, 2019.

\bibitem[Gu and Berger(2016)]{Gu2016PPGaSP}
Mengyang Gu and James~O Berger.
\newblock Parallel partial {G}aussian process emulation for computer models
  with massive output.
\newblock \emph{Annals of Applied Statistics}, 10\penalty0 (3):\penalty0
  1317--1347, 2016.

\bibitem[Gu et~al.(2018{\natexlab{a}})Gu, Wang, and Berger]{Gu2018robustness}
Mengyang Gu, Xiaojing Wang, and James~O Berger.
\newblock Robust {G}aussian stochastic process emulation.
\newblock \emph{Annals of Statistics}, 46\penalty0 (6A):\penalty0 3038--3066,
  2018{\natexlab{a}}.

\bibitem[Gu et~al.(2018{\natexlab{b}})Gu, Xie, and Wang]{gu2018theoretical}
Mengyang Gu, Fangzheng Xie, and Long Wang.
\newblock A theoretical framework of the scaled {G}aussian stochastic process
  in prediction and calibration.
\newblock \emph{arXiv preprint arXiv:1807.03829}, 2018{\natexlab{b}}.

\bibitem[Gu et~al.(2019)Gu, Palomo, and Berger]{gu2018robustgasp}
Mengyang Gu, Jes{\'u}s Palomo, and James~O Berger.
\newblock Robust{G}a{SP}: Robust {G}aussian stochastic process emulation in
  {R}.
\newblock \emph{The R Journal}, 11\penalty0 (1), June 2019.

\bibitem[Hartikainen and Sarkka(2010)]{hartikainen2010kalman}
Jouni Hartikainen and Simo Sarkka.
\newblock Kalman filtering and smoothing solutions to temporal gaussian process
  regression models.
\newblock In \emph{Machine Learning for Signal Processing (MLSP), 2010 IEEE
  International Workshop for Signal Processing}, pages 379--384. IEEE, 2010.

\bibitem[Higdon et~al.(2008)Higdon, Gattiker, Williams, and
  Rightley]{higdon2008computer}
Dave Higdon, James Gattiker, Brian Williams, and Maria Rightley.
\newblock Computer model calibration using high-dimensional output.
\newblock \emph{Journal of the American Statistical Association}, 103\penalty0
  (482):\penalty0 570--583, 2008.

\bibitem[Hoff(2013)]{hoff2013bayesian}
Peter~D Hoff.
\newblock Bayesian analysis of matrix data with rstiefel.
\newblock \emph{arXiv preprint arXiv:1304.3673}, 2013.

\bibitem[Hoffmann(2007)]{hoffmann2007kernel}
Heiko Hoffmann.
\newblock Kernel {PCA} for novelty detection.
\newblock \emph{Pattern recognition}, 40\penalty0 (3):\penalty0 863--874, 2007.

\bibitem[Jolliffe(2011)]{jolliffe2011principal}
Ian Jolliffe.
\newblock \emph{Principal component analysis}.
\newblock Springer, 2011.

\bibitem[Kokiopoulou et~al.(2011)Kokiopoulou, Chen, and
  Saad]{kokiopoulou2011trace}
Effrosini Kokiopoulou, Jie Chen, and Yousef Saad.
\newblock Trace optimization and eigenproblems in dimension reduction methods.
\newblock \emph{Numerical Linear Algebra with Applications}, 18\penalty0
  (3):\penalty0 565--602, 2011.

\bibitem[Lam and Yao(2012)]{lam2012factor}
Clifford Lam and Qiwei Yao.
\newblock Factor modeling for high-dimensional time series: inference for the
  number of factors.
\newblock \emph{The Annals of Statistics}, 40\penalty0 (2):\penalty0 694--726,
  2012.

\bibitem[Lam et~al.(2011)Lam, Yao, and Bathia]{lam2011estimation}
Clifford Lam, Qiwei Yao, and Neil Bathia.
\newblock Estimation of latent factors for high-dimensional time series.
\newblock \emph{Biometrika}, 98\penalty0 (4):\penalty0 901--918, 2011.

\bibitem[Liaw and Wiener(2002)]{liaw2002classification}
Andy Liaw and Matthew Wiener.
\newblock Classification and regression by randomforest.
\newblock \emph{R news}, 2\penalty0 (3):\penalty0 18--22, 2002.

\bibitem[Liu and West(2009)]{liu2009dynamic}
Fei Liu and Mike West.
\newblock A dynamic modelling strategy for {B}ayesian computer model emulation.
\newblock \emph{Bayesian Analysis}, 4\penalty0 (2):\penalty0 393--411, 2009.

\bibitem[Mika et~al.(1999)Mika, Sch{\"o}lkopf, Smola, M{\"u}ller, Scholz, and
  R{\"a}tsch]{mika1999kernel}
Sebastian Mika, Bernhard Sch{\"o}lkopf, Alex~J Smola, Klaus-Robert M{\"u}ller,
  Matthias Scholz, and Gunnar R{\"a}tsch.
\newblock Kernel {PCA} and de-noising in feature spaces.
\newblock In \emph{Advances in neural information processing systems}, pages
  536--542, 1999.

\bibitem[Nocedal(1980)]{nocedal1980updating}
Jorge Nocedal.
\newblock Updating quasi-newton matrices with limited storage.
\newblock \emph{Mathematics of computation}, 35\penalty0 (151):\penalty0
  773--782, 1980.

\bibitem[Oakley(1999)]{oakley1999bayesian}
Jeremy Oakley.
\newblock \emph{Bayesian uncertainty analysis for complex computer codes}.
\newblock PhD thesis, University of Sheffield, 1999.

\bibitem[Overstall and Woods(2016)]{overstall2016multivariate}
Antony~M Overstall and David~C Woods.
\newblock Multivariate emulation of computer simulators: model selection and
  diagnostics with application to a humanitarian relief model.
\newblock \emph{Journal of the Royal Statistical Society: Series C (Applied
  Statistics)}, 65\penalty0 (4):\penalty0 483--505, 2016.

\bibitem[Paulo et~al.(2012)Paulo, Garc{\'\i}a-Donato, and
  Palomo]{paulo2012calibration}
Rui Paulo, Gonzalo Garc{\'\i}a-Donato, and Jes{\'u}s Palomo.
\newblock Calibration of computer models with multivariate output.
\newblock \emph{Computational Statistics and Data Analysis}, 56\penalty0
  (12):\penalty0 3959--3974, 2012.

\bibitem[Rasmussen(2006)]{rasmussen2006gaussian}
Carl~Edward Rasmussen.
\newblock \emph{Gaussian {P}rocesses for {M}achine {L}earning}.
\newblock MIT Press, 2006.

\bibitem[Saad(1992)]{saad1992numerical}
Youcef Saad.
\newblock \emph{Numerical {M}ethods for {L}arge {E}igenvalue {P}roblems}.
\newblock Manchester University Press, 1992.

\bibitem[Sacks et~al.(1989)Sacks, Welch, Mitchell, and Wynn]{sacks1989design}
Jerome Sacks, William~J Welch, Toby~J Mitchell, and Henry~P Wynn.
\newblock Design and analysis of computer experiments.
\newblock \emph{Statistical Science}, 4\penalty0 (4):\penalty0 409--423, 1989.

\bibitem[Sch{\"o}lkopf et~al.(1998)Sch{\"o}lkopf, Smola, and
  M{\"u}ller]{scholkopf1998nonlinear}
Bernhard Sch{\"o}lkopf, Alexander Smola, and Klaus-Robert M{\"u}ller.
\newblock Nonlinear component analysis as a kernel eigenvalue problem.
\newblock \emph{Neural Computation}, 10\penalty0 (5):\penalty0 1299--1319,
  1998.

\bibitem[Seeger et~al.(2005)Seeger, Teh, and Jordan]{seeger2005semiparametric}
Matthias Seeger, Yee-Whye Teh, and Michael Jordan.
\newblock Semiparametric latent factor models.
\newblock Technical report, 2005.

\bibitem[Shen(2017)]{SShen2017}
Samuel~S.P. Shen.
\newblock \emph{R programming for climate data analysis and visualization:
  computing and plotting for NOAA data applications}.
\newblock San Diego State University, San Diego, USA., 2017.

\bibitem[Taylor and Lane(2004)]{taylor2004development}
B~Taylor and A~Lane.
\newblock Development of a novel family of military campaign simulation models.
\newblock \emph{Journal of the Operational Research Society}, 55\penalty0
  (4):\penalty0 333--339, 2004.

\bibitem[Tipping and Bishop(1999)]{tipping1999probabilistic}
Michael~E Tipping and Christopher~M Bishop.
\newblock Probabilistic principal component analysis.
\newblock \emph{Journal of the Royal Statistical Society: Series B (Statistical
  Methodology)}, 61\penalty0 (3):\penalty0 611--622, 1999.

\bibitem[Wen and Yin(2013)]{wen2013feasible}
Zaiwen Wen and Wotao Yin.
\newblock A feasible method for optimization with orthogonality constraints.
\newblock \emph{Mathematical Programming}, 142\penalty0 (1-2):\penalty0
  397--434, 2013.

\bibitem[West(2003)]{west2003bayes7}
M.~West.
\newblock Bayesian factor regression models in the \lq\lq large p, small n"
  paradigm.
\newblock In J.~M. Bernardo, M.~J. Bayarri, J.~O. Berger, A.~P. David,
  D.~Heckerman, A.~F.~M. Smith, and M.~West, editors, \emph{Bayesian Statistics
  7}, pages 723--732. Oxford University Press, 2003.
\newblock URL \url{http://ftp.isds.duke.edu/WorkingPapers/02-12.html}.

\end{thebibliography}

\end{document}